\documentclass[nocopyright]{sigplanconf}[8pt]

\usepackage{graphicx} % embedded graphics
\usepackage{mathpartir} % proofs and inference rules
\usepackage{amsmath,amssymb,amsthm} % more mathematical symbols

\usepackage{relsize}
\usepackage{enumerate, xspace}
\usepackage{color}
\usepackage{enumitem, url}

\newcommand{\cut}[1]{}

\newcommand{\lang}{{{System~M}}\xspace}
\newcommand{\lss}{$\mbox{LS}^2$}
\newcommand{\fstar}{$\mbox{F}^\star$}

\newcommand{\Paragraph}[1]{\vspace{5pt}\noindent{\bf #1}}
\newcommand{\rulename}[1]{{\sc #1}}

\newcommand{\nf}{{\ensuremath{{\tt nf}}}}
\newcommand{\fv}{\ensuremath{{\tt fv}}}
\newcommand{\ftv}{\ensuremath{{\tt ftv}}}
\newcommand{\fa}{\ensuremath{{\tt fa}}}

\newcommand{\dom}{\ensuremath{{\tt dom}}}

\newcommand{\clete}{\ensuremath{{\tt lete}}}
\newcommand{\cletc}{\ensuremath{{\tt letc}}}

\newcommand{\ecomp}{\ensuremath{{\tt comp}}}
\newcommand{\cret}{\ensuremath{{\tt ret}}}
\newcommand{\cfix}{\ensuremath{{\tt fix}}}
\newcommand{\cif}{\ensuremath{{\tt if}}}
\newcommand{\cthen}{\ensuremath{{\tt then}}}
\newcommand{\celse}{\ensuremath{{\tt else}}}
\newcommand{\cact}{\ensuremath{{\tt act}}}

\newcommand{\kact}{\ensuremath{{\tt Act}}}

\newcommand{\tcomp}{\ensuremath{{\tt comp}}}
\newcommand{\tcmp}{\ensuremath{{\tt cmp}}}

\newcommand{\tptr}{\ensuremath{{\tt ptr}}}

\newcommand{\tunit}{\ensuremath{{\tt unit}}}
\newcommand{\tbool}{\ensuremath{{\tt bool}}}
\newcommand{\texp}{\ensuremath{{\tt any}}}
\newcommand{\tmsg}{\ensuremath{{\tt msg}}}
\newcommand{\btrue}{\ensuremath{{\tt tt}}}
\newcommand{\bfalse}{\ensuremath{{\tt ff}}}
\newcommand{\const}{\ensuremath{{\mathit bv}}}
\newcommand{\bt}{\ensuremath{{\tt b}}}

\newcommand{\jok}{\ensuremath{\mathrel{\sf ok}}}

\newcommand{\jtrue}{\ensuremath{\mathrel{\sf true}}}
\newcommand{\jonempty}{\ensuremath{\mathrel{\sf silent}}}

\newcommand{\conj}{\mathrel{\wedge}}
\newcommand{\disj}{\mathrel{\vee}}
\newcommand{\imp}{\mathrel{\Rightarrow}}
\newcommand{\pred}[1]{\textsf{#1}\xspace}

\newcommand{\lsem}{[\![}
\newcommand{\rsem}{]\!]}

\newcommand{\trace}{\mathcal{T}}
\newcommand{\defeq}{\,\mathrel{\stackrel{\textit{def}}{=}}\,}

\newcommand{\rg}[2]{\mathcal{RG}\lsem{#1}\rsem_{#2}}

\newcommand{\rd}[2]{\mathcal{RT}\lsem{#1}\rsem_{#2}}

\newcommand{\rvi}[3]{{\mathcal{RV}}(#3)\lsem{#1}\rsem_{#2}}
\newcommand{\rei}[3]{{\mathcal{RE}}(#3)\lsem{#1}\rsem_{#2}}
\newcommand{\rci}[3]{{\mathcal{RC}}(#3)\lsem#1\rsem_{#2}}
\newcommand{\rai}[3]{{\mathcal{RA}}(#3)\lsem{#1}\rsem_{#2}}
\newcommand{\rfi}[3]{{\mathcal{RF}}(#3)\lsem#1\rsem_{#2}}
\newcommand{\rgi}[3]{{\mathcal{RG}}(#3)\lsem{#1}\rsem_{#2}}

\newcommand{\reinv}[2]{\mathcal{RE}_\textit{INV}\lsem{#1}\rsem_{#2}}
\newcommand{\rcinv}[2]{\mathcal{RC}_\textit{INV}\lsem{#1}\rsem_{#2}}
\newcommand{\rvinv}[2]{\mathcal{RV}_\textit{INV}\lsem{#1}\rsem_{#2}}
\newcommand{\rfinv}[2]{\mathcal{RF}_\textit{INV}\lsem{#1}\rsem_{#2}}

\newcommand{\interp}[1]{\lsem{#1}\rsem}

\newcommand{\steps}[1]{\xrightarrow{#1}}
\newcommand{\stepsone}{\hookrightarrow}
\newcommand{\betaone}{\hookrightarrow_\beta}

\newcommand{\rt}{\mathrel{\triangleright}}

\newcommand{\next}{\ensuremath{{\sf next}}}
\newcommand{\stuck}{{\sf stuck}}

\newcommand{\nat}{\mathbb{N}}

\newcommand{\cset}{\mathsf{C}}

\newcommand{\ubertype}{\mathsf{Type}}
\newcommand{\pset}[1]{\mathcal{P}(#1)}

\newcommand{\tpair}[3]{#1.(#2, #3)}

\newcommand{\cnfn}[2]{\textit{confine}~(#1)~(#2)}

\theoremstyle{plain}
\newtheorem{thm}{Theorem}
\newtheorem{lem}[thm]{Lemma}

\theoremstyle{definition}

\newcommand{\ee}{\ensuremath{\mathcal{E}}}

 \newcommand{\reada}{\textsf{\small read}\xspace}
 \newcommand{\writea}{\textsf{\small write}\xspace}

\newcounter{linenum}

\newenvironment{packedenumerate}{\vspace{-3pt}
\begin{enumerate}
\setlength{\topsep}{0pt}
\setlength{\itemsep}{0pt}
\setlength{\partopsep}{0pt}
}{\end{enumerate}}

\newenvironment{packeditemize}{\vspace{-3pt}
\begin{itemize}
\setlength{\topsep}{0pt}
\setlength{\itemsep}{0pt}
\setlength{\partopsep}{0pt}
}{\end{itemize}}

\newenvironment{proofsketch}{\vspace{-5pt}\noindent\emph{Proof (sketch):}\hspace*{0.25em}}{ \hspace*{\fill} \qed}

\newcommand{\notes}[1]{}
\newcommand{\dg}[1]{\notes{Deepak says: #1}}

\newcommand{\limin}[1]{\notes{Limin says: #1}}

\newcommand{\eat}[1]{}
\newcommand{\todo}[1]{}
\long\def\symbolfootnotetext[#1]#2{\begingroup%
\def\thefootnote{\fnsymbol{footnote}}\footnotetext[#1]{#2}\endgroup}

\long\def\symbolfootnotemark[#1]{\begingroup%
\def\thefootnote{\fnsymbol{footnote}}\footnotemark[#1]\endgroup}

\newcommand{\action}[1]{\mathtt{#1}}
\newcommand{\at}{@}
\newcommand{\jon}{\circ}
\newcommand{\shash}{s\_hash}
\newcommand{\valpcr}{\pred{val\_pcr}}
\newcommand{\valNV}{\pred{val\_NV}}
\newcommand{\hprefix}{\pred{hash\_prefix}}
\newcommand{\llpcr}{\mathtt{pcr17}}

\newcommand{\hchain}{\mathtt{hash\_chain}}
\newcommand{\chash}{\mathtt{code\_hash}}
\newcommand{\sinit}{\pred{service\_init}}
\newcommand{\skey}{skey}

\newcounter{linen}
\newcommand{\linestart}{\setcounter{linen}{1}{\scriptstyle\arabic{linen}}\quad}
\newcommand{\newl}{\stepcounter{linen}\\{\scriptstyle\arabic{linen}}\quad}
\newcommand{\stlet}[1]{\mathtt{let}~#1 = }

\newcommand{\cbnd}[2]{#1 \leftarrow #2;}

\newcommand{\mi}[1]{\mathit{#1}}
\newcommand{\runmodule}{\mi{runmodule}}
\newcommand{\srvc}{\mi{srvc}}
\newcommand{\service}{\mi{service}}
\newcommand{\Nloc}{\mi{Nloc}}
\newcommand{\state}{\mi{state}}
\newcommand{\newstate}{\mi{state}'}
\newcommand{\stry}{\mathsf{service\_try}}
\newcommand{\sexec}{\mathsf{service\_invoke}}

\begin{document}
\pagestyle{plain}
\title{\lang: A Program Logic for Code Sandboxing and Identification}

\authorinfo{Limin Jia}
{ECE \& INI\\
Carnegie Mellon University}
{liminjia@cmu.edu}
\authorinfo{Shayak Sen}
{CS\\
Carnegie Mellon University}
{shayaks@cs.cmu.edu}
\authorinfo{Deepak Garg}
{Max Planck Institute for\\
Software Systems}
{dg@mpi-sws.org}
\authorinfo{Anupam Datta}
{CS \& ECE \\
Carnegie Mellon University}
{danupam@cmu.edu}

\maketitle

\begin{abstract}
Security-sensitive applications that execute untrusted code often
check the code's integrity by comparing its syntax to a known good
value or sandbox the code to contain its effects. \lang is a new
program logic for reasoning about such security-sensitive
applications. \lang extends Hoare Type Theory (HTT) to trace safety
properties and, additionally, contains two new reasoning
principles. First, its type system internalizes logical equality,
facilitating reasoning about applications that check code
integrity. Second, a confinement rule assigns an effect type to a
computation based solely on knowledge of the computation's sandbox. We
prove the soundness of \lang relative to a step-indexed trace-based
semantic model. We illustrate both new reasoning principles of \lang
by verifying the main integrity property of the design of Memoir, a
previously proposed trusted computing system for ensuring state
continuity of isolated security-sensitive applications.
\end{abstract}

%\keywords Program logic, traces, adversarial code, security
\section{Introduction}
\label{sec:intro}

\dg{High-level comment. The \rulename{Eq} rule is now conceptually
  trivial. It will help to explain what difficulty it causes in the
  semantics. This can be done in the intro or later.}

Software systems, such as Web browsers, smartphone platforms, and
extensible operating systems and hypervisors, are designed to provide
subtle security properties in the presence of adversaries who can
supply code, which is then executed with the privileges of the trusted
system.  For example, webpages routinely execute third-party
JavaScript with full access to their content; 
smartphones execute apps from open app stores, often with very lax
sandboxes; operating system kernels include untrusted (and often
buggy) device drivers; and trusted computing platforms load programs
from disk and only later verify loaded programs using the Trusted Platform
Module (TPM)~\cite{TPM2.0specs}. Despite executing potentially adversarial code,
all these systems have security-related goals, often \emph{safety
properties} over traces~\cite{lamport77:safety}. For example, a hypervisor must
ensure that an untrusted guest operating system running on top of it cannot
modify the hypervisor's page table, a webpage must ensure that an
embedded untrusted advertisement cannot access a user's password, and
trusted computing mechanisms must enable a remote party to check that an
expected software stack was loaded in the expected order on an
untrusted server.

Secure execution of untrusted code in trusted contexts rely on two
common mechanisms. First, untrusted code is often run inside a
\emph{sandbox} that confines its interaction with key system resources
to a restricted set of interfaces. This practice is seen in
Web browsers, hypervisors, and other security-critical systems. Second,
\emph{code identification} mechanisms are used to infer that
an untrusted piece of code is in fact syntactically equal to a known
piece of code. These mechanisms include distribution of signed
code, and trusted computing mechanisms~\cite{TPM2.0specs} that leverage
hardware support to enable remote parties to check the identity of
code on an untrusted computer.
Motivated by these systems, we present a program logic, called
{\lang}, for modeling and proving safety properties of systems that
securely execute adversary-supplied code via sandboxing and code
identification.

\lang's design is inspired by Hoare Type Theory
(HTT)~\cite{nanevski:jfp08,nanevski:esop07,nanevski:tldi09}. Like HTT,
a monad separates computations with side-effects from pure
expressions, and a monadic type both specifies the return type of a
computation and includes a postcondition that specifies the
computation's side-effects.
The postcondition of a computation type in {\lang} uses predicates
over the entire trace of the computation. This is motivated by our
desire to verify safety properties~\cite{lamport77:safety}, which are,
by definition, predicates on traces. Further, the postcondition
contains not one but two predicates on traces. One predicate, the
standard \emph{partial correctness assertion}, holds if the
computation completes. The other, called the \emph{invariant
assertion}, holds at all intermediate points of the computation, even
if the computation is stuck or divergent. The invariant assertion is
directly used to represent safety properties.

To this basic infrastructure, we add two novel reasoning principles
that internalize the rationale behind commonly used mechanisms for
ensuring secure execution of adversary-supplied code: code
identification and sandboxing.  These rules derive effects of
untyped code
potentially provided by an adversary and, hence, enable the typing
derivation of the trusted code to include as sub-derivations, the
reasoning of effects of the adversarial code.
 \dg{I don't think that our rules allow
 embedding of the typing derivation of the adversarial code. The whole
 point is to not type the adversarial code.}

The first principle, a rule called \rulename{Eq}, ascribes the type of
a program to another program $e'$: if $e$ is
syntactically equal to $e'$ and $e : \tau$, then $e': \tau$.  This
rule is useful for typing programs read from
adversary-modifiable memory locations when separate reasoning can
establish that the value stored in the location is, in fact, syntactically equal to
some known expression with a known type.
Depending on the application,
such reasoning may be based in a dynamic check (e.g., in secure
boot~\cite{tcsurvey} the hash of a textual reification of a program
read from adversary-accessible memory is compared to the corresponding
hash of a known program before executing the read program) or it
may be based in a logical proof showing the inability of the adversary
to write the location in question (e.g., showing that guests cannot
write to hypervisor memory).

Our second reasoning principle, manifest in a rule called
\rulename{Confine}, allows us to type partially specified
adversary-supplied code from knowledge of the \emph{sandbox} in which
the code will execute.  The intuition behind this
rule is that if all side-effecting interfaces available to a
computation maintain a certain invariant on the shared state, then
that computation cannot violate that invariant, irrespective of its
actual code. 
The \rulename{Confine} rule generalizes prior work of Garg \emph{et al.}
on reasoning about interface-confined adversarial code in a
first-order language~\cite{garg10:ls2}.
The main difference from Garg \emph{et al.} \cite{garg10:ls2} is that in this paper
trusted interfaces can receive and execute code, in addition to data,
from the adversary and other trusted components. Our use of the
\rulename{Confine} rule stresses our view that assumptions made about
adversarial code should be minimized. In contrast, a lot of work,
e.g., proof-carrying code~\cite{necula97:pcc}, requires that
adversarial code be checked in a rich type system prior to execution,
which eliminates the need for a rule like \rulename{Confine}.
Section~\ref{sec:example} explains intuitions behind these two
principles in more detail.

We show soundness of \lang relative to a step-indexed
model~\cite{ahmed:esop06} built over syntactic traces. As in some
prior
work~\cite{datta03:composition,datta07:entcs,dfgk09:ls2,garg10:ls2},
our semantics of assertions and postconditions account for
interleaving actions from concurrently executing programs including
adversarial programs and, hence, our soundness theorem implies that
all verified properties hold in the presence of adversaries, which is
a variant of robust safety, proposed by Gordon et
al.~\cite{Gordon:2003:ATS}.
\lang supports \emph{compositional
proofs}---security proofs of sequentially composed programs are built
from proofs of their sub-programs. \lang also admits concurrent
composition---properties proved of a program hold when that program
executes concurrently with other, even adversarial, programs.

{\lang} is the first program logic that allows
proofs of safety for programs that \emph{execute} adversary-supplied
code with adequate precautions, but does not force the adversarial
code to be completely available for typing. Other frameworks like
Bhargavan \emph{et al}'s contextual theorems~\cite{bhargavan10:rcf}
for F7 achieve expressiveness similar to the \rulename{Confine} rule
for a slightly limited selection of trace properties. (We compare to
related work in Section~\ref{sec:related}.)
Our step-indexed model of
Hoare types is also novel; although our exclusion of preconditions,
our use of call-by-name $\beta$-reduction,
and the inclusion of adversary-supplied code make the model nonstandard.

\lang can be used to model and verify protocols as well as system designs.  We
demonstrate the reasoning principles of \lang by verifying the state
continuity property of the design of Memoir~\cite{memoir}, a
previously proposed trusted computing system. For reasons of space, we
elide proofs, some technical details and several typing rules from
this paper. These are presented in the accompanying technical
appendix.

\section{Term Language and Operational Semantics}
\label{sec:language}
We summarize \lang's term syntax in Figure~\ref{fig:exp-syntax}.  Pure
expressions, denoted $e$, are distinguished from effectful
computations, denoted $c$.  An expression can be a variable, a
constant, a function, a polymorphic function, a function application,
a polymorphic function instantiation, or a suspended
computation. Constants can be Booleans ($\btrue, \bfalse$), natural
numbers ($n \in \mathcal{N}$), thread identifiers ($\iota \in
\mathcal{I}$), and memory locations ($\ell \in \mathcal{L}$). We use
$\cdot$ as the place holder for the type in a polymorphic function
instantiation.  Suspended computations $\ecomp(c)$ constitute a monad
with return $\cret(e)$ and bind $\clete(e_1,x.c_2)$.

{\lang} is parametrized over a set of action symbols $A$, which are
instantiated with concrete actions based on specific application
domains. For instance, $A$ may be instantiated with memory operations
such as \pred{read} and \pred{write}. An action, denoted $a$, is the
application of an action symbol $A$ to expression arguments.

A basic computation is either an atomic action ($\cact(a)$) or
$\cret(e)$ that returns the pure expression $e$ immediately.
$\cfix\,f(x).c$ is a fixpoint operator. $f$, which represents a
suspended fixpoint computation, may appear free in the body
$c$. Computation $(c\; e)$ is the application of a fixpoint
computation to its argument. $\cletc(c_1, x.c_2)$ denotes the
sequential composition of $c_1$ and $c_2$, while $\clete(e_1, x.c_2)$
is the sequential composition of the suspended computation to which
$e_1$ reduces and $c_2$. In both cases, the expression returned by the
first computation is bound to $x$, which may occur free in~$c_2$. We
sometimes use the alternate syntax $\cbnd{x}{c_1}c_2$ and
$\stlet{x}e_1;c_2$. When the expression returned by the first
computation is not used $c_2$, we write $c_1;c_2$ and $e_1;c_2$.

\begin{figure}
\centering
\(%\small
\begin{array}{lclll}
 \textit{Base values} & \const & ::= & \btrue ~|~ \bfalse ~|~ \iota ~|~
 \ell ~|~ n 
\\ 
% %
\textit{Expressions} & e & ::=  & x ~|~ \const~|~
\lambda x.e  ~|~ \Lambda X.e 
\\ & & ~|~ & e_1\;e_2  ~|~e\; \cdot 
~|~ \ecomp(c) \\
\textit{Actions} & a & ::= & A ~|~ a\; e ~|~ a\; \cdot \\
\textit{Computations} & c & ::= &  
 \cact(a) ~|~  \cret(e) ~|~ \cfix\,f(x).c ~|~
c\; e 
\\ & & ~|~&  \cletc(c_1, x.c_2)  ~|~ \clete(e_1, x. c_2)
\\ & & ~|~&  c_1;c_2  ~|~ e_1;c_2
~|~ 
\cif\ e\ \cthen\ c_1\ \celse\ c_2
\end{array}
\)
\caption{Term Syntax}
\label{fig:exp-syntax}
\end{figure}

%\paragraph{Operational Semantics} 
The operational semantics of {\lang} are small-step and based on
interleaving of concurrent threads.

\vspace{5pt}
\(\begin{array}{llll}
%\textit{Frame} & F & ::= & x.c \\ 
%
\textit{Stack} & K & ::= & [] ~|~ x.c :: K \\
\textit{Thread} & T & ::= & \langle \iota; K; c \rangle ~|~  
\langle \iota; K; e \rangle ~|~   \langle \iota;\stuck \rangle\\
\textit{Configuration} & C & ::= & \sigma \rt T_1, \ldots, T_n
\end{array}\)
\vspace{5pt}

A thread $T$ is a unit of sequential execution. A non-stuck thread is
a triple $\langle\iota; K; c\rangle$ or $\langle\iota; K; e\rangle$,
where $\iota$ is a unique identifier of that thread (drawn from a set
$\mathcal{I}$ of such identifiers), $K$ is the execution
(continuation) stack, and $c$ and $e$ are the computation and expression
currently being evaluated. A thread permanently enters a stuck state,
denoted $\langle \iota;\stuck \rangle$, after performing an illegal
action, such as accessing an unallocated memory location.
An execution stack is a list of frames of the form $x.c$ recording the
return points of sequencing statements in the enclosing context.  In a
frame $x.c$, $x$ binds the return expression of the computation
preceding $c$.  A configuration of the system is a shared state
$\sigma$ and a set of all threads. $\sigma$ is application-specific;
for the rest of this paper, we assume that it is a standard heap
mapping pointers to expressions, but this choice is not essential. For
example, in modeling network protocols, the heap could be replaced by
the set of undelivered (pending) messages on the network.

For pure expressions, we use call-by-name $\beta$-reduction
$\rightarrow_\beta$. This choice simplifies the operational semantics
and the soundness proofs, as explained in Sections~\ref{sec:discussion}. We elide the standard rules for
$\rightarrow_\beta$.
The small-step transitions for threads and system configurations are
shown in Figure~\ref{fig:semantics:reduction-config}.  The relation
$\sigma \rt T \stepsone \sigma' \rt T'$ defines a small-step
transition of a single thread. ${\cal C} \steps{} {\cal C'}$ denotes a
small-step transition for configuration ${\cal C}$; it results from
the reduction of any single thread in ${\cal C}$.
%, which is defined in terms of $\stepsone$.

\begin{figure}[tb]
%\small
\noindent\framebox{$\sigma \rt T \stepsone \sigma' \rt T'$}
\begin{center}
\(%
%\vspace{3pt}
 \inferrule*[right=R-ActS]{
\\\next(\sigma, a) = (\sigma', e) 
\\ e\neq\stuck
 }{ \sigma \rt \langle \iota;
   x.c ::K; \cact(a) \rangle \stepsone \sigma' \rt \langle \iota;
   K;  c[e/x]  \rangle}
\) 

\vspace{7pt} 

\(
 \inferrule*[right=R-ActF]{
\\\next(\sigma, a) = (\sigma', \stuck) 
 }{ \sigma \rt \langle \iota;
   x.c ::K; \cact(a) \rangle \stepsone \sigma' \rt \langle \iota;\stuck \rangle}
\) 

\vspace{7pt} 

\(
 \inferrule*[right=R-Stuck]{  }{ \sigma \rt \langle \iota;\stuck \rangle 
\stepsone \sigma \rt \langle \iota;\stuck \rangle}
\) 

\vspace{7pt} 

\(
\inferrule*[right=R-Ret]{
 }{ \sigma \rt \langle \iota;  x.c :: K; \cret(e) \rangle
  \stepsone  \sigma \rt \langle \iota; K; c[e/x]
  \rangle}
\)

\vspace{7pt} 

\(
\inferrule*[right=R-SeqE2]{e \rightarrow_\beta e'
}{\sigma \rt \langle \iota;
  K; e \rangle \betaone  \sigma \rt \langle
  \iota; K; e' \rangle}
\)

\vspace{7pt} 

\(
\inferrule*[right=R-SeqE3]{
}{\sigma \rt \langle \iota;
  x.c_2:: K; \ecomp(c_1) \rangle \stepsone  \sigma \rt \langle
  \iota; x.c_2:: K; c_1 \rangle}
\)

\vspace{7pt} 

\(
 \inferrule*[right=R-Fix]{ 
}{\sigma \rt \langle \iota; K; (\cfix f(x).c)\; e 
   \rangle\\ \stepsone \sigma \rt \langle \iota; K; c[\lambda z.\ecomp(\cfix(f(x).c)\; z)/f][e/x] \rangle}
\)
\end{center}
\caption{Selected small-step reduction semantics of configurations}
\label{fig:semantics:reduction-config}
\end{figure}

The rules for $\sigma \rt T \stepsone \sigma' \rt T'$ are mostly
straightforward.  The rules for evaluating an atomic action
(\rulename{R-ActS} and \rulename{R-ActF}) rely on a function
$\mathsf{next}$ that takes the current store $\sigma$ and an
action $a$, and returns a new store and an
expression, which are the result of the action. If the action is
illegal, then $\mathsf{next}(\sigma, a) = (\sigma', \stuck)$. If the
action returns a non-stuck expression $e$ (rule \rulename{R-ActS}),
then the top frame ($x.c$) is popped off the stack, and $c[e/x]$
becomes the current computation of the thread.  If $\mathsf{next}$
returns $\stuck$ (rule \rulename{R-ActF}), then the thread enters the
stuck state and permanently remains there.
When a sequencing statement $\clete(e_1, x.c_2)$ is evaluated, the
frame $x.c_2$ is pushed onto the stack, and $e_1$ is first reduced to
a suspended computation $\ecomp(c_1)$; then $c_1$ is evaluated.
When a fixpoint $(\cfix f(x).c); e$ is evaluated, $f$ is substituted
with a function whose body is a suspension of $\cfix f(x).c$.

Any \emph{finite} execution of a configuration results in a trace
$\trace$, defined as a finite sequence of reductions. With each
reduction we associate a time point $u$, also called a (logical)
\emph{time point}. These time points on the trace are monotonically
increasing. A trace annotated with time is
written $\steps{u_0}C_0 \steps{u_1} C_1 \ldots \steps{u_n} C_n$, where
$u_i\leq u_{i+1}$.  We follow the convention that the reduction from
$C_{i}$ to $C_{i+1}$ happens at time $u_{i+1}$ and that its effects
occur immediately. Thus the state at time $u_i$ is the state in
$C_{i}$.

\section{Motivating Application}
\label{sec:example}

We briefly review Memoir~\cite{memoir}, our main application, and
highlight the challenges in analyzing Memoir to motivate the novel
typing rules for deriving properties of adverary-supplied code using
code identification and sandboxing.

\subsection{Overview of Memoir}
\label{ssec:challenges}

Memoir provides state-integrity guarantees for stateful security-sensitive
services invoked by potentially malicious parties. Such services often rely on
untrusted storage to store their persistent state. An example of such a
service is a password manager that responds with a stored
password when it receives a request containing a URL and a username.
The service would want to ensure secrecy and integrity of its state;
in this case, the set of stored passwords.
Simply encrypting and signing the service's state cannot
prevent the attacker from invoking the service with a valid but old
state, and consequently mounting service rollback attacks. For the password
manager service, this attack could cause the service to respond with old
(possibly compromised) passwords. Memoir solves this problem by using the
TPM to provide state integrity guarantees.
Memoir relies on the following TPM features: 

\begin{packeditemize} 
\item \emph{Platform configuration registers} (PCRs) contain
20-byte hashes known as \emph{measurements} that summarize the current
configuration of the system. The value they contain can only be updated
in two ways: (1) a \emph{reset} operation which sets the value of the PCR to a fixed
default value; (2) an \emph{extend} operation which takes as argument a value $v$
and updates the value of the PCR to the hash of the concatenation of its current value with $v$.

\item \emph{Late launch} is a command that can be used to securely load a
program.  It extends the hash of the textual reification of the program into a
special PCR (PCR17). Combined with the guarantees provided by a PCR, late launch
provides a mechanism for precise code identification.

\item \emph{Non-volatile RAM} (NVRAM) provides persistent storage that
allows access control based on PCR measurements. Specifically, permissions on
NVRAM locations can be tied to a PCR $p$ and value $v$ such that the location
can only be read when the value contained in $p$ is $v$.
\end{packeditemize}

Memoir has two phases: service initialization and service
invokation.
During initialization, the Memoir module is assigned an NVRAM
block. It is also given a
service to protect. The module generates a new symmetric key that is used
throughout the lifetime of the service.  It sets the permissions on accesses to
the NVRAM block to be tied to the hash stored in PCR 17, which contains the
hash of the code for Memoir and the service. To prevent rollback attacks, it
uses a \emph{freshness tag} which is a chain of hashes of all the requests
received so far. The secret key and an initial freshness tag are stored in the
designated NVRAM location. The service then runs for the first time to generate
an initial state, which along with the freshness tag is encrypted with the
secret key and stored to disk. This encryption of the service's state along
with the freshness tag is called a \emph{snapshot}.
\begin{figure}
\(\begin{array}{l}
\linestart
\mi{runmodule} (\mi{srvc}, \mi{snap},  \mi{req}, \mi{Nloc})  =\newl
~~~\cdots\newl
~~~\cbnd{(\skey, \mi{freshness\_tag})}{\cact(\action{NVRAMread}~\mi{Nloc})} \newl
~~~\cbnd{service\_state}{\mi{check\_decrypt\_snapshot}~(\mi{snap})}\newl
~~~\cdots\newl
~~~\cbnd{(\mi{state}', \mi{resp})\\
~~~~~~~~~~~~~~~}{(\mi{srvc}~\mi{ExtendPCR}~\mi{ResetPCR}~\cdots)~(\mi{state}, \mi{req})}\newl
~~~\cdots
\end{array}\)
\caption{Snippet of invokation code}
\label{fig:snippet}
\end{figure}

After initialization, a service can be invoked by providing Memoir with an
NVRAM block, a piece of service code, and a snapshot. In
Figure~\ref{fig:snippet}, we show a snippet of the Memoir service invokation
code.  Memoir retrieves the key and freshness tag  from the NVRAM.  Memoir then
decrypts the snapshot and verifies that the freshness tag in the provided state
matches the one stored in NVRAM.  If the verification succeeds, Memoir computes
a new freshness tag and updates the NVRAM.  Next, it executes the service to
generate a new state and a response. The new snapshot corresponding to the new
state and freshness tag is stored to disk.

The security property we prove about Memoir is that the service can
only be invoked on the state generated by the last completed instance
of the service.
The proof of security for Memoir requires reasoning about the effects
the service, which is provided by potentially malicious parties.

To derive properties of the $\mi{runmodule}$ code shown above one needs to
assign a type to $\mi{srvc}$, which is provided by an adversary. The service
$\mi{srvc}$, run on line 6, is a function that contains no free actions.
However, $\mi{srvc}$ takes as arguments interface functions corresponding to
every atomic action in our model.  Shown above are $\mi{ExtendPCR}$ and
$\mi{ResetPCR}$ which are simply wrappers for the corresponding atomic actions.

For example, the proof requires deriving the following two invariant properties
about $\mi{srvc}$:

\begin{packedenumerate}
\item It does not change the value of the PCR to a state that allows the adversary to later read the NVRAM.
\item It does not leak the secret key.
\end{packedenumerate}
\vspace{-5pt}
The first invariant is derived using the fact that the service is confined
to the interface exposed by the TPM. The second invariant is
derived in three steps: (i) prove that $\mi{srvc}$ is syntactically equal
to the initial service; (ii) assume that the initial service does not
leak the secret key; and (iii) hence infer that $\mi{srvc}$ does not
leak the secret key. We next describe \lang's typing rules that enable
such reasoning.

\subsection{Typing Adversary Supplied Code}

\paragraph{Reasoning about effects of confinement}
In analyzing %security properties of 
programs that execute
adversary-supplied code, one often encounters a partially trusted
program, whose code is \emph{unknown}, but which is \emph{known} or
assumed to be confined to the use of a specific set of interfaces to
perform actions on shared state. In our Memoir example, every program
on the machine is confined to the interface provided by the TPM. Using
just this confinement information, we can sometimes deduce a useful
effect-type for the partially trusted program. Suppose $c$ is a closed
computation, which syntactically does not contain any actions and can
invoke as subprocedures the computations $c_1, \ldots, c_n$ only
(i.e., $c$ is \emph{confined} to $c_1,\ldots,c_n$). If all actions
performed by $c_1, \ldots, c_n$ satisfy a predicate $\varphi$, then
the actions performed by $c$ must also satisfy~$\varphi$, irrespective
of the code of~$c$. Hence, we can statically specify the
effects of $c$, without knowing its code, but knowing the effects of
$c_1,\ldots,c_n$.

We formalize this intuition in a typing rule called
\rulename{Confine}.
To explain this rule, we introduce some notation. Let
$\tau$ denote types in {\lang} that include postconditions for
computations and, specifically, let $\tcmp(\tau, \varphi)$ denote the
monadic type of computations that return a value of type $\tau$ and
whose actions satisfy the predicate $\varphi$. (The notation
$\tcmp(\tau, \varphi)$ is simpler than our actual computation types,
but it suffices for the explanation here.)

As an illustration of our \rulename{Confine} rule, consider any closed
expression $e$. Assume that $e$ does
not contain any primitive actions. Then, we claim that for any
$\varphi$, $e$ has the type $\tcmp(\tbool, \varphi)
\rightarrow \tcmp(\tbool, \varphi)$. To understand this claim, assume
that $\varphi$ is the property ``the action is not a
\writea to memory''. To show
that $e: \tcmp(\tbool, \varphi) \rightarrow \tcmp(\tbool, \varphi)$,
we must show that for any $v: \tcmp(\tbool, \varphi)$, $e\, v:
\tcmp(\tbool, \varphi)$. Hence, we must show that the actions
performed by the computation, say $c$, that $e\, v$ evaluates to do
not include \writea. This can be argued easily: Because $e$ is closed
and does not contain any actions, the only way this computation $c$
could \writea is by invoking the computation $v$. However, because $v:
\tcmp(\tbool, \varphi)$, $v$ does not \writea. Hence, $e \, v:
\tcmp(\tbool, \varphi)$.

In fact, we can assign $e$ any type, including higher-order function
types, as long as the effects in that type are $\varphi$.  Let the
predicate $\cnfn{\tau} {\varphi}$ mean that $\varphi=\varphi'$ for all
nested types of the form $\tcomp(\tau',\varphi')$ in $\tau$.
Let $\cnfn{\Gamma}{\varphi}$ mean that every type $\tau$ that $\Gamma$ maps
to satisfies  $\cnfn{\tau} {\varphi}$.
Let
$\fa(e)=\emptyset$ mean that $e$ syntactically does not contain any
actions.
Then, the idea of typing
through confinement is captured by the following
rule. The rule says that for any $e$ without any
actions, if $\tau$'s nested effects are $\varphi$, and the types
of the free variables in $e$ also only have $\varphi$ as effects, then
 $e: \tau$ with any predicate $\varphi$. (Our actual typing rule, shown
in Section~\ref{ssec:typing-rules} after more notation has been
introduced, is more complex. The actual rule also admits predicates
over traces, which are more general than predicates over individual
actions that we have considered here.)
%
%\vspace{-10pt}
\[
\inferrule*[right=Confine] {
\fa(e)=\emptyset\\\fv(e)\in\Gamma
\\\\ \cnfn{\tau} {\varphi}
\\ \cnfn{\Gamma} {\varphi}} {\Gamma \vdash e : \tau}
\]

In our Memoir example, we use the \rulename{Confine} rule to derive the invariants
of the service invoked by the attacker.  For instance, if we can show
that each of the TPM primitives do not reset the value of the PCR, then using
the \rulename{Confine} rule, we can claim that $\srvc$, when applied to these
primitives does not reset the value of the PCR. We revisit this proof with
specific details in Section~\ref{ssec:typing-examples}.

In typing a statically unknown expression using the \rulename{Confine} rule we
assume that the expression is syntactically free of actions and that all of its
free variables are in $\Gamma$.  These are reasonable assumptions for untrusted
code to be sandboxed. In an implementation these assumptions can be discharged
either by dynamic checks during execution, by static checks during program
linking, or by hardware-enforced interface confinement. For example, in our Memoir
analysis, the hardware ensures that TPM state can be modified by the service only
using the TPM interface.

\paragraph{Deriving properties based on code integrity} Next we need
to show that $\mi{srvc}$ does not leak its secret key. We assume this
property about the initial service Memoir was invoked with. (This
property could be verified either by manual audits or automated static
analysis of the service code).  However, in our model the adversary
could invoke Memoir on malicious service code (e.g., replacing a
legitimate password manager service with code of the adversary's
choice). In this case, we can show with additional reasoning that
$\mi{srvc}$ invoked later must be the same program as the intial
service.  To allow typing $\mi{srvc}$, based on the proof of equality
with the initial service and an assumed type for the initial service,
we add a new rule called \rulename{Eq}.

\vspace{-10pt}
\[
\inferrule*[right=Eq]{\Gamma \vdash e : \tau\\
 \Gamma \vdash  e = e' \jtrue }{
\Gamma \vdash e' : \tau }
\]
\vspace{-10pt}

The \rulename{Eq} rule assigns the type $\tau$ of any expression $e$
to any other expression~$e'$, which is known to be syntactically equal to
$e$. This rule is trivially sound.

This pattern of first establishing code identity (identify an unknown
code with some known code) and then using it to refine types is quite
common in proofs of security-relevant properties. A similar
pattern arises in analysis of systems that rely on memory
protections to ensure that code read from the shared memory is the
same as a piece of trusted code, and therefore, safe to execute.  In
Datta \emph{et al.}'s work on analysis of remote attestation
protocols~\cite{dfgk09:ls2}, similar patterns arise for typing
potentially modified software executed in a machine's boot
sequence. Their model is untyped, but if it were to be typed,
\rulename{Eq} could be used to complete the proofs.

\section{Type System and Assertion Logic}
\label{sec:type-system}
\begin{figure}[t!]
\(%\small
\hspace{-5pt}
\begin{array}{lcll}
\textit{\small Expr types} & \tau & ::= &  X ~|~
\bt ~|~\Pi x{:}
\tau_1. \tau_2 ~|~ \forall X.\tau
~|~ \tcomp(\eta_c) ~|~ \texp
\\
\textit{\small Comp types} & \eta & ::= & x{:}\tau.\varphi
~|~  \varphi ~|~  (x{:}\tau.\varphi, \varphi')\\
\textit{\small Closed c types} & \eta_c & ::= &
u_1.u_2.i.(x{:}\tau.\varphi_1, \varphi_2)
\\ & &~|~ & \Pi x{:}\tau. u_1.u_2.i.(y{:}\tau.\varphi_1, \varphi_2)
\\
\textit{\small Assertions} & \varphi & ::= & P ~|~ e_1 = e_2 ~|~ \varphi \; e ~|~ \top ~|~ \bot ~|~ \neg\varphi
\\ & & ~|~ & \varphi_1 \conj \varphi_2
~|~\varphi_1 \disj \varphi_2
~|~ \forall x{:} \tau. \varphi ~|~ \exists x{:}\tau. \varphi
\vspace{5pt}
\\
  \textit{\small Action Kinds} & \alpha & ::= & \kact(\eta_c)
  ~|~ \Pi x{:}\tau. \alpha ~|~ \forall X.\alpha
\\
\textit{\small Type var ctx} & \Theta & ::= & \cdot ~|~ \Theta, X\\
\textit{\small Signatures} & \Sigma & ::= & \cdot ~|~ \Sigma, A:: \alpha\\
\textit{\small Logic var ctx} & \Gamma^L & ::= & \cdot ~|~ \Gamma^L, x:\bt ~|~ \Gamma^L, x: \texp \\
\textit{\small Typing ctx} & \Gamma & ::= & \cdot ~|~ \Gamma, x: \tau \\
\textit{\small Formula ctx} & \Delta & ::= & \cdot ~|~\Delta, \varphi\\
\textit{\small Exec ctx} & \Xi & ::= & u_b:b, u_e:b, i:b
\end{array}
\)
\caption{Types and typing contexts}
\label{fig:type-syntax}
\end{figure}

The syntax for {\lang} types is shown in Figure~\ref{fig:type-syntax}.
Types for expressions, denoted $\tau$, include type variables ($X$),
a base type $\bt$, dependent function types ($\Pi
x{:}\tau_1. \tau_2$), and polymorphic function types ($\forall
X.\tau$).
Since \lang focuses on deriving trace properties of programs, the
difference between base types such as \tunit\ and \tbool\ is of little
significance. Therefore, \lang has one base type $\bt$ to classify all
first-order terms.  The type $\texp$ contains all syntactically
well-formed expressions ($\texp$ stands for ``untyped''). Memory
always stores expressions of type {\texp} because the adversary could
potentially write to any memory location.

Similar to HTT, a
suspended computation $\ecomp(c)$ is assigned a monadic type
$\tcomp(\eta_c)$, where $\eta_c$ is a closed computation type. A
closed computation type $u_1.u_2.i.(x{:}\tau.\varphi_1, \varphi_2)$
contains two postconditions, $\varphi_1$ and $\varphi_2$. Both are
interpreted relative to a trace $\trace$. $\varphi_1$, the
\emph{partial correctness assertion}, holds whenever a computation of
this type finishes execution on the trace. It is parametrized by the
id $i$ of the thread that runs the computation, the interval $(u_b,
u_e]$ during which the computation runs and the return value $x$ of
  the computation. $\varphi_2$, called the \emph{invariant assertion},
  holds while a computation of the computation type is still executing
  (or is stuck), but has not returned. It is parametrized by the id
  $i$ of the thread running the computation and the time interval
  $(u_b, u_e]$ over which the computation has executed. Formally, a
    suspended computation $\ecomp(c)$ has type
    $\tcomp(u_1.u_2.i.(x{:}\tau.\varphi_1, \varphi_2))$ if the
    following two properties hold for every trace $\trace$: (1) if a
    thread $\iota$ on trace $\trace$ begins to run $c$ at time $U_1$
    and at time $U_2$, $c$ returns an expression $e$, then $e$ has
    type $\tau$, and $\trace$ satisfies $\varphi_1[U_1, U_2, \iota,
      e/u_1, u_2, i,x]$; (2), if a thread $\iota$ on trace $\trace$
    begins to run $c$ at time $U_1$ and at time $U_2$, $c$ has not
    finished, then $\trace$ satisfies $\varphi_2[U_1, U_2, \iota/u_1,
      u_2, i]$. The meaning of all types is made precise in
    Section~\ref{ssec:model}.

The type $\eta$ may be
either a partial correctness assertion, an invariant assertion, or a
pair of both. Fixpoint
computations have the type $\Pi
x{:}\tau. u_1.u_2.i.(y{:}\tau.\varphi_1, \varphi_2)$,
discussed in more detail with typing rules. If
$f$ has this type, then for any $e: \tau$, $(f\; e)$ is a recursive
computation of closed computation type $u_1.u_2.i.(y{:}\tau.\varphi_1,
\varphi_2)[e/x]$.

Assertions, denoted $\varphi$, are standard first-order logical
formulas interpreted over traces. Atomic assertions are denoted $P$.

We write $\alpha$ to categorize actions. A fully applied action has the
type  $\kact(\eta_c)$, where $\eta_c$ denotes the action's effects.
\subsection{Typing Rules}
\label{ssec:typing-rules}
Our typing judgments use several contexts.  $\Theta$ is a list of type
variables. The signature $\Sigma$ contains specifications for action
symbols. $\Gamma^L$ contains logical variable type bindings. These
variables can only be of the type $\bt$ or $\texp$.  $\Gamma$ contains
dependent variable type bindings. $\Delta$ contains logical
assertions.  The ordered context $\Xi = u_b,u_e,i$ provides reference
time points and a thread id to typing judgments for computations. When
typing a computation, $(u_b, u_e]$ are parameters representing the
  interval during which the computation executes and $i$ is a
  parameter representing the id of the thread that executes the
  computation. A summary of the typing judgments is shown below.

\vspace{2pt}
\hspace{-18pt}
\(%\small
\begin{array}{ll}
  u{:}\bt; \Theta; \Sigma; \Gamma^L;\Gamma;\Delta \vdash_Q e: \tau
  & \mbox{expression $e$ has type $\tau$}
  \\  u{:}\bt; \Theta; \Sigma;
 \Gamma^L;\Gamma;\Delta \vdash_Q c : \eta_c
  & \mbox{fixed-point computation $c$ has type $\eta_c$}
  \\\Xi;  \Theta; \Sigma;\Gamma^L; \Gamma;\Delta \vdash_Q c
  : \eta
  &  \mbox{computation $c$ has type $\eta$}
  \\  \Xi;  \Theta; \Sigma; \Gamma^L;\Gamma;\Delta \vdash \varphi
  \jonempty &
  \mbox{ $\varphi$ holds while reductions are}
\\&\mbox{non-effectful}
  \\\Theta;  \Sigma; \Gamma^L;\Gamma;\Delta \vdash \varphi  \jtrue
  & \mbox{$\varphi$ is true}
\end{array}
\)
\vspace{2pt}

When typing expressions and fixpoint computations, $u$ is earliest
time point when the term can be evaluated on the trace. The first
three judgments are indexed by a qualifier $Q$, which can either be
empty or $u_b.u_e.i.\varphi$, which we call an invariant.  Variables
$u_b$, $u_e$, and $i$ have the same meaning as the context $\Xi$, and
may appear free in $\varphi$.
Rules indexed with $u_b.u_e.i.\varphi$ are used for deriving
properties of programs that execute adversarial code. Roughly
speaking, the context $\Gamma$ in these rules contains variables that
are place holders for expressions that satisfy the invariant $\varphi$.
 We explain here some selected rules of our type system; the
remaining rules are listed in the accompanying technical appendix.

\paragraph{Silent threads}
Reductions on a trace can be categorized into those induced by the
rules \rulename{R-ActS} and \rulename{R-ActF} in
Figure~\ref{fig:semantics:reduction-config} and those induced by other
rules. We call the former effectful and the latter non-effectful or
silent. The typing judgment $\Xi;\Theta;\Sigma;\Gamma^L;\Gamma;\Delta \vdash \varphi\jonempty$
specifies properties of threads while they perform only silent
reductions or do not reduce at all. The judgment is auxiliary in
proofs of both partial correctness and invariant assertions, as will
become clear soon. The following rule states that if $\varphi$ is
true, then a trace containing a thread's silent computation
satisfies $\varphi$.

\vspace{-8pt}
\[
\inferrule*[right=Silent]{
  \Xi;\Theta;\Sigma; \Gamma^L;\Gamma;\Delta \vdash \varphi\jtrue
\\  \Xi; \Theta;\Sigma; \Gamma^L;\Gamma;\Delta \vdash \varphi \jok}{
\Xi;\Theta;\Sigma;\Gamma^L;\Gamma;\Delta  \vdash \varphi\jonempty
}
\]
\vspace{-10pt}

The type system may be extended with other sound rules for this
judgment. \limin{to update} For instance, the following is a trivially
sound rule: $u_b.u_e.i;\Theta;\Sigma;\Gamma^L;\Gamma;\Delta \vdash
(\forall l, t, u_b{<}t\leq u_e \imp\neg\pred{Read}\ i\ l\
t)\jonempty$. If a thread $i$ is not performing any action
during time interval $(u_b,u_e]$, then it does not read memory
during that time interval.

\begin{figure}[t!]
{\bf Partial correctness typing}
\begin{mathpar}
\mprset{flushleft}
\inferrule*[right=Act]{u_1;\Theta; \Sigma;\Gamma^L,u_2,i; \Gamma;\Delta \vdash_Q a ::
   \kact(\tpair{u_b.u_e.j}{x{:}\tau.\varphi_1}{\varphi_2} )
\\
   u_1, u_2, i; \Theta; \Sigma;\Gamma^L; \Gamma;\Delta  \vdash
   \varphi \jonempty
\\ \fv(a)\in\dom(\Gamma)
\\ \mbox{let}~\gamma=[u_1,u_2,i/u_b,u_e,j]
\\ d; \Sigma;\Gamma^L; \Gamma\vdash
    u_1.u_2.i.(x{:}\tau.\varphi_1\gamma,
 \varphi_2\gamma\conj \varphi)\jok
}{
 u_1, u_2, i; \Theta; \Sigma;\Gamma^L; \Gamma;\Delta
    \vdash_Q \cact(a) : (x{:}\tau.\varphi_1\gamma,
 \varphi_2\gamma\conj \varphi)}
\and
\inferrule*[right=Ret]{
u_2; \Theta; \Sigma;\Gamma^L,  u_1, i;
\Gamma;\Delta %,u_0\geq u_2
\vdash_Q e: \tau \\
    u_1, u_2, i; \Theta; \Sigma;\Gamma^L; \Gamma;\Delta  \vdash
   \varphi \jonempty
\\ \fv(e)\subseteq\dom(\Gamma) }{
   u_1, u_2, i;  \Theta; \Sigma;\Gamma^L; \Gamma;\Delta
     \vdash_Q \cret(e): x{:}\tau.((x = e) \conj
   \varphi)}
\and
\inferrule*[right=SeqC]{
    u_0, u_1, i;
 \Theta; \Sigma;\Gamma^L;  u_3,\Gamma;\Delta , u_0\leq u_1 \vdash \varphi_0 \jonempty
   \\ u_1, u_2, i; \Theta; \Sigma;\Gamma^L, u_0:
   \bt,  u_3;\Gamma;\Delta,
    u_1 < u_2, \varphi_0
\\\\~~~~\vdash_Q c_1 : x{:}\tau.\varphi_1
  \\\\ u_2, u_3, i; \Theta; \Sigma;\Gamma^L,
   u_0, u_1;\Gamma,x : \tau;\Delta,u_2 < u_3,\varphi_0, \varphi_1
    \\\\~~~~\vdash_Q
    c_2: y{:}\tau'.\varphi_2
\\\\
  \Theta; \Sigma;\Gamma^L,u_1, u_2,  u_0, u_3, i;\Gamma,
   x{:}\tau,y:\tau';\Delta
\\\\~~~~\vdash
  (\varphi_0 \conj \varphi_1 \conj
  \varphi_2) \Rightarrow \varphi \jtrue
%\\\\ \Theta;\Sigma;\Gamma^L;u_0, u_3, i,\Gamma\vdash  \tau'\jok
\\  \Theta;\Sigma;\Gamma^L, u_0, u_3, i;\Gamma,
y:\tau' \vdash  \varphi \jok
\\\fv(\cletc(c_1, x.c_2))\subseteq\dom(\Gamma)
}{
  u_0, u_3, i; \Theta; \Sigma;\Gamma^L; \Gamma;\Delta  \vdash_Q
   \cletc(c_1, x.c_2): y{:}\tau'.\varphi
}
\and
\inferrule*[right=SeqCComp]{
    u_0, u_1, i;
 \Theta; \Sigma;\Gamma^L, u_3;\cdot;\Delta , u_0\leq u_1 \vdash
 \varphi_0 \jonempty
   \\ u_1, u_2, i; \Theta; \Sigma;\Gamma^L, u_0:
   \bt,  u_3;\cdot;\varphi_0, u_1\leq u_2  
\\\\~~~~\vdash_Q c_1 : x{:}\tau.\varphi_1
  \\\\ u_2, u_3, i; \Theta; \Sigma;\Gamma^L;
   u_0, u_1; x : \tau;\Delta,u_2 \leq u_3,\varphi_0, \varphi_1
   \\\\~~~~\vdash_{Q2}
    c_2: y{:}\tau'.\varphi_2
\\\\
  \Theta; \Sigma;\Gamma^L;  u_0, u_3, i;\Gamma,
   u_1, u_2, y:\tau';\Delta
\\\\~~~~\vdash
  (\varphi_0 \conj \varphi_1 \conj
  \varphi_2) \Rightarrow \varphi \jtrue
\\  \Theta;\Sigma;\Gamma^L;  u_0, u_3, i,\Gamma,
y:\tau' \vdash  \varphi \jok
}{
  u_0, u_3, i; \Theta; \Sigma;\Gamma^L; \Gamma;\Delta  \vdash_{Q2}
   (c_1; c_2): y{:}\tau'.\varphi
}
\end{mathpar}

{\bf Invariant typing}
\begin{mathpar}
\mprset{flushleft}
 \inferrule*[right=SeqCI]{
 \Theta ;\Sigma;\Gamma^L,  u_0, u_3, i;\Gamma;\Delta\vdash  \varphi \jok
\\  u_0, u_1, i;
 \Theta; \Sigma;\Gamma^L, u_3;\Gamma;\Delta , u_0\leq u_1 \vdash \varphi_0 \jonempty
\\  u_0, u_3, i;
 \Theta; \Sigma;\Gamma^L; \Gamma;\Delta , u_0\leq u_3 \vdash \varphi'_0 \jonempty
 \\u_1, u_2, i;  \Theta; \Sigma;\Gamma^L,u_0:
 \bt, u_3;\Gamma;\Delta,   u_1 < u_2, \varphi_0
\\\\~~~~\vdash_Q c_1 : x{:}\tau.\varphi_1
\\u_1, u_3, i;  \Theta; \Sigma;\Gamma^L;\Gamma;\Delta, u_0:
 \bt, u_1 \leq u_3, \varphi_0 \vdash_Q c_1 : \varphi'_1
  \\ u_2, u_3, i; \Theta; \Sigma;\Gamma^L; \Gamma;\Delta,
   u_0, u_1, x: \tau, u_2\leq u_3,\varphi_0, \varphi_1
\\\\~~~~\vdash_Q
   c_2: \varphi_2
\\\\
 \Theta; \Sigma;\Gamma^L,  u_0, u_3, i;\Gamma;\Delta  \vdash
   \varphi'_0 \imp \varphi \jtrue
\\\\
 \Theta; \Sigma;\Gamma^L,  u_0, u_3, i;\Gamma,u_1;\Delta
\vdash
  (\varphi_0 \conj \varphi'_1) \imp \varphi \jtrue
\\\\
  \Theta; \Sigma;\Gamma^L,  u_0, u_3, i;\Gamma, u_1, u_2,
  x{:}\tau;\Delta
\\\\~~~~\vdash (\varphi_0\conj \varphi_1 \conj
  \varphi_2) \imp \varphi \jtrue
\\\fv(\cletc(c_1, x.c_2))\subseteq\dom(\Gamma)
}{
  u_0, u_3, i; \Theta; \Sigma;\Gamma^L; \Gamma;\Delta  \vdash_Q
   \cletc(c_1, x.c_2): \varphi }
\end{mathpar}
\caption{Selected Rules for Computation Typing}
\label{fig:comp-typing}
\end{figure}

%\subsubsection{Computation Typing}
\paragraph{Partial correctness typing for computations}
Figure~\ref{fig:comp-typing} shows selected rules for establishing
partial correctness postconditions of computations.  The judgment
$u_1, u_2, i;\Theta; \Sigma;\Gamma^L;\Gamma;\Delta \vdash c: x{:}\tau.\varphi$ means that if in
trace $\trace$ any thread with id $\iota$ begins to execute
computation $c$ at time $U_1$, and at time $U_2$, $c$ returns an
expression $e$, and $\trace$ satisfies all the formulas in $\Delta$,
then $e$ has type $\tau$, and $\trace$ also satisfies $\varphi[U_1,
U_2, \iota, e/ u_1, u_2, i, x]$.

In rule \rulename{Act}, the type of an atomic action is directly
derived from the specification of the action symbol in $a$.  We elide
rules for the judgment $a ::
\kact(\tpair{u_1.u_2.i}{x{:}\tau.\varphi_1}{\varphi_2})$, which
derives types for actions based on the specifications in~$\Sigma$.  We
explain the invariant assertions for actions with the discussion of invariant
typing for computations. When
typing $a$, the logical variable typing context includes $u_2:\bt$ and
$i:\bt$, because they may appear free in $\Gamma$ and $\Delta$.  For
brevity, we elide the types for variables of type $\bt$, as they are
obvious from the context. %\dg{The ACT rule has several typos.}

Rule \rulename{Ret} assigns $e$'s type to $\cret(e)$. The trace
$\trace$ containing the evaluation of $\cret(e)$ satisfies two
properties, which appear in the postcondition of $\cret(e)$. First,
the return expression, which is bound to $x$, is $e$ (assertion $(x =
e)$). Second, $\trace$ satisfies any property $\varphi$ such that
$\varphi \jonempty$ holds. This is because reduction of $\cret(e)$ is
silent. Here $e$ is typed under the time point $u_2$, indicating that
$e$ can only be evaluated after
$u_2$. % We explain this more with the typing rules for
% expressions.

Rule \rulename{SeqC} types the sequential composition
$\cletc(c_1,x.c_2)$. Starting at time point $u_0$ and returning at
$u_3$, the execution of $\cletc(c_1,x.c_2)$ in any thread $i$ can be
divided into three segments for some $u_1,u_2$: between time $u_0$ and
$u_1$, where thread $i$ takes only a silent step, pushing $x.c_2$ onto
the stack; between time $u_1$ and $u_2$, where the computation $c_1$
runs; and between time $u_2$ and $u_3$, where $c_2$ runs.  The first
three premises of \rulename{SeqC} assert the effects of each these
three segments. When type checking $c_2$, the facts learned from the
execution so far ($\varphi_0$ and $\varphi_1$) are included in the
context. The fourth premise checks that $\varphi$ is the logical
consequence of the conjunction of the three evaluation segments'
properties.
% \dg{In rule SeqC, many premises split into more than one
%   line and it is difficult to discern where a premise ends. Use
%   indentation.}

The above rules have the same qualifier $Q$ in the premises and the
conclusion. Rule \rulename{SeqCComp} combines derivations with
different qualifiers in a sequencing statement. 
%  The first computation
% $c_1$ is typed with an invariant
% $u_b.u_e.i.\varphi_1$, denoted $Q$. The second computation $c_2$ is
% type checked under a different qualifier $Q2$, which is either empty
% or an invariant $u_b.u_e.i.\varphi_2$.  
The $\Gamma$ context in the typing of $c_1$ and $c_2$ must be
empty. 
Because the free variables in
$c_1$ are place holders for expressions that satisfy an invariant
$\varphi_1$, while the free variables in $c_2$ are for ones
that satisfy a different invariant $\varphi_2$, $c_1$ and $c_2$ cannot share
free variables except those in $\Gamma^L$. 
Note that both $Q$ and $Q2$ can be empty.
This rule is necessary for
typing the sequential composition of two programs that contain
differently sandboxed code: $c_1$ executes sandboxed code that satisfies
$\varphi_1$ and $c_2$ either contains no sandboxed programs, or ones that
satisfy $\varphi_2$ .

% \dg{Two points. First, in
%   the conclusion of \rulename{SeqCComp}, you want $c_1;c_2$, not
%   $(c_1,c_2)$.}
% \dg{Second, the explanation in this paragraph lack a proper
%   motivation. Why are we even interested in qualifiers that contain
%   invariants?}

% The substitutions for
% $\Gamma^L$ are expressions of base types and are independent of any
% effects.

\paragraph{Invariant typing for computations}
The meaning of the invariant typing judgment $u_1, u_2, i;\Theta;
\Sigma;\Gamma^L; \Gamma;\Delta
\vdash c: \varphi$ is the following: Assuming that on a trace
$\trace$, thread $\iota$ begins to execute $c$ at time $U_1$, and at
time $U_2$ $c$ has not yet returned (this includes the possibility
that $c$ is looping indefinitely or is stuck), if $\trace$ satisfies
assumptions in $\Delta$, then $\trace$ also satisfies $\varphi[U_1,
  U_2, \iota/ u_1, u_2, i]$.

  We first explain the invariant assertions for actions (rule
  \rulename{Act}).  The thread executing the atomic action is silent
  before the action returns. Therefore, the invariant assertion of the
  action is the conjunction of the invariant specified in $\Sigma$ and
  the effect of being silent.

Next, we explain the rule \rulename{SeqCI} for the sequencing statement
$\cletc(c_1, x.c_2)$. We need to consider three cases when deriving
the invariant assertion $\varphi$ of $\cletc(c_1, x.c_2)$
in the interval $(u_0, u_3]$: (1) the computation
  has not started until $u_3$ (2) the computation $c_1$ started but
  has not returned until $u_3$, (3) the computation $c_1$ has
  returned, but $c_2$ has not returned until~$u_3$. The first five
  premises of rule \rulename{SeqCI} establish properties of a silent
  thread, the partial correctness and invariant assertions of the
  computation in $c_1$, and the invariant assertion of $c_2$. The next
  three judgments check that in each of the three cases (1)--(3), the
  final assertion $\varphi$ holds.

  For example, $\ecomp(\cletc(\cact(\reada\ e), x.\cret x))$ can be
  assigned the following type. Predicate $(\pred{mem}\ l\ v\ u)$ is
  true when at time $u$, memory location $l$ is allocated and stores
  the expression $v$.  Predicate $\pred{eval}\ e\ e'$ is true if $e$
  $\beta$-reduces to $e'$, which cannot reduce further. 
$\pred{Write}\ \iota\ l\ e\ u$ states that
thread $\iota$ writes to address $l$ expression $e$ at time $u$.
 The partial
  correctness assertion states that this suspended computation returns
what's stored in the location that $e$ reduces to. The invariant
assertion states that during its execution, the thread executing it
does not write to the memory.

\vspace{-8pt}
\begin{tabbing}
  $\tcomp(u_b.u_e.i.($\= $r{:}\texp. 
    \forall l, v,$ $\pred{eval}\ e\ l \conj \pred{mem}\ l\ v\ u_e 
    \imp y =e $, 
\\\> $\forall l,v,u,$\= $u_b< u \leq u_e     
  \imp\neg \pred{write}\ i \ l\ v\ u))$
\end{tabbing}
\vspace{-8pt}
%% \Paragraph{Assumptions in postconditions}
%% % The computation types do not distinguish between a precondition and a
%% % postcondition. Instead, preconditions can be expressed using an
%% % implication $\varphi_1\imp\varphi_2$.
%% The following rule allows the typing of a computation $c$ to directly
%% use the assumption $\varphi_1$, if the assertion of $c$ in the
%% conclusion is an implication $\varphi_1\imp\varphi_2$. This rule is
%% necessary for our case study as the typing of $c$ could depend on
%% assertions about the system as a whole, which can be assumed in
%% $\varphi_1$.
%%  %%  These two rules
%% %% move the left-hand-side formula of an implication to the local logical
%% %% context $\Delta_L$.
%% There is a corresponding rule for invariant typing as well.

%% \vspace{5pt}
%% \(
%%  \inferrule*[right=Imp]{ u_1, u_2, i;  \Theta; \Sigma; \Gamma
%%      ; \Delta_G; \Delta_L, \varphi_1 \vdash c: x{:}\tau.\varphi_2}
%%  { u_1, u_2, i;  \Theta; \Sigma; \Gamma
%%      ; \Delta_G; \Delta_L \vdash c: x{:}\tau.(\varphi_1\imp\varphi_2)}
%% \)
%% \vspace{5pt}

\paragraph{Fixpoint computation}
The fixpoint is typed under 
a time point $u$, which is the earliest
time when the fixpoint is unrolled. 

%  The reason is
% that the fixpoint computation may unroll many times, each has a
% different ending time, so we only
% specify the beginning time point when the fixpoint is being unrolled
% (for computation $c\;e$).

\(\hspace{-15pt}
\mprset{flushleft}
 \inferrule*[right=Fix]{
\Gamma_1 =  y:\tau, f: \Pi
 y{:}\tau. \tcomp(u_1.u_3.i.(x{:}\tau_1.\varphi, \varphi'))
\\\\
 u_1, u_2, i; \Theta;\Sigma;\Gamma^L; \Gamma;\Delta,
  u\leq  u_1\leq u_2 \vdash\varphi_0 \jonempty
\\\\ u_2, u_3, i; \Theta;
   \Sigma;\Gamma^L,u_1, u;
   \Gamma, \Gamma_1;\Delta,  u_2 <  u_3, \varphi_0  \vdash_Q c : x{:}\tau_1.\varphi_1
\\\\ u_2, u_3, i; \Theta;
   \Sigma;\Gamma^L; u_1, u;
   \Gamma, \Gamma_1;\Delta, u_2 \leq   u_3, \varphi_0  \vdash_Q c :\varphi_2
   \\\\ \Theta;\Sigma;\Gamma^L, u_1, u,  u_2, u_3,  i; \Gamma, \Gamma_1,
  x:\tau_1;\Delta   \vdash
  (\varphi_0 \conj \varphi_1) \Rightarrow \varphi \jtrue
   \\\\ \Theta;\Sigma;\Gamma^L , u_1 , u_2, u_3,
   i, u; \Gamma,\Gamma_1;\Delta
             \vdash (\varphi_0 \conj \varphi_2 \imp \varphi') \jtrue
   \\\\ \Theta;\Sigma;\Gamma^L, u_1,u_3 ,
   i , u; \Gamma, y:\tau;\Delta\vdash
   \varphi_0[u_3/u_2] \imp \varphi' \jtrue
\\\\ \Theta;\Sigma;\Gamma^L,u;\Gamma \vdash \Pi   y{:}\tau.u_1.u_3.i.(x{:}\tau_1.\varphi,
   \varphi') \jok
\\ \fv(\cfix(f(y).c))\in\dom(\Gamma)
 }{ u; \Theta;\Sigma;\Gamma^L; \Gamma;\Delta  \vdash_Q \cfix(f(y).c) : \Pi
   y{:}\tau.u_1.u_3.i.(x{:}\tau_1.\varphi, \varphi') }
\)

Rule \rulename{Fix} simultaneously establishes the partial correctness
and invariant assertions of a fixpoint. The third and fourth premises
establish the partial correctness and invariant assertions of the body
$c$ of the fixpoint. The fifth premise checks that the specified
partial correctness assertion $\varphi$ is entailed by the conjunction
of the assertions of a silent thread and the assertion of the
body. The next two premises check the invariant assertion $\varphi'$.
For example, $\cfix\ f(x).\writea\ x\ 0; \reada\ x;\clete(f(x{+}1); z.\cret\ z)$
has the type:
\vspace{-5pt}
 \begin{tabbing}
 $\Pi x{:}\bt.u_b.u_e.i.($\=$y{:}\texp.\bot$,
\\\>$\forall u, l, v,$\=$u_b<u\leq u_e\conj\reada\ i\ l\ u$
\\\>\> $\imp\exists u', u'<u\conj\writea\ i\ l\ v\ u')$
\end{tabbing}
\paragraph{Expression typing}
 %(Figure~\ref{fig:static:expression}).

\begin{figure}[t!]
\begin{mathpar}
\mprset{flushleft}
\inferrule*[right=Comp]{
 u_1, u_2, i;  \Theta;\Sigma;\Gamma^L;u_e,\Gamma;\Delta, u_1\geq u_e
\vdash_Q c : (x{:}\tau.\varphi_1, \varphi_2)
%\\ \fv(\Delta)\subseteq\dom(\Gamma)
\\
   \Theta;\Sigma;\Gamma^L, u_e{:}\bt, u_1{:}\bt, u_2{:}\bt, i{:}\bt;\Gamma, x:\tau ;\Delta
\vdash \varphi_1 \Rightarrow \varphi'_1  \jtrue
\\\\
   \Theta;\Sigma;\Gamma^L,u_e{:}\bt, u_1{:}\bt, u_2{:}\bt, i{:}\bt ;\Gamma;\Delta
 \vdash \varphi_2 \imp \varphi'_2  \jtrue
 \\ \Theta;\Sigma;\Gamma^L,u_e{:}\bt;\Gamma\vdash
 \tpair{u_1.u_2.i}{x{:}\tau.\varphi'_1}{\varphi'_2}\jok
\\\fv(c)\subseteq\dom(\Gamma)
}{
u_e; \Theta;\Sigma;\Gamma^L;\Gamma;\Delta  \vdash_Q \ecomp(c) :
  \tcomp(\tpair{u_1.u_2.i}{x{:}\tau.\varphi'_1}{\varphi'_2})}
\and
\inferrule*[right=Eq]{u;\Theta;\Sigma;\Gamma^L;\Gamma;\Delta \vdash_Q e : \tau\\
\Theta;\Sigma;\Gamma^L,u;\Gamma;\Delta \vdash  e = e' \jtrue
\\ \fv(e')\subseteq\dom(\Gamma)}{
u;\Theta;\Sigma;\Gamma^L;\Gamma;\Delta \vdash_Q e' : \tau }
\and
\inferrule*[right=Confine]
{
 \mbox{$\varphi$ is trace composable}
\\ u_b, u_e, i; \Theta;\Sigma;\Gamma^L,u;\Gamma;\Delta \vdash \varphi\jonempty
\\ u_b{:}\bt, u_e{:}\bt, i{:}\bt \vdash
\varphi\jok
 \\\fa(e)=\emptyset
\\\fv(e)\subseteq\Gamma
 \\\cnfn{\tau}{u_b.u_e.i.\varphi}
 \\\cnfn{\Gamma}{u_b.u_e.i.\varphi}
}
{u;\Theta;\Sigma;\Gamma^L;\Gamma;\Delta  \vdash_{u_b.u_e.i.\varphi} e : \tau}
\and
\inferrule*[right=Conf-sub]
{u;\Theta;\Sigma;\Gamma^L;\Gamma;\Delta  \vdash e : \tau
\\ u_b{:}\bt, u_e{:}\bt, i{:}\bt \vdash \varphi\jok
}
{u;\Theta;\Sigma;\Gamma^L;\Gamma;\Delta  \vdash_{u_b.u_e.i.\varphi} e : \tau}
\end{mathpar}
\vspace{-10pt}
\caption{Selected expression typing rules}
\label{fig:typing-expression}
\end{figure}

Similar to the fixpoint, the expression typing judgment is
parameterized over a time point $u$, which is the earliest time point
that $e$ is evaluated. Recall that the typing rule for $\cret(e)$
types $e$ under the time point when $\cret(e)$ returns. This is
because $e$ can only be evaluated after $\cret(e)$ finishes. Most
expression typing rules are standard. A representative subset is
listed in Figure~\ref{fig:typing-expression}.

Rule \rulename{Comp} assigns a
monadic type to a suspended computation by checking the
computation. Since the suspended computation can only execute after
$u_e$, the logical context of the first premise can safely assume that
the beginning time point of $c$ is no earlier than $u_e$. As usual, the
rule also builds-in weakening of postconditions.

The rule \rulename{Eq}, motivated in Section~\ref{ssec:challenges},
assigns an expression $e'$, the type of $e$, if $e$ is syntactically equal to
$e'$.

The rule \rulename{Confine}, motivated in
Section~\ref{ssec:challenges}, allows us to type an expression from
the knowledge that it contains no actions and that its free variables
will be substituted with expressions with effect $\varphi$. The main
generalization from the simpler rule presented in
Section~\ref{ssec:challenges} is that now $\varphi$ is a predicate
over an interval and a thread in a trace, not just a predicate over
individual actions. The intuitive idea behind the rule is similar: If
$c$ is a computation that is free of actions and confined to use the
computations $c_1,\ldots,c_n$ for interaction with the shared state,
and each of the computations $c_1,\ldots,c_n$ maintain a trace
invariant $\varphi$ while they execute, then as $c$ executes, it
maintains $\varphi$.

Technically, because $\varphi$ also accepts as arguments any interval
on a trace (it has free variables $u_b, u_e$), we require that
$\varphi$ be \emph{trace composable}, meaning that if $\varphi$ holds
on two consecutive intervals of a trace, then it hold across the union
of the intervals. Formally, $\varphi$ is trace composable if $\forall
u_1, u_2, u_3,i. \; (\varphi(u_1, u_2,i) \conj \varphi(u_2, u_3,i))
\imp \varphi(u_1, u_3,i)$.  Further $\varphi$ has to hold on intervals
when thread $i$ is silent. This prevents us from derving arbitrary
properties of untrusted code. For instance, $\varphi$ cannot be
$\bot$. (No trace can satisfy the invariant $\bot$.) This rule relies
on checking that $\tau$ relates to the invariant $\varphi$,
represented as the relation $\cnfn{\tau} {u_b.u_e.i.\varphi}$. This
relation means that $\varphi$ is both the partial correctness
assertion and the invariant assertion in every computation type
$\tcomp(\eta_c)$ occurring in $\tau$. Similarly, $\Gamma$ is required
to map every free variable in $e$ to a type that satisfied the same
relation.
The conclusion is indexed by the invariant
$u_b.u_e.i.\varphi$ to record the fact that all substitutions for
variables in $\Gamma$ need to satisfy $\varphi$.

\begin{center}\vspace{-5pt}
\(
\inferrule{ }{
  \cnfn{b}{u_b.u_e.i.\varphi}}
\)

\vspace{5pt}

\(
\inferrule{ \cnfn{\tau_1}{u_b.u_e.i.\varphi}~~~~
\cnfn{\tau_2}{u_b.u_e.i.\varphi} }{
  \cnfn{\Pi
    \_{:}\tau_1. \tau_2}{u_b.u_e.i.\varphi}}
\)

\vspace{5pt}

\(
\inferrule{\cnfn{\tau}{u_b.u_e.i.\varphi} }{
 \cnfn{\tcomp(\tpair{u_b.u_e.i}{x{:}\tau.\varphi}{\varphi})}{u_b.u_e.i.\varphi}}
\)
\end{center}\vspace{-5pt}

The \rulename{Confine} rule itself does not stipulate any conditions
on the predicate $\varphi$, other than requiring that $\varphi$ be
trace composable. However, if $e$ is of function type, and expects
some interfaces as arguments, then in applying \rulename{Confine} to
$e$, we must choose a $\varphi$ to match the actual effects of those
interfaces, else the application of $e$ to the interfaces cannot be
typed.

The rule \rulename{Conf-Sub} constrains a regular typing derivation to
a specific invariant $u_b.u_e.i.\varphi$. This is sound because the
first premise does not require the substitutions for $\Gamma$ to
satisfy any specific invariant, so they can be narrowed down to any
invariant. The conclusion must be tagged with the invariant $\varphi$,
because: (1) $\tau$ could be a base type, in which case, the invariant is
not evident in $e$'s type; and (2) the types in $\Gamma$ are allowed to
contain nested effects that are not $\varphi$.
Reason (1) is also why the conclusion of the \rulename{Confine} rule
is indexed.

Finally, the time point enables expression types to include facts that
are established by programs executed earlier. For example, the return type of 
$\cletc(a_1; z.\cret(\ecomp (a_2)))$ can be the following, assuming
that the effect of action $a_1$ is $\pred{A}_1\ i\ u$, and $a_2$ is 
$\pred{A}_2\ i\ u$.
\vspace{-5pt}
\begin{tabbing}
 $\tcomp(u_b.u_e.i.($\=$r{:\bt}.\exists\ u, u_b{<}u{\leq} u_e
\conj \pred{A}_2\ i\ u\conj\exists j, u', u'{<}u
 \conj\pred{A}_1\ j\ u',$ 
\\\>$\top))$.
\end{tabbing}\vspace{-5pt}

We wouldn't have been able to know that $A_1$ happens before $A_2$
without the time point in the expression typing rules.

\paragraph{Logical Reasoning}
{\lang} includes a proof system for first-order logic, most of which
is standard.  We show here the rule \rulename{Honest}, which allows us
to deduce properties of a thread based on the invariant assertion of
the computation it executes.

\vspace{3pt}
\(
\inferrule*[right=Honest]{
   u_1, u_2, i; \Theta;\Sigma;\Gamma^L;\cdot; \Delta \vdash c :  \varphi
\\  \Theta;\Sigma;\Gamma^L; \cdot; \Delta \vdash \pred{start}(I, c, u)
\jtrue
\\ \Theta;\Sigma\vdash\Gamma^L, \Gamma \jok
 }{
 \Theta;\Sigma;\Gamma^L; \Gamma; \Delta  \vdash \forall u'{:}\bt.(u'{>}u)  \imp
 \varphi[u,u',I/u_1, u_2,i] \jtrue}
\)
\vspace{3pt}

If we know that a thread $\iota$ starts executing at time $u$ with
payload computation $c$ (premise $\pred{start}(\iota, c, u)$) and
computation $c$ has an invariant postcondition $\varphi$, then we can
conclude that at any later point $u'$, $\varphi$ holds for the
interval $(u,u']$. The condition that $c$ be typed under an empty
$\Gamma$ context is required by the soundness proofs, which we
discuss in Section~\ref{ssec:soundness}.

\subsection{Examples}
\label{ssec:typing-examples}
\label{ssec:formalexamples}

We prove the following state continuity property of Memoir. It
states that after the service has been initialized at time $u_i$ with
the key $\skey$, whenever we invoke the service with
$\state$ at a time point $u$, later than $u_i$,
 it must be the case that, the service
was either initialized or produced the state $\state$ at a time point
$u'$. Moreover, there is  no invokations of the service between $u'$
and $u$.
\vspace{-3pt}
\[
\begin{array}{@{}l}
~~\forall u_i, \state, \newstate, \skey, i_{init}, s_{init}\\
~~~~~~\sinit (i_{init}, \skey, \service, s_{init})@u_i \imp\\
~~~~~~~~~~\forall u>u_i.~\stry(i, \skey, \state, \newstate)@u \imp\\
~~~~~~~~~~~  \exists j, u' < u.~((\exists s. \sexec(j, \skey, s, \state)@u'\\
~~~~~~~~~~~~~~~~~~~~~~~~~~~~~~~~~~~~\lor \stry(j, \skey, \state)@u'  \\
~~~~~~~~~~~~~~~~~~~~~~~~~~~~~~~~~~~~\lor \sinit(j, \skey, \state)@u')\\
~~~~~~~~~~~~~~~~~~~~~~~~~~~\land (\forall j'.~ \neg \sexec(j', skey, \cdots)\jon(u',u)]))
\end{array}\vspace{-3pt}
\]

The expressiveness of the first-order logic enables us to specify the
above property, where the ordering of events is crucial.
For the full proofs, we refer the reader to our technical appendix.
We now revisit our discussion in Section~\ref{sec:example} and highlight critical uses
of the \lang program logic in the proof.
Recall that Memoir has two phases: service
initialization and service invocation.  During initialization,
we assume that the Memoir module $\runmodule$ (Figure~\ref{fig:snippet}) is
assigned NVRAM location $\Nloc$ and service $\service$. The
permission for accessing $\Nloc$ (which stores the secret key used to encrypt state and
the freshness tag) is set to the value of
PCR 17. This PCR stores a nested hash $\shash = H(h || \mathit{code\_hash}(\service))$.
Here, the term $H(x)$ denotes hash of $x$, $||$ denotes concatenation, $h$ is
any value and
$\mathit{code\_hash(x)}$ is a hash of the textual reification of program $x$.
After initialization, we prove the following two key invariants about executions of
$\runmodule$:

\begin{packedenumerate}

\item\emph{PCR Protection:}  The value of PCR 17 contains the
value $\shash$ only during late launch sessions running $\runmodule$.

\item\emph{Key Secrecy:} If the key corresponding to a service
is available to a thread, then it must have either generated it or read it from
$Nloc$.

\end{packedenumerate}
We prove these invariants using the \rulename{Honest} rule, which requires us to type
$\mi{runmodule}$. Since $\mi{runmodule}$ invokes $\mi{srvc}$,
we need to type $\mi{srvc}$. Recall that $\mi{srvc}$ is adversarially-supplied
code. Thus, in typing it we make use of the \rulename{CONFINE} and \rulename{EQ} rules.

For the first invariant, we derive the necessary
type for $\mi{srvc}$ by typing against the TPM interface.
The particular invariant type we wish to derive about $\mi{srvc}$ is that in a late launch
session if the value in the PCR has been set to a value that is not a prefix
of $\shash$, then $\mi{srvc}$ cannot change the value in the PCR to something that
is a prefix of $\shash$ (i.e., it cannot fool the NVRAM access control mechanism
into believing that $\mi{service}$ was loaded when it was not).

\(\vspace{3pt}
\begin{array}{@{}l}
 (\mi{srvc}~\mi{ExtendPCR}~\mi{ResetPCR}~\cdots)~(\mi{state}, \mi{req}) : \\
~~~~~~~~~~~~~~~~\tcmp(u_b,u_e,i.~\neg\pred{PCRPrefix}(\llpcr, \shash)@u_b \imp \\
~~~~~~~~~~~~~~~~~~~ \forall u \in (u_b, u_e].~(\pred{InLLSession}(u, \runmodule, i)  \\
~~~~~~~~~~~~~~~~~~~~~~~~~~~~~~~~~~~\imp \neg\pred{PCRPrefix}(\llpcr,\shash)@u)
\end{array}\vspace{3pt}
\)

To derive this type using the \rulename{Confine} rule, it is sufficient to show
that each function in the TPM interface can be assigned this type. For example,
the  $\mi{ExtendPCR}$ interface satisfies this invariant as it can only
extend a PCR value. This derivation is a key step in proving that the service
does not change the value of the PCR to a state that allows any entity other than
$\mi{runmodule}$ to read the NVRAM location $Nloc$ (i.e., the first invariant of
$\mi{srvc}$ in Section~\ref{ssec:challenges}).

Similarly, we can prove that the permissions on $Nloc$
are always tied to PCR 17 being $\shash$, by typing $srvc$ with the invariant
that the permissions on $Nloc$ cannot be changed.
Thus, whenever $\Nloc$ is read from, the value
of PCR 17 is $\shash$.
We also show separately that in any particular
instance of $\mi{runmodule}$ with $\srvc$, the state of PCR 17 must be $H(h ||
\mi{code\_hash}(\srvc))$ for some $h$. Therefore, by $\Nloc$'s access
control mechanism, we prove that $H(h||\mi{code\_hash}(\srvc)) = \shash$ and
therefore $\srvc = \service$ (where $=$ denotes syntactic equality).

This is a key step to proving the key secrecy invariant. It allows us to
transfer assumptions about the known Memoir service $service$ to
the adversarially-supplied service $srvc$. Specifically, we assume that
$service$ has the following type $\tau_{sec}$ (which means that if the input of
$service$ does not contain a secret $s$ then the output doesn't contain it) and
an invariant $\pred{KeepsSecret}(i, s, Nloc)$ (which means that $s$ is
not sent out on the network and the only NVRAM location $s$ possibly written to
is $Nloc$).

\(\vspace{3pt}
\begin{array}{@{}l}
  \tau_{sec} = \Pi i:\tmsg.~\tcmp(
  u_b, u_e, i.~\\
  ~~~~~~~(x:\tmsg. \forall s.~\neg\pred{Contains}(i, s) \imp \neg \pred{Contains}(x,s),\\
  ~~~~~~~~\forall s.~\neg \pred{Contains}(i, s) \imp \pred{KeepsSecret}(i, s, \mi{Nloc})\jon(u_b, u_e])) \\
\end{array}\vspace{3pt}
\)

Using the above assumption about $\mi{service}$ and the proof that $\mi{srvc} = \mi{service}$, we use \rulename{Eq} to derive the required type for $srvc$ (i.e., the second invariant of $\mi{srvc}$ discussed in Section~\ref{ssec:challenges}).

\section{Semantics and Soundness}
\label{sec:semantics}

We build a step-indexed semantic model~\cite{ahmed:esop06} for types
and prove soundness of \lang relative to that. Central to the
semantics is the notion of invariant. We build two sets of semantics:
one is a semanticsx for invariants of the form $u_b.u_e.i.\varphi$
($\reinv{u_b.u_e.i.\varphi}{}$), and the other is an invariant-indexed
semantics for types ($\rei{\tau}{}{u_b.u_e.i.\varphi}$). These two
sets coincide when $\cnfn{\tau}{u_b.u_e.i.\varphi}$ holds
(Lemma~\ref{lem:cnfn-confine-short}).

\subsection{A Step-indexed Semantics for Invariants}
We define $\rvinv{\Phi}{\trace;u}$, $\reinv{\Phi}{\trace;u}$,
$\rcinv{\Phi}{\trace;u}$ ($\Phi=u_b.u_e.i.\varphi$), the sets of
step-indexed normal forms, expressions, and computations that satisfy
the invariant $\varphi$ respectively.
$\trace$ is the trace that the term is evaluated on and $u$ is the
earliest time point when the term is evaluated.  These sets categorize
invariant-confined adversarial programs.

We first define the set of step-indexed computations that satisfy an
invariant $\varphi$ below. An indexed computation $(k,c)$ belongs to
this relation if the following holds: (1) during any interval $u_B$
and $u_E$ when thread $\iota$ executes $c$ on $\trace$,
$\varphi[u_B,u_E,\iota/u_b,u_e,i]$ holds on $\trace$ and (2) if $c$
completes at time $u_E$, then the expression that $c$ returns, indexed
by the remaining steps of the trace, satisfies the same invariant.
\vspace{-5pt}
\begin{tabbing}
  $\rcinv{u_b.u_e.i.\varphi }{\trace; u} =$  

\\$\{ (k,  c)~|~ $
 $\forall u_B, u_E, \iota, u\leq u_B\leq u_E,$
\\\quad\=
 let $\gamma=[u_B,u_E,\iota/u_1,u_2,i]$,
\\\>
 $j_b$ is the length of the trace from time $u_B$ to the end of
$\trace$
\\~\> $j_e$ is the length of the trace from time $u_E$ to the end of $\trace$
\\~\>$k \geq j_b > j_e$, 
\\\> the configuration at time $u_1$ is
%\\\>  ~~
$\steps{u_B}\sigma_{b}\rt \cdots, \langle\iota; x.c'::K;
c\rangle\cdots$
\\\> the configuration at time $u_E$ is
%\\\>  ~~
$\steps{u_E}\sigma_{e}\rt \cdots, \langle\iota; K;
c'[e'/x]\rangle\cdots$ 
\\\>between $u_B$ and $u_E$, the stack of thread $i$
%\\\> 
always contains $x.c'{::}K$
%\\\>and  
\\\> $\Longrightarrow$ \=
%$\forall u', u'\geq u_2$, 
 $(j_e, e') \in  \reinv{u_b.u_e.i.\varphi}{\trace; u_E} $
 and $\trace\vDash_\theta \varphi[e'/x]\}$ 
\\$\cap$
 $\{ (k, c)~|~ $
 $\forall u_B, u_E, \iota, u\leq u_B\leq u_E,$
 let $\gamma=[u_B,u_E,\iota/u_1,u_2,i]$,
\\\>
$j_b$ is the length of the trace from time $u_B$ to the end of
$\trace$,
\\~\> 
$j_e$ is the length of the trace from time $u_E$ to the end of $\trace$
\\~\> $k \geq j_b \geq j_e$, 
\\\> the configuration at time $u_B$ is
%\\\> 
 ~~$\steps{u_B}\sigma_{b}\rt \cdots, \langle\iota; x.c'::K; c\rangle\cdots$
\\\> between $u_B$ and $u_E$ (inclusive), the stack of 
%\\\> 
thread $i$ always 
\\\>contains prefix $x.c'{::}K$
\\\> $\Longrightarrow$
 $\trace\vDash_\theta \varphi\}$ 
 \end{tabbing}
\vspace{-5pt}

We explain some parts of the definition. 
At time $u_B$, thread
$\iota$ begins to run $c$, which is formalized by requiring that the
thread $\langle\iota; K; c\rangle$ is in the configuration right
after time $u_B$. At time $u_E$, $c$ returns an expression $e'$ to its
context, which is formalized by requiring that thread $\iota$'s top
frame is popped off the stack with $e'$ substituted for $x$, and that
the top frame remains unchanged between $u_B$ and $u_E$.
 Both $u_B$
and $u_E$ are within the last $k$ configurations of the trace because
the length of the trace is $n$ and $k \geq j_b >j_e$.  The earliest time point to interpret $e'$ is $u_E$, which
is when $e'$ is returned. The
index for the returned expression $e'$ is $j_e$, which is less
than~$k$. Hence, our step-indices count the number of remaining steps
in the trace. 
Moreover, these remaining steps include not just steps of the thread
containing $c$, but also other threads. This ensures the 
computation $c$'s postconditions hold even when it executes concurrently
with other threads (robust safety; Theorem~\ref{thm:robust}).
For the second set, $c$ must not have finished at $u_E$, so between $u_b$ and $u_e$, no
frame on the stack $x.c'::K$ should have been popped.

% As such programs are
% untyped, we can only build well-founded relations by step-indexing
% the programs.

The relation $\rvinv{u_b.u_e.i.\varphi}{\trace;u}$ includes all normal
expressions that are not introduction forms (i.e. functions and
suspended computations). These normal forms cannot be further reduced
in any evaluation context, and therefore do not have any effects (they
are silent). A function is in this relation if, given arguments
maintaining the same invariant, the function body also maintains that
invariant. As is standard, the step-index of the argument is smaller
than that of the function because function application consumes a
step. The case of polymorphic functions is defined similarly.  A
suspended computation $\ecomp(c)$ belongs to this relation if $c$
belongs to the $\rcinv{u_b.u_e.i.\varphi}{\trace;u}$ relation defined
earlier.  \vspace{-5pt}
\begin{tabbing}
$\rvinv{u_b.u_e.i.\varphi}{\trace;u}~=$ $\{(k, \nf)~|~$
$\nf \neq \lambda x.e, \Lambda X.e, \ecomp(c)\} $
\\$\cup \{(k,\ecomp(c))~|~$
$(k, c)\in \rcinv{u_b.u_e.i.\varphi}{\trace;u}
\}$
\\ $\cup \{(k,\lambda x.e')~|~$\=$
 \forall j, u', j<k, u'\geq u$
\\\>
$ (j,e')\in\reinv{u_b.u_e.i.\varphi}{\trace;u'}$ 
\\\>
$\Longrightarrow(j,
e[e'/x])\in\reinv{u_b.u_e.i.\varphi}{\trace;u'}\}$
\\ $\cup$
$\{(k,\Lambda x.e)~|~\forall j, j<k$
$\Longrightarrow(j, e)\in\reinv{u_b.u_e.i.\varphi}{\trace;u}
\}$
\end{tabbing}\vspace{-5pt}

The definition of the $\reinv{u_b.u_e.i.\varphi}{\trace;u}$ relation
is standard: if $e$ evaluates to a normal form $\nf$ in $m$ steps,
then $\nf$ has to be in the value relation indexed by the number of
the remaining steps.
\vspace{-5pt}
\begin{tabbing}
$\reinv{u_b.u_e.i.\varphi}{\trace;u}$=
\\$\{(k, e)~|$\=
$\forall 0\leq m \leq k, e\rightarrow^{m} e'\nrightarrow$
\\\>$\Longrightarrow (n-m, e')\in \rvinv{u_b.u_e.i.\varphi}{\trace;u}\}$
\end{tabbing}\vspace{-5pt}

This relation includes all programs (including ill-typed ones) that
satisfy the invariant if executed in a context that satisfies that
invariant. This relation justifies the soundness of
\rulename{Confine} rule. Confined adversary-supplied code is in the
$\reinv{u_b.u_e.i.\varphi}{\trace; u}$ relation (Lemma~\ref{lem:inv-confine-short}).

\subsection{A Step-indexed Model for Types}
\label{ssec:model}
As programs include adversarial code, which requires its evaluation
context to maintain an invariant, the semantics of types need to be
indexed by invariants of the form $u_b.u_e.i.\varphi$.

\Paragraph{Types}
The interpretation of an expression type $\tau$ is a semantic type,
written $\cset$. Each $\cset$ is a set of pairs; each pair contains a
step-index and an expression. The expression has to be in normal form,
denoted $\nf$, that cannot be reduced further under call-by-name
$\beta$-reduction. $\cset$ contains the set of all possible indices
and all syntactically well-formed normal forms. This is used to
interpret the type $\texp$ of untyped programs.  As usual, we require
that $\cset$ be closed under reduction of step-indices.  Let
$\pset{S}$ denote the powerset of $S$. The set of all semantic types
is denoted $\ubertype$.  \vspace{-5pt}
\begin{tabbing}
% Values~~~$v$~$::=$~$ \btrue ~|~ \bfalse ~|~ \iota ~|~ \ell ~|~ n ~|~ \lambda x.e ~|~\Lambda X.e
% ~|~ \ecomp(c)$
% \\
 $\ubertype$ $\defeq$\= $\{\cset ~|~$\= $\cset\in \pset{\{(j, \nf) ~|~
    j\in\nat\}}\conj$
\\\>\> $(\forall k, \nf, (k, \nf)\in \cset \conj
j < k \Longrightarrow (j, \nf)\in \cset) \conj$
\\\>\> $(\forall k, \nf, \nf\neq\lambda x.e, \Lambda X.e, \ecomp(e)
\Longrightarrow (j, \nf)\in \cset)\}$
%\\\> $\mathrel{\cup} \{\{(j, e) ~|~ j \in \mathbb{N}\}\}$
\end{tabbing}\vspace{-5pt}

\Paragraph{Interpretation of expression types}
We define the \emph{value and expression interpretations} of
expression types $\tau$ (written
$\rvi{\tau}{\theta;\trace;u}{\Phi}$ and
$\rei{\tau}{\theta;\trace;u}{\Phi}$), as well as the
\emph{interpretation} of computation types $\eta$ (written
$\rci{\eta}{\theta;\trace;u}{\Phi}$) simultaneously by
induction on types ($\Phi=u_b.u_e.i.\varphi$).
Let $\theta$ denote a partial map from type variables to $\ubertype$,
$\trace$ denote the trace that expressions are evaluated on, and $u$
denote the time point after which expressions are evaluated.
Figure~\ref{fig:re-inv-rel} defines the value and expression
interpretations. We omit the cases for $\texp$ and $X$.

\begin{figure*}[t!]
\(
\begin{array}{l@{~}c@{~}l}
%  \rvi{\texp}{\theta;\trace;u}{u_b.u_e.i.\varphi} &=&  \{(k, \nf) ~|~ k \in \mathbb{N} \}
% \\
%  \rvi{X}{\theta;\trace;u}{u_b.u_e.i.\varphi} &= &  \theta(X) 
% \\
\rvi{\bt}{\theta;\trace;u}{u_b.u_e.i.\varphi} &=&
\{(k, e)~|~ (k,e)\in\rvinv{u_b.u_e.i.\varphi}{\theta;\trace;u}\}
 \\
\rvi{\Pi x{:}\tau_1. \tau_2}{\theta;\trace;u}{u_b.u_e.i.\varphi} &=&
 \{(k,\lambda x.e) ~|~ 
\forall j< k,  \forall u', u'\geq u,
\forall e', (j,  e')  \in \rei{\tau_1}{\theta;\trace;u'}{u_b.u_e.i.\varphi}
\\& &\qquad\qquad
 \Longrightarrow (j, e_1[e'/x]) 
 \in\rei{\tau_2[e'/x]}{\theta;\trace;u'}{u_b.u_e.i.\varphi} \}\cup
\\ & &   
 \{(k,\nf) ~|~  \nf \neq \lambda x.e
 \Longrightarrow
 (k, \nf)\in\reinv{ u_b.u_e.i.\varphi}{\trace;u}\}
\\
 \rvi{\forall X.\tau}{\theta;\trace;u}{u_b.u_e.i.\varphi} &=& 
 \{(k,\Lambda X) ~|~ 
\forall j < k, \forall
 \cset\in \ubertype 
\Longrightarrow  (j, e')\in
 \rei{\tau}{\theta[X\mapsto\cset];\trace;u}{u_b.u_e.i.\varphi}\}\cup
\\ & &   
 \{(k,\nf) ~|~ \nf\neq\Lambda X.e
\Longrightarrow
 (k, \nf)\in\reinv{u_b.u_e.i.\varphi}{\trace;u}\}
\\
\multicolumn{3}{l}{\rvi{\tcomp(\tpair{u_1.u_2.i}{x{:}\tau.\varphi_1}{\varphi_2})}{\theta;\trace;u}{u_b.u_e.i.\varphi}~
=} 
\\\multicolumn{3}{l}{\quad\{(k,\ecomp(c))~|~
 \forall u_B, u_E, \iota, u\leq u_B\leq u_E,
\mbox{let}~\gamma=[u_B,u_E,\iota/u_1,u_2,i]}
\\\multicolumn{3}{l}{\qquad\qquad\qquad\quad
(k, c)\in 
\rci{x{:}\tau\gamma.\varphi_1\gamma}
{\theta;\trace;u_B, u_E, \iota}{ u_b.u_e.i.\varphi}
\cap \rci{\varphi_2\gamma}{\theta;\trace;u_B, u_E, \iota}{\_}
\}\cup}
\\\multicolumn{3}{l}{\quad
 \{(k,\nf) ~|~  \nf \neq\ecomp(c)
 \Longrightarrow
 (k, \nf)\in\reinv{u_1.u_2.i.\varphi}{\trace;u}\}
}
\\\\
\multicolumn{3}{l}{\rei{\tau}{\theta;\trace;u}{u_b.u_e.i.\varphi}~
=\{(k, e)~|~ 
\forall j < m, 
e\rightarrow^m_\beta e'\nrightarrow \Longrightarrow 
(k-m,e')\in\rvi{\tau}{\theta;\trace;u}{u_b.u_e.i.\varphi}
\}}
\end{array}
\)
\caption{Semantics for inv-indexed types}
\label{fig:re-inv-rel}
\end{figure*}

The interpretation of the base type $\bt$ is the same as
$\rvinv{\Phi}{\theta;\trace;u}$. The type $\bt$ itself
doesn't specify any effects, and, therefore, expressions in the
interpretation of $\bt$ only need to satisfy the invariant $\Phi$.
The interpretation of the function type $\Pi x{:}\tau_1.\tau_2$ is
nonstandard: the substitution for the variable $x$ is an expression,
not a value. This simplifies the proof of soundness of function
application: since {\lang} uses call-by-name $\beta$-reduction, the
reduction of $e_1\,e_2$ need not evaluate $e_2$ to a value before it
is applied to the function that $e_1$ reduces to. Further, the
definition builds-in both step-index downward closure and time delay:
given any argument $e'$ that has a smaller index $j$ and evaluates
after $u'$, which is later than $u$, the function application belongs
to the interpretation of the argument type with the index $j$ and time
point $u'$.
The interpretation of the function type also includes normal forms
that are not $\lambda$ abstractions that are in the
$\rvinv{u_b.u_e.i.\varphi}{\theta;\trace;u}$ relation. These are
adversary-supplied untyped code, which is required by our type system
to satisfy the invariant $u_b.u_e.i.\varphi$.

The interpretation of the monadic type includes suspended
computations $(k, \ecomp(c))$ such that $(k, c)$ belongs to the interpretation
of computation types, defined below. Because $c$ executes
after time $u$, the beginning and ending time points selected for
evaluating $c$ are no earlier than $u$. Similar to the interpretation
of the function type, the interpretation of the monadic type also
includes normal forms that are not monads, but satisfy the
invariant $u_b.u_e.i.\varphi$. 
The interpretation of the $\texp$ type contains all
normal forms.
%The sets $\rv{\tau}{\theta}$ are downward closed, hence, that they are all in
%$\ubertype$. %We show this in Lemma~\ref{lem:downward-closed}.

We lift the value interpretation 
$\rvi{\tau}{\theta;\trace;u}{\Phi}$
to the expression
interpretation $\rei{\tau}{\theta;\trace;u}{\Phi}$
 in a standard way.

\Paragraph{Interpretation of formulas}
Formulas are interpreted on traces. We write $\trace \vDash
\varphi$ to mean that $\varphi$ is true on trace $\trace$. 
\vspace{-3pt}
\[
\begin{array}{lcl}
\trace \vDash P\;\vec{e} & ~\mbox{iff}~ & P\;\vec{e}  \in\varepsilon(\trace)
\\ 
\trace \vDash \pred{start}(I, c, U)
 & ~\mbox{iff}~ & 
 \mbox{thread $I$ has $c$ as the active} 
\\ & &\mbox{computation with an
   empty stack} 
\\ & & \mbox{at time $U$ on $\trace$}
\\
\trace \vDash \forall x{:}\tau.\varphi & ~\mbox{iff}~ &
\forall e, e \in\interp{\tau}~\mbox{implies}~
 \trace\vDash \varphi[e/x]
% \\ 
%  \trace \vDash \exists x{:}\tau.\varphi & ~\mbox{iff}~ &
% \exists e, e\in\interp{\tau}~\mbox{and}~
%   \trace\vDash \varphi[e/x]
\end{array}\vspace{-3pt}
\]

We assume a valuation function $\varepsilon(\trace)$ that returns the
set of atomic formulas that are true on the trace $\trace$.  For
first-order quantification, we select terms in the denotation of the
types ($\interp{\tau}$), which is defined as follows:
\vspace{-3pt}
\[
\begin{array}{l@{~}c@{~}l}
 \interp{\texp} &=&  \{e ~|~ e~\mbox{is an expression}\}
\\
\interp{\bt} &=&
\{e~|~ e\rightarrow^* \const\}
 \\
\interp{\Pi x{:}\tau_1. \tau_2} &=&
 \{\lambda x.e ~|~ 
\forall e',e'  \in \interp{\tau_1}
 \Longrightarrow e_1[e'/x]
 \in\interp{\tau_2}\}
\end{array}\vspace{-3pt}
\]

The types of the logical variables can only be $\bt$, $\texp$ and
function types. The interpretation of these types is much simpler
than that of expressions.

\Paragraph{Interpretation of computation types}
The interpretation of a computation type, $\rci{ x{:}\tau.\varphi }
{\theta;\trace; \Xi}{u_b.u_e.i.\varphi_1}$, is a set of step-indexed
computations. The trace $\trace$ contains the execution of the
computation. $\Xi = u_b,u_e,i$ has its usual meaning, except that
$u_b$, $u_e$, and $i$ are concrete values, not variables.

%\noindent\framebox{$\rc{\eta}{\Delta; \Xi}$}
We define the semantics of the partial correctness type, denoted $\rci{ x{:}\tau.\varphi
}{\theta; \trace; \Xi}{u_b.u_e.i.\varphi_1}$, below.
Informally, it contains the set of indexed computations
$c$, if $\trace$ contains a complete execution of the computation $c$
in the time interval $(u_b,u_e]$ in thread $\iota$ such that $c$
returns $e'$ at time $u_e$ and it is also the case that $\trace$
satisfies $\varphi[e'/x]$ and that $e'$ has type $\tau$ semantically.
Similar to the $\rcinv{\Phi}{\trace;u}$ relation,
these remaining steps include not just steps of the thread
executing $c$, but also other threads. 
The invariant ${u_b.u_e.i.\varphi_1}$ is used in the
specification of the return value. 

\vspace{-5pt}
\begin{tabbing}
  $\rci{ x{:}\tau.\varphi }
  {\theta;\trace; u_1, u_2, i}{u_b.u_e.i.\varphi_1} =$  $\{ (k,
  c)~|~ $\\
~ \=
 $j_b$ is the length of the trace from time $u_1$ to the end of
$\trace$
\\~\> $j_e$ is the length of the trace from time $u_2$ to the end of $\trace$
\\~\>$k \geq j_b > j_e$, 
\\\> the configuration at time $u_1$ is
%\\\>  ~~
$\steps{u_1}\sigma_{b}\rt \cdots, \langle\iota; x.c'::K;
c\rangle\cdots$
\\\> the configuration at time $u_2$ is
%\\\>  ~~
$\steps{u_2}\sigma_{e}\rt \cdots, \langle\iota; K;
c'[e'/x]\rangle\cdots$ 
\\\>between $u_1$ and $u_2$, the stack of thread $i$
%\\\> 
always contains $x.c'{::}K$
%\\\>and  
\\\> $\Longrightarrow$ \=
%$\forall u', u'\geq u_2$, 
 $(j_e, e') \in  \rei{\tau}{\theta;\trace; u_2}{u_b.u_e.i.\varphi_1} $
\\\>\>  and $\trace\vDash \varphi[e'/x]\}$ 
 \end{tabbing}
\vspace{-5pt}

The interpretation for the invariant assertions is defined similarly,
and we omit its definition.  Because $c$ is being evaluated and
produces no return value, the interpretation need not be indexed by an
invariant. We write \_ in place of the invariant.

%%%%%%%%%%%%%%%%%%%%%%%%%%%%%%%%%%%%%%%

%%%%%%%%%%%%%%%%%%%%%%%%%%%%%%%%%%%%%%%

\subsection{Examples}
\label{ssec:semantics-examples}

We illustrate some key points of our semantic model. We
instantiate the {\next} function (Section~\ref{sec:language}) for the
\reada %and \writea
action as follows:
\vspace{-5pt}
\[
\next(\sigma, \reada\ e_1\ e_2) =
\left\{\begin{array}{ll}
 (\sigma, \sigma(\ell))  & \ell\in\dom(\sigma) 
\\
(\sigma, \stuck) & \ell\notin\dom(\sigma)
\end{array}
\right.
\]\vspace{-5pt}

Predicate
$\pred{stuck}\ \iota\ u$ is true when thread $\iota$ is in the
$\stuck$ state at time $u$. 
The first example below shows the semantic specification of the \reada
action. The partial correctness assertion states that as long as the
location $l$ being read is allocated when the \reada happens, the
thread does not get stuck and the expression $y$ returned by \reada is
the in-memory content $v$ of the location read. The
invariant assertion states that between the time the \reada action
becomes the redex and the time it reduces, the thread is not stuck.

\vspace{-3pt}
\begin{packedenumerate}
 \item \begin{tabbing}
 $(n, \cact(\reada\ e))\in$ 
 \\ $\rci{y{:}\texp.
   \forall l, v,$\= $\pred{mem}\ l\ v\ u_2 
   \conj \pred{eval}\ e\ l
   \imp$
\\\>
   $(y = e)\conj\neg\pred{stuck}\ i @(u_1, u_2]}{\theta; \trace; u_1,u_2,i}{\Phi}$
 \end{tabbing}

\item $\rci{\forall j, l, e, t. (\neg\pred{Write}\ j\ l\ e\ t)}
{\theta;\trace; u_1, u_2, i}{\Phi} = \emptyset$ 
\end{packedenumerate}

The second example states that
the interpretation of the invariant computation type
($\forall j, l, e, t. (\neg\pred{Write}\ j\ l\ e\ t)$), which
states that no thread performs a write action at any time, is the empty
set. This is because
the semantics of invariant assertions require that \emph{any} trace
containing the execution of such a computation satisfy this
invariant. A trivial counterexample is a trace containing a second
thread that writes to memory.

\subsection{Soundness of the Logic}
\label{ssec:soundness}

We prove that our type system is sound relative to the semantic model
of Section~\ref{ssec:model}. 
We start by defining valid substitutions for contexts.  We
write $\rd{\Theta}{}$ to denote the set of valid semantic
substitutions for $\Theta$.
We write $\rgi{\Gamma}{\theta;\trace;u}{\Phi}$ to denote a set of
substitutions for variables in $\Gamma$.  
Each indexed substitution is
a pair of an index and a substitution $\gamma$ for 
variables.

We first prove two key lemmas. Lemma~\ref{lem:cnfn-confine-short}
states that when all the effects in $\tau$ are $u_b.u_e.i.\varphi$,
then the interpretation of $\tau$ is the same as the interpretation of
the invariant $u_b.u_e.i.\varphi$. The proof is by induction on the
structure of $\tau$.

\begin{lem}[Indexed types are confined]
~\label{lem:cnfn-confine-short}
$\cnfn{\tau}{u_b.u_e.i.\varphi}$ implies
$\rei{\tau}{\theta;\trace;u}{u_b.u_e.i.\varphi} =
\reinv{u_b.u_e.i.\varphi}{\trace;u}$. 
\end{lem}

The following lemma states that if $e$
does not contain any actions, then $e$, with its free variables
substituted by expressions that satisfy an invariant
$u_b.u_e.i.\varphi$, satisfies the same invariant. The proof is by
induction on the structure of $e$. 

\begin{lem}[Invariant confinement]~\label{lem:inv-confine-short}
If $\varphi$ is composable, 
and  thread $\iota$ silent between time $u_B$
and $u_E$ implies $\trace\vDash\varphi[u_B,u_E,I/u_b,u_e,i]$,
then $\fa(e)=\emptyset$, $\fv(e)\in\dom(\gamma)$,
and $(n,\gamma)\in\reinv{u_b.u_e.i.\varphi}{\trace;u}$
imply $(n, e\gamma)\in\reinv{u_b.u_e.i.\varphi}{\trace;u}$.

\end{lem}

The soundness theorem (Theorem~\ref{thm:soundness-short}) has two
different statements for judgements with the empty qualifier and the
invariant qualifier. The ones for judgments with an empty qualifier
state that for any invariant $\Phi$, if the substitution for $\Gamma$
belongs to the interpretation of types, then the expression
(computation) belongs to the interpretation of its type, indexed by
the same invariant $\Phi$. For judgments qualified by a specific
invariant $\Phi$, the soundness theorem statements are also specific
to that $\Phi$.

\dg{Use of bullets in this theorem is awkward. Each bullet is not an independent statement.}

\begin{thm}[Soundness]
~\label{thm:soundness-short}

Assume that $\forall A::\alpha \in\Sigma$, 
$\forall \Phi, \trace, n, u, (n,
A)\in\rai{\alpha}{\cdot;\trace;u}{\Phi}$, 
\begin{enumerate}
%%%%%%%%%%%%%%%%%%%%%%%%%%%%%%%%%%%%%%%%%%%%%%
%% |-varphi e:t  (hat(re)[[]])
%%%%%%%%%%%%%%%%%%%%%%%%%%%%%%%%%%%%%%%%%%%%%%
\item %$\varphi$ is trace composable, and holds on silent threads,
% let $\Phi = u_b.u_e.i.\varphi$,
\begin{enumerate}
%%% e
\item  
 $ \ee::u:\bt; \Theta;\Sigma;\Gamma^L;\Gamma;\Delta \vdash_\Phi e : \tau$, 
 $\forall \theta\in\rd{\Theta}{}$, 
 $\forall \gamma^L\in\interp{\Gamma^L}$,
 $\forall U, U', U'\geq U$, let $\gamma_u=[U/u]$, 
 $\forall \trace$,
 $\forall n, \gamma$,   $(n; \gamma)\in \rgi{\Gamma\gamma_u\gamma^L}{\theta;\trace;U'}{\Phi}$,
 $\trace\vDash\Delta\gamma\gamma_u\gamma^L$
 implies $(n; e\gamma) \in \rei{\tau\gamma\gamma_u\gamma^L}
 {\theta;\trace;U'}{\Phi}$
%%% c
\item 
 $ \ee:: u_1, u_2, i; \Theta;\Sigma;\Gamma^L;\Gamma;\Delta \vdash_\Phi c : \eta$, 
 $\forall$ $u$, $u_B$, $u_E$, $\iota$ s.t. $u\leq u_B\leq u_E$,  
let $\gamma_1 = [u_B, u_E,\iota/u_1,u_2,i]$
 $\forall \theta\in\rd{\Theta}{}$, 
 $\forall \gamma^L\in\interp{\Gamma^L}$,
 $\forall \trace$,
   $\forall n, \gamma, (n; \gamma)\in
 \rgi{\Gamma\gamma_1\gamma^L}{\theta;\trace;u}{\Phi}$,
 $\trace\vDash\Delta\gamma\gamma_1\gamma^L$ 
implies $(n; c\gamma) \in \rci{{\eta}\gamma\gamma_1\gamma^L}{\theta;
   \trace; u_B,u_E,\iota}{\Phi}$
\end{enumerate}
\item 
\begin{enumerate}
\item  
 $ \ee::u:\bt; \Theta;\Sigma;\Gamma^L;\Gamma;\Delta \vdash e : \tau$, 
 $\forall \theta\in\rd{\Theta}{}$, 
 $\forall \gamma^L\in\interp{\Gamma^L}$,
 $\forall U, U', U'\geq U$, let $\gamma_u=[U/u]$, 
  $\forall \trace$, $\forall\Phi$,
 $\forall n, \gamma$,   $(n; \gamma)\in \rgi{\Gamma\gamma_u\gamma^L}{\theta;\trace;U'}{\Phi}$,
 $\trace\vDash\Delta\gamma\gamma_u\gamma^L$
 implies $(n; e\gamma) \in \rei{\tau\gamma\gamma_u\gamma^L}
 {\theta;\trace;U'}{\Phi}$
%%% c
\item 
 $ \ee:: u_1, u_2, i; \Theta;\Sigma;\Gamma^L;\Gamma;\Delta \vdash c : \eta$, 
 $\forall$ $u$, $u_B$, $u_E$, $\iota$ s.t. $u\leq u_B\leq u_E$,  
let $\gamma_1 = [u_B, u_E,\iota/u_1,u_2,i]$
 $\forall \theta\in\rd{\Theta}{}$, 
 $\forall \gamma^L\in\interp{\Gamma^L}$,
 $\forall \trace$, $\forall\Phi$,
   $\forall n, \gamma, (n; \gamma)\in
 \rgi{\Gamma\gamma_1\gamma^L}{\theta;\trace;u}{\Phi}$,
 $\trace\vDash\Delta\gamma\gamma_1\gamma^L$ 
implies $(n; c\gamma) \in \rci{{\eta}\gamma\gamma_1\gamma^L}{\theta;
   \trace; u_B,u_E,\iota}{\Phi}$
%%% phi true
\item 
 $\ee::\Theta;\Sigma;\Gamma^L;\Gamma;\Delta\vdash \varphi
  \jtrue$,
 $\forall \theta\in\rd{\Theta}{}$, 
 $\forall \gamma^L\in\interp{\Gamma^L}$,
 $\forall \trace$, $\forall \Phi$, 
   $\forall n, \gamma,u, (n; \gamma)\in
  \rgi{\Gamma\gamma^L}{\theta;\trace;u}{\Phi}$,
 $\trace \vDash \Delta\gamma^L\gamma$
 implies
 $\trace \vDash \varphi\gamma^L\gamma$
\end{enumerate}
\end{enumerate}
\end{thm}

We prove the soundness theorem by induction on typing derivations and
a subinduction on step-indices for the case of fixpoints.

The proof of soundness of the rule \rulename{Confine} (\textit{2.(a)}) first uses
Lemma~\ref{lem:cnfn-confine-short} to show that a substitution
$\gamma$ for $\Gamma$, where $\gamma$ maps each variable in $\Gamma$ to
the type interpretation of $\Gamma(x)$ is also a substitution
where  $\gamma(x)$ belongs to the interpretation of the invariant. Then
we use Lemma~\ref{lem:inv-confine-short} to show that the untyped term $e\gamma$
belongs to the interpretation of the invariant. Applying
Lemma~\ref{lem:cnfn-confine-short}  again, we can show that $e\gamma$
is in the interpretation of $\tau$. The $\mathit{confine}$ relations in
the premises are key to this proof.
The proof of the rule \rulename{Conf-Sub} uses the induction hypothesis
directly: a derivation with an empty qualifier can pick substitutions
with any invariant $\varphi$.

To prove the soundness of \rulename{Honest}, we need to show that
given any substitution $(n,\gamma)$ for $\Gamma$, the trace satisfies
the invariant of $c$. From the last premise of \rulename{Honest}, we
know that $c$ starts with an empty stack. $c$ can never return because
there is no frame to be popped off the empty stack. Therefore, at any
time point after $c$ starts, the invariant of $c$ should
hold. However, the length of the trace after $c$ starts, denoted $m$,
is not related to $n$. To use the induction hypothesis, we need to use
substitution $(m,\gamma)$ for $\Gamma$. 
Because $\Gamma$ is empty, we
complete the proof by using the induction hypothesis on the first
premise given an empty substitution $(m,\cdot)$.

An immediate corollary of the soundness theorem is the following
robust safety theorem, which states that the invariant assertion of a
computation $c$'s postcondition holds even when $c$ executes
concurrently with other threads, including those that are
adversarial. The theorem holds because we account for adversarial
actions in the definition of
$\rci{\eta}{\theta;\Delta;\Xi}{u_b.u_e.i.\varphi}$. A similar theorem
holds for partial correctness assertions.

\begin{thm}[Robust safety]\label{thm:robust}
If
\begin{packeditemize}
\item $u_1, u_2, i;  \Delta \vdash c : \varphi$, $\trace\vDash \Delta $,
\item $\trace$ is a trace obtained by executing the parallel
  composition of threads of ID ($\iota_1$, .. $\iota_k$),
\item at time $U_b$, the computation that thread $\iota_j$ is about to run is $c$
\item at time $U_e$, $c$ has not returned
\end{packeditemize}
 then $\trace\vDash \varphi [U_b, U_e, \iota_j/u_1, u_2, i]$.
\end{thm}

\section{Discussion}
\label{sec:discussion}
\paragraph{Proving non-stuckness}
We can use {\lang}'s invariant assertions
to verify that a program always remains non-stuck. 
Recall the example from Section~\ref{ssec:semantics-examples}.
We can
prove non-stuckness for a computation $c$ by showing that it has
the invariant postcondition $(\neg\pred{stuck}\; i) @(u_b, u_e]$.
To complete such a proof, we would require that
all action types assert non-stuckness in their postconditions under
appropriate assumptions on the past trace. For instance, the first example in
Section~\ref{ssec:semantics-examples} states that we can assert
non-stuckness in the postcondition of the \reada action, if the memory
location being read has been allocated.

\paragraph{Choice of reduction strategy}
{\lang} uses call-by-name $\beta$-reduction for expressions, which
simplies the operational semantics as well as the soundness proofs.
Other evaluation strategies we have considered force us to use
$\beta$-equality in place of syntactic equality in \rulename{Eq}. This
makes the system design, semantics, and soundness proofs very
complicated. In particular, the \rulename{Eq} rule that uses
$\beta$-equality cannot be proven sound in a model where expressions
are indexed by their reduction steps.

\section{Related Work}
\label{sec:related}
\Paragraph{Hoare Type Theory (HTT)}
In HTT~\cite{nanevski:jfp08,nanevski:esop07,nanevski:tldi09}, a monad
classifies effectful computations, and is indexed by the return type,
a pre-condition over the (initial) heap, and a postcondition over the
initial and final heaps. This allows proofs of functional correctness
of higher-order imperative programs. The monad in {\lang} is motivated
by, and similar to, HTT's monad.  However, there are several
differences between {\lang}'s monad and HTT's monad.  A {\lang}
postcondition is a predicate over the entire execution trace, not just
the initial and final heaps as in HTT. It also includes an invariant
assertion which holds even if the computation does not return. This
change is needed because we wish to prove safety properties, not just
properties of heaps. Although moving from predicates over heaps to
predicates over traces in a sequential language is not very difficult,
our development is complicated because we wish to reason about robust
safety, where adversarial, potentially untyped code interacts with
trusted code. Hence, we additionally incorporate techniques to reason
about untyped code (rules \rulename{Eq} and \rulename{Confine}).  We
also exclude standard Hoare pre-conditions from computation
types. Usually, pre-conditions ensure that well-typed programs do not
get stuck. We argued in Section~\ref{sec:discussion} that in {\lang}
this property can be established for individual programs using only
invariant postconditions. The standard realizability semantics of
HTT~\cite{peterson:esop08} are based on a model of CPOs, whereas our
model is syntactic and step-indexed~\cite{ahmed:esop06}.

RHTT~\cite{nanevski:oakland11} is a relational extension of HTT used
to reason about access and information flow properties of
programs. That extension to HTT is largely orthogonal to ours and the
two could potentially be combined into a larger framework with
capabilities of both. The properties that can be proved with RHTT and
System M are different. System M can verify safety properties in the
presence of untyped adversaries; RHTT verifies relational, non-trace
properties assuming fully typed adversaries.

\Paragraph{{\lss} and PCL}
{\lang} is inspired by and based upon a prior program logic, {\lss},
for reasoning about safety properties of first-order order programs in
the presence of adversaries~\cite{garg10:ls2}. The main conceptual
difference from {\lss} is that in {\lang} trusted and untrusted
components may exchange code and data, whereas in {\lss} this
interface is limited to data. Our \rulename{Confine} rule for
establishing invariants of an unknown expression from invariants of
interfaces it has access to is based on a similar rule called RES in
{\lss}. The difference is that {\lang}'s rule allows typing
higher-order expressions, which makes it more complex, e.g., we must
index the typing derivations with invariants and define
interpretations for invariants based on step-indexing programs to
obtain soundness. {\lss} itself is based on a logic for reasoning
about Trusted Computing Platforms~\cite{dfgk09:ls2} and Protocol
Composition Logic (PCL) for reasoning about safety properties of
cryptographic protocols~\cite{datta07:entcs}.

\Paragraph{Rely-guarantee reasoning}
There are two broad kinds of techniques to prove invariants over state
shared by concurrent programs. \emph{Coarse-grained} reasoning
followed in, e.g., Concurrent Separation Logic
(CSL)~\cite{brookes07:csl} and the concurrent version of
HTT~\cite{nanevski:tldi09}, assumes clearly marked critical regions
and allows programs to violate invariants on shared state only within
them. This assumes that resource contention is properly synchronized,
which is generally unrealistic when executing concurrently with an
unspecified adversary.  In contrast, \emph{fine-grained} reasoning
followed in, e.g., the method of Owicki-Gries~\cite{owicki:ai76} and
its successor, rely-guarantee reasoning~\cite{jones83:rely}, makes no
synchronization assumption, but has a higher proof burden at each
individual step of a computation. In proofs with {\lang}, including
the Memoir example in this paper, we use a template for
rely-guarantee reasoning taken from {\lss}. The methods used to prove
invariants within this template are different because of the new
higher-order setting.

\Paragraph{Type systems that reason about adversary-supplied code}
The idea of using a non-informative type, $\texp$, for
typing expressions obtained from untrusted sources goes back to the
work of Abadi~\cite{Abadi:1999:STS}. Gordon and
Jeffrey develop a very widely used proof
technique for proving robust safety based on this type~\cite{Gordon:2003:ATS}.
In their system, any program can be \emph{syntactically} given the type
$\texp$ by typing all subexpressions of the program $\texp$.
Although {\lang}'s use of the $\texp$ type is similar, our
proof technique for robust safety is different. It is \emph{semantic}
and based on that in PCL---we allow for arbitrary adversarial
interleaving actions in the semantics of our computation types
(relation $\rci{\eta}{\theta;\trace;\Xi}{\Phi}$ in
Section~\ref{ssec:model}). Due to this generalized semantic
definition, robust safety (Theorem~\ref{thm:robust}) is again a
trivial consequence of soundness (Theorem~\ref{thm:soundness-short}).

Several type systems for establishing different kinds of safety
properties build directly or indirectly on the work of
Abadi~\cite{Abadi:1999:STS} and Gordon and
Jeffrey~\cite{Gordon:2003:ATS}. Of these, the most recent and advanced
are RCF~\cite{bengtson11:rcf} and its
extensions~\cite{bhargavan10:rcf,swamy11:fstar}.
RCF is based on types
refined with logical assertions, which provide roughly the same
expressiveness as {\lang}'s dependently-typed computation types.
By
design, RCF's notion of trace is monotonic: the trace is an unordered
\emph{set} of actions (programmer specified ghost annotations) that
have occurred in the past~\cite{Fournet:2007:TDA}. This simplified
design choice allows scalable implementation.
On the other
hand, there are safety properties of interest that rely on the order
of past events and, hence, cannot be directly represented in RCF's
limited model of traces. An example of this kind is measurement
integrity in attestation protocols~\cite[Theorems~2~\&~4]{dfgk09:ls2}.
In contrast to RCF, we designed {\lang} for verification of general
safety properties (so the measurement integrity property can be
expressed and verified in {\lang}), but we have not considered
automation for {\lang} so far.

{\fstar}~\cite{swamy11:fstar} extends F7 with quantified types, a rich
kinding system, concrete refinements and several other features taken
from the language Fine~\cite{Swamy:2010:Fine}. This allows
verification of stateful authorization and information flow properties
in {\fstar}. Quantified predicates can also be used for full
functional specifications of higher-order programs. Although we have
not considered these applications so far, we believe that {\lang} can
be extended similarly.

The main novelty of {\lang} compared to the above mentioned line of
work lies in the \rulename{Eq} and \rulename{Confine} rules that
statically derive computational effects of untyped adversary-supplied
code.

Code-Carrying Authorization (CCA)~\cite{cca:esorics08} is another
extension to~\cite{Gordon:2003:ATS} that enforces authorization
policies. CCA introduces dynamic type casts to allow untrusted code to
construct authorization proofs (e.g., Alice can review
paper number 10).  The language runtime uses
logical assertions made by trusted programs to constructs proofs
present in the type cast. 
The soundness of
type cast in CCA relies on the fact that untrusted code cannot make
any assertions and that it can only use those made by trusted code.
Similar to CCA, \lang also
assigns untrusted code descriptive types.
CCA checks those types at runtime; whereas the \rulename{Confine}
rule assigns types statically.

\Paragraph{Verification of TPM and Protocols based on TPM}
Existing work on verification of TPM APIs and protocols relying on TPM
APIs uses a variety of
techniques~\cite{Delaune:2011:TPM:csf,Delaune:2010:TPM:fast,Chen:2009:ASV,Gurgens:2007:SES,dfgk09:ls2}. Gurgens
et al. uses automaton to model the transitions of TPM
APIs~\cite{Gurgens:2007:SES}.
Several results~\cite{Delaune:2011:TPM:csf,Delaune:2010:TPM:fast,Chen:2009:ASV}
use the automated tool Proverif~\cite{proverif}. Proverif translates
protocols encoded in Pi calculus into horn clauses. To check security
properties such as secrecy and correspondence, the tool runs a
resolution engine on these horn clauses and clauses representing an Dolev-Yao attacker.  Proverif
over-approximates the protocol states and works with a monotonic set of
facts. Special techniques need to be applied to use Proverif to
analyze stateful protocols such as ones that use TPM
PCRs~\cite{Delaune:2011:TPM:csf}.  \lang is more expressive: it can
model and reason about higher-order functions and programs that invoke
adversary-supplied code. Reasoning about shared non-monotonic state is
possible in \lang. However, verification using \lang requires
manual proofs.  It is unclear whether our Memoir case study can be
verified using the techniques introduced in~\cite{Delaune:2011:TPM:csf},
as it requires reasoning about higher-order code.

A proof of safety formalized in TLA+ ~\cite{tlaplus} was presented in
the Memoir paper~\cite{memoir}. They showed that
Memoir's design refines an obviously safe specification that cannot be
rolled back thus implying the state integrity property we prove. However, this
proof assumes that the service being protected is a constant action with no effects. Consequently, they do not
need to reason about the service program being changed or causing unsafe effects.
Our proofs assume a more realistic model requiring that the identity of
the service be proven and that the effects of the service be analyzed based on
the sandbox provided by the TPM.

\section{Conclusion}
\label{sec:conclusion}

{\lang} is a program logic for proving safety properties of programs
that may execute adversary-supplied code with some
precautions. {\lang} generalizes Hoare Type Theory with invariant
assertions, and adds two novel typing rules---\rulename{Eq} and
\rulename{Confine}---that allow typing adversarial code using
reasoning in the assertion logic and assumptions about the code's
sandbox, respectively. We prove soundness and robust safety relative to a
step-indexed, trace model of computations. Going further, we would
like to build tools for proof verification and automatic deduction in
{\lang}.

\bibliographystyle{abbrvnat}

\appendix
\section{Semantics}
% \begin{tabbing}
% % Values~~~$v$~$::=$~$ \btrue ~|~ \bfalse ~|~ \iota ~|~ \ell ~|~ n ~|~ \lambda x.e ~|~\Lambda X.e
% % ~|~ \ecomp(c)$
% % \\
%  $\ubertype$ $\defeq$\= $\{\cset ~|~$\= $\cset\in \pset{\{(j, \nf) ~|~
%     j\in\nat\}}\conj$
% \\\>\> $(\forall k, \nf, (k, \nf)\in \cset \conj
% j < k \conj \nf \Longrightarrow (j, \nf)\in \cset)\}$
% %\\\> $\mathrel{\cup} \{\{(j, e) ~|~ j \in \mathbb{N}\}\}$
% \end{tabbing}

\paragraph{Semantics for invariant properties}
Next we define a logical relation indexed only by an invariant
property $u_b.u_e.i.\varphi$. 
\begin{tabbing}
$\rvinv{u_b.u_e.i.\varphi}{\trace;u}~=$ $\{(k, \nf)~|~$
$\nf \neq \lambda x.e, \Lambda X.e, \ecomp(c)\} $
\\$\cup \{(k,\ecomp(c))~|~$
$(k, c)\in \rcinv{u_b.u_e.i.\varphi}{\trace;u}
\}$
\\ $\cup \{(k,\lambda x.e')~|~$\=$
 \forall j, u', j<k, u'\geq u$
\\\>
$ (j,e')\in\reinv{u_b.u_e.i.\varphi}{\trace;u'}$ 
\\\>
$\Longrightarrow(j,
e[e'/x])\in\reinv{u_b.u_e.i.\varphi}{\trace;u'}\}$
\\ $\cup$
$\{(k,\Lambda x.e)~|~\forall j, j<k$
$\Longrightarrow(j, e)\in\reinv{u_b.u_e.i.\varphi}{\trace;u}
\}$
\end{tabbing}

\begin{tabbing}
$\reinv{u_b.u_e.i.\varphi}{\trace;u}$=
\\$\{(k, e)~|$\=
$\forall 0\leq m \leq k, e\rightarrow^{m} e'\nrightarrow$
\\\>$\Longrightarrow (n-m, e')\in \rvinv{u_b.u_e.i.\varphi}{\trace;u}\}$
\end{tabbing}

\begin{tabbing}
  $\rcinv{u_b.u_e.i.\varphi }{\trace; u} =$  
 $\{ (k, c)~|~ $\\
~~ \=  
 $\forall u_B, u_E, \iota, u\leq u_B\leq u_E,$
 let $\gamma=[u_B,u_E,\iota/u_1,u_2,i]$,
\\\>
$j_b$ is the length of the trace from time $u_B$ to the end of
$\trace$,
\\~\> 
$j_e$ is the length of the trace from time $u_E$ to the end of $\trace$
\\~\> $k \geq j_b \geq j_e$, 
\\\> the configuration at time $u_B$ is
%\\\> 
 ~~$\steps{u_B}\sigma_{b}\rt \cdots, \langle\iota; x.c'::K; c\rangle\cdots$
\\\> between $u_B$ and $u_E$ (inclusive), the stack of 
%\\\> 
thread $i$ always 
\\\>contains prefix $x.c'{::}K$
\\\> $\Longrightarrow$
 $\trace\vDash_\theta \varphi\}\cap$ 
\\$\{ (k,  c)~|~ $
 $\forall u_B, u_E, \iota, u\leq u_B\leq u_E,$
\\\quad\=
 let $\gamma=[u_B,u_E,\iota/u_1,u_2,i]$,
\\\>
 $j_b$ is the length of the trace from time $u_B$ to the end of
$\trace$
\\~\> $j_e$ is the length of the trace from time $u_E$ to the end of $\trace$
\\~\>$k \geq j_b > j_e$, 
\\\> the configuration at time $u_1$ is
%\\\>  ~~
$\steps{u_B}\sigma_{b}\rt \cdots, \langle\iota; x.c'::K;
c\rangle\cdots$
\\\> the configuration at time $u_E$ is
%\\\>  ~~
$\steps{u_E}\sigma_{e}\rt \cdots, \langle\iota; K;
c'[e'/x]\rangle\cdots$ 
\\\>between $u_B$ and $u_E$, the stack of thread $i$
%\\\> 
always contains $x.c'{::}K$
%\\\>and  
\\\> $\Longrightarrow$ \=
%$\forall u', u'\geq u_2$, 
 $(j_e, e') \in  \reinv{u_b.u_e.i.\varphi}{\trace; u_E} $
 and $\trace\vDash_\theta \varphi[e'/x]\}$ 
 \end{tabbing}

\begin{tabbing}
  $\rfinv{u_b.u_e.i.\varphi}{\trace;u} =$ \\
  $\{ (k, c)~|~ $\= 
  $ \forall e$, $(k, e)\in \reinv{u_b.u_e.i.\varphi}{\trace;u}
 \Longrightarrow $ 
\\\> $(k, c\;
  e)\in$ \=$\rcinv{u_b.u_e.i.\varphi}{\trace;u}\}$ 
 \end{tabbing}

\paragraph{Semantics for invariant indexed types}
Figure~\ref{fig:re-inv-rel} summaries the interpretation of types
indexed by the invariant property $u_b.u_e.i.\varphi$. The invariant
property is used to constrain the behavior of expressions that
evaluate to normal forms that do not agree with their types. 

\begin{figure*}[t!]
\[
\begin{array}{l@{~}c@{~}l}
 \rvi{\texp}{\theta;\trace;u}{u_b.u_e.i.\varphi} &=&  \{(k, \nf) ~|~ k \in \mathbb{N} \}
\\
 \rvi{X}{\theta;\trace;u}{u_b.u_e.i.\varphi} &= &  \theta(X) 
\\
\rvi{\bt}{\theta;\trace;u}{u_b.u_e.i.\varphi} &=&
\{(k, e)~|~ (k,e)\in\rvinv{u_b.u_e.i.\varphi}{\theta;\trace;u}\}
 \\
\rvi{\Pi x{:}\tau_1. \tau_2}{\theta;\trace;u}{u_b.u_e.i.\varphi} &=&
 \{(k,\lambda x.e) ~|~ 
\forall j< k,  \forall u', u'\geq u,
\forall e', (j,  e')  \in \rei{\tau_1}{\theta;\trace;u'}{u_b.u_e.i.\varphi}
\\& &\qquad\qquad
 \Longrightarrow (j, e_1[e'/x]) 
 \in\rei{\tau_2[e'/x]}{\theta;\trace;u'}{u_b.u_e.i.\varphi} \}\cup
\\ & &   
 \{(k,\nf) ~|~  \nf \neq \lambda x.e
 \Longrightarrow
 (k, \nf)\in\reinv{ u_b.u_e.i.\varphi}{\trace;u}\}
\\
 \rvi{\forall X.\tau}{\theta;\trace;u}{u_b.u_e.i.\varphi} &=& 
 \{(k,\Lambda X) ~|~ 
\forall j < k, \forall
 \cset\in \ubertype 
\Longrightarrow  (j, e')\in
 \rei{\tau}{\theta[X\mapsto\cset];\trace;u}{u_b.u_e.i.\varphi}\}\cup
\\ & &   
 \{(k,\nf) ~|~ \nf\neq\Lambda X.e
\Longrightarrow
 (k, \nf)\in\reinv{u_b.u_e.i.\varphi}{\trace;u}\}
\\
\multicolumn{3}{l}{\rvi{\tcomp(\tpair{u_1.u_2.i}{x{:}\tau.\varphi_1}{\varphi_2})}{\theta;\trace;u}{u_b.u_e.i.\varphi}~
=} 
\\\multicolumn{3}{l}{\quad\{(k,\ecomp(c))~|~
 \forall u_B, u_E, \iota, u\leq u_B\leq u_E,
\mbox{let}~\gamma=[u_B,u_E,\iota/u_1,u_2,i]}
\\\multicolumn{3}{l}{\qquad\qquad\qquad\quad
(k, c)\in 
\rci{x{:}\tau\gamma.\varphi_1\gamma}
{\theta;\trace;u_B, u_E, \iota}{ u_b.u_e.i.\varphi}
\cap \rci{\varphi_2\gamma}{\theta;\trace;u_B, u_E, \iota}{\_}
\}\cup}
\\\multicolumn{3}{l}{\quad
 \{(k,\nf) ~|~  \nf \neq\ecomp(c)
 \Longrightarrow
 (k, \nf)\in\reinv{u_1.u_2.i.\varphi}{\trace;u}\}
}
\\\\
\multicolumn{3}{l}{\rei{\tau}{\theta;\trace;u}{u_b.u_e.i.\varphi}~
=\{(k, e)~|~ 
\forall j < m, 
e\rightarrow^m_\beta e'\nrightarrow \Longrightarrow 
(k-m,e')\in\rvi{\tau}{\theta;\trace;u}{u_b.u_e.i.\varphi}
\}}
\end{array}
\]
\caption{Semantics for inv-indexed types}
\label{fig:re-inv-rel}
\end{figure*}

\begin{tabbing}
  $\rci{ x{:}\tau.\varphi }
  {\theta;\trace; u_1, u_2, i}{u_b.u_e.i.\varphi_1} =$  $\{ (k,
  c)~|~ $\\
~ \=
 $j_b$ is the length of the trace from time $u_1$ to the end of
$\trace$
\\~\> $j_e$ is the length of the trace from time $u_2$ to the end of $\trace$
\\~\>$k \geq j_b > j_e$, 
\\\> the configuration at time $u_1$ is
%\\\>  ~~
$\steps{u_1}\sigma_{b}\rt \cdots, \langle\iota; x.c'::K;
c\rangle\cdots$
\\\> the configuration at time $u_2$ is
%\\\>  ~~
$\steps{u_2}\sigma_{e}\rt \cdots, \langle\iota; K;
c'[e'/x]\rangle\cdots$ 
\\\>between $u_1$ and $u_2$, the stack of thread $i$
%\\\> 
always contains $x.c'{::}K$
%\\\>and  
\\\> $\Longrightarrow$ \=
%$\forall u', u'\geq u_2$, 
 $(j_e, e') \in  \rei{\tau}{\theta;\trace; u_2}{u_b.u_e.i.\varphi_1} $
\\\>\>  and $\trace\vDash \varphi[e'/x]\}$ 
 \end{tabbing}

\begin{tabbing}
  $\rci{\varphi }{\theta;\trace; u_1, u_2, i}{\_} =$  $\{ (k, c)~|~ $\\
~~ \=  
$j_b$ is the length of the trace from time $u_1$ to the end of
$\trace$,
\\~\> 
$j_e$ is the length of the trace from time $u_2$ to the end of $\trace$
\\~\> $k \geq j_b \geq j_e$, 
\\\> the configuration at time $u_1$ is
%\\\> 
 ~~$\steps{u_1}\sigma_{b}\rt \cdots, \langle\iota; x.c'::K; c\rangle\cdots$
\\\> between $u_1$ and $u_2$ (inclusive), the stack of 
%\\\> 
thread $i$ always 
\\\>contains prefix $x.c'{::}K$
\\\> $\Longrightarrow$
 $\trace\vDash \varphi\}$ 
 \end{tabbing}

\begin{tabbing}
  $\rfi{ \Pi x{:}\tau.u_1.u_2.i.(y{:}\tau'.\varphi, \varphi')}{\theta;\trace;u}{u_b.u_e.i.\varphi_1} =$ \\
  $\{ (k, c)~|~ $\= 
  $ \forall e,$ \=$\forall u', u_B, u_E, \iota, u\leq u'\leq
  u_B\leq u_E$, 
\\\>let $\gamma = [u_B, u_E, \iota/u_1, u_2,i]$
\\\>$(k, e)\in \rei{\tau\gamma}{\theta;\trace;u'}{u_b.u_e.i.\varphi_1}
 \Longrightarrow $ 
\\\> $(k, c\;
  e)\in$ \=$\rci{(y{:}\tau'\gamma.\varphi\gamma)[e/x]}{\theta;\trace;u_B, u_E, \iota}{u_b.u_e.i.\varphi_1}$ 
\\\>\> $\cap \rci{\varphi'\gamma[e/x]}{\theta;\trace;u_B, u_E, \iota}{}\}$
 \end{tabbing}

%\noindent\framebox{$\rai{\alpha}{\theta}$}
\begin{tabbing}
  $\rai{ \kact(u_1.u_2.i.(x{:}\tau.\varphi_1, \varphi_2))}
{\theta;\trace;u}{u_b.u_e.i.\varphi} =$ \\
$\{ (k, a)~|~ $\= $
\forall u_B, u_E, \iota,  u\leq u_B\leq u_E$,
\\\>let $\gamma = [u_B, u_E, \iota/u_1, u_2,i]$
\\\>$(k, \cact(a)) \in 
  ($\=$\rci{x{:}\tau\gamma.\varphi_1\gamma}
{\theta;\trace; u; u_B, u_E,\iota}{u_b.u_e.i.\varphi}$
\\\>$\cap   \rci{\varphi_2\gamma}
 {\theta;\trace; u; u_B, u_E,\iota}{u_b.u_e.i.\varphi} )\}$\\
\\
  $\rai{ \Pi x{:}\tau.\alpha}{\theta;\trace;u}{u_b.u_e.i.\varphi} =$ \\
  $\{ (k, a)~|~ $
 $\forall e,$\=$\forall u', , u'\geq u$, 
$(k,e)\in\rei{\tau}{\theta;\trace; u'}{u_b.u_e.i.\varphi}$
\\\>$ \Longrightarrow   (k, a\; e) \in 
\rai{\alpha[e/x]}{\theta;\trace;u'}{u_b.u_e.i.\varphi}\}$\\
\\
$\rai{ \forall X.\alpha}{\theta;\trace;u}{u_b.u_e.i.\varphi} =$
\\$\{ (k, a)~|~$\=$ 
\forall j\leq k,\forall  \cset\in \ubertype$
\\\>$\Longrightarrow (j, a\; \cdot)\in
\rai{\alpha}{\theta[X\mapsto\cset];\trace;u}{u_b.u_e.i.\varphi}\}$
\end{tabbing}

\paragraph{Formula semantics}

\[
\begin{array}{l@{~}c@{~}l}
 \interp{\texp} &=&  \{e ~|~ e~\mbox{is an expression}\}
\\
\interp{\bt} &=&
\{e~|~ e\rightarrow^* \const\}
 \\
\interp{\Pi x{:}\tau_1. \tau_2} &=&
 \{\lambda x.e ~|~ 
\forall e',e'  \in \interp{\tau_1}
 \Longrightarrow e_1[e'/x]
 \in\interp{\tau_2}\}
\end{array}
\]

\[
\begin{array}{lcl}
\trace \vDash P\;\vec{e} & ~\mbox{iff}~ & P\;\vec{e}  \in\varepsilon(\trace)
\\ 
\trace \vDash \pred{start}(I, c, U)
 & ~\mbox{iff}~ & 
 \mbox{thread $I$ has $c$ as the active} 
\\ & &\mbox{computation with an
   empty stack} 
\\ & & \mbox{at time $U$ on $\trace$}
\\
\trace \vDash \forall x{:}\tau.\varphi & ~\mbox{iff}~ &
\forall e, e \in\interp{\tau}~\mbox{implies}~
 \trace\vDash \varphi[e/x]
\\ 
 \trace \vDash \exists x{:}\tau.\varphi & ~\mbox{iff}~ &
\exists e, e\in\interp{\tau}~\mbox{and}~
  \trace\vDash \varphi[e/x]
\end{array}
\]
\section{Term Language and Operational Semantics}

\paragraph{Syntax}
%% \begin{figure}
%% \centering
\[%\small
\begin{array}{lclll}
 \textit{Base values} & \const & ::= & \btrue ~|~ \bfalse ~|~ \iota ~|~
 \ell ~|~ n 
% \\ 
%  \textit{Simple terms} & t & ::= &\const ~|~ \lambda x.t ~|~ t_1\;t_2
\\ 
% %
\textit{Expressions} & e & ::=  & x ~|~ \const~|~
\lambda x.e  ~|~ \Lambda X.e 
\\ & & ~|~ & e_1\;e_2  ~|~e\; \cdot 
~|~ \ecomp(c) \\
\textit{Actions} & a & ::= & A ~|~ a\; e ~|~ a\; \cdot \\
\textit{Computations} & c & ::= &  
 \cact(a) ~|~  \cret(e) ~|~ \cfix\,f(x).c ~|~
c\; e 
\\ & & ~|~&  \cletc(c_1, x.c_2)  ~|~ \clete(e_1, x. c_2)
\\ & & ~|~&  c_1;c_2  ~|~ e_1;c_2
\\& &~|~ & \cif\ e\ \cthen\ c_1\ \celse\ c_2
\end{array}
\]

\[
\begin{array}{lcll}
\textit{\small Expr types} & \tau & ::= &  X ~|~
\bt ~|~\Pi x{:}
\tau_1. \tau_2 ~|~ \forall X.\tau 
~|~ \tcomp(\eta_c) ~|~ \texp
\\
\textit{\small Comp types} & \eta & ::= & x{:}\tau.\varphi
~|~  \varphi ~|~  (x{:}\tau.\varphi, \varphi')\\
\textit{\small Closed c types} & \eta_c & ::= &
u_1.u_2.i.(x{:}\tau.\varphi_1, \varphi_2)
\\ & &~|~ & \Pi x{:}\tau. u_1.u_2.i.(y{:}\tau.\varphi_1, \varphi_2)
\\
\textit{\small Assertions} & \varphi & ::= & P ~|~ e_1 \equiv e_2 ~|~
\varphi \; e ~|~ \top ~|~ \bot ~|~ \neg\varphi 
\\ & & ~|~ & \varphi_1 \conj \varphi_2 
~|~\varphi_1 \disj \varphi_2 
%\\ & & ~|~ &
~|~\forall x{:} \tau. \varphi ~|~ \exists
x{:}\tau. \varphi
\vspace{5pt}
\\
  \textit{\small Action Kinds} & \alpha & ::= & \kact(\eta_c)
  ~|~ \Pi x{:}\tau. \alpha ~|~ \forall X.\alpha
\\
\textit{\small Type var ctx} & \Theta & ::= & \cdot ~|~ \Theta, X\\
\textit{\small Signatures} & \Sigma & ::= & \cdot ~|~ \Sigma, A:: \alpha\\
\textit{\small Logic var ctx} & \Gamma^L & ::= & \cdot ~|~ \Gamma^L, x:\bt ~|~ \Gamma^L, x: \texp \\ 
\textit{\small Typing ctx} & \Gamma & ::= & \cdot ~|~ \Gamma, x: \tau \\ 
\textit{\small Formula ctx} & \Delta & ::= & \cdot ~|~\Delta, \varphi\\
\textit{\small Exec ctx} & \Xi & ::= & u_b:\bt, u_e:\bt, i:\bt
\end{array}
\]

\paragraph{Beta reductions}
We define the $\beta$-reduction rules below. % We reduce
% expressions inside the monad as well, but a monadic redex is not
% reduced.

%
\noindent\framebox{$e\rightarrow_\beta e'$}
\begin{mathpar}
\inferrule{ }{(\lambda x.e)\, e_2\rightarrow_\beta e[e_2/x] }\and
\inferrule{ }{\Lambda X.e\, \cdot\rightarrow_\beta e }\and
\inferrule{e_1\rightarrow_\beta e'_1}{e_1\, e_2\rightarrow_\beta e'_1\,
e_2}\and
\inferrule{e_1\rightarrow_\beta e'_1}{e_1\, \cdot\rightarrow_\beta e'_1\,
\cdot}\and
\end{mathpar}

\noindent\framebox{$\sigma \rt T \stepsone \sigma' \rt T'$}
\begin{center}
\(%
%\vspace{3pt}
 \inferrule*[right=R-ActS]{
\\\next(\sigma, a) = (\sigma', e) 
\\ e\neq\stuck
 }{ \sigma \rt \langle \iota;
   x.c ::K; \cact(a) \rangle \stepsone \sigma' \rt \langle \iota;
   K;  c[e/x]  \rangle}
\) 

\vspace{10pt} 

\(
 \inferrule*[right=R-ActF]{
\\\next(\sigma, a) = (\sigma', \stuck) 
 }{ \sigma \rt \langle \iota;
   x.c ::K; \cact(a) \rangle \stepsone \sigma' \rt \langle \iota;\stuck \rangle}
\) 

\vspace{10pt} 

\(
 \inferrule*[right=R-Stuck]{  }{ \sigma \rt \langle \iota;\stuck \rangle 
\stepsone \sigma \rt \langle \iota;\stuck \rangle}
\) 

\vspace{10pt} 

\(
\inferrule*[right=R-Ret]{
 }{ \sigma \rt \langle \iota;  x.c :: K; \cret(e) \rangle
  \stepsone  \sigma \rt \langle \iota; K; c[e/x]
  \rangle}
\) 

\vspace{10pt} 

\(
\inferrule*[right=R-SeqE1]{
}{\sigma \rt \langle \iota;
  K; \clete(e_1, x.c_2) \rangle \stepsone  \sigma \rt \langle
  \iota; x.c_2\rhd K; e_1 \rangle}
\)

\vspace{10pt} 

\(
\inferrule*[right=R-SeqE2]{e \rightarrow_\beta e'
}{\sigma \rt \langle \iota;
  K; e \rangle \betaone  \sigma \rt \langle
  \iota; K; e' \rangle}
\)

\vspace{10pt} 

\(
\inferrule*[right=R-SeqE3]{
}{\sigma \rt \langle \iota;
  x.c_2:: K; \ecomp(c_1) \rangle \stepsone  \sigma \rt \langle
  \iota; x.c_2:: K; c_1 \rangle}
\)

\vspace{10pt} 

\(
\inferrule*[right=R-SeqC]{
 }{ \sigma \rt \langle \iota; K; \cletc(c_1, x.c_2)
  \rangle \stepsone  \sigma \rt \langle \iota; x.c_2 :: K; c_1
  \rangle}
\) 

\vspace{10pt} 

\(
 \inferrule*[right=R-Fix]{ 
}{\sigma \rt \langle \iota; K; (\cfix f(x).c)\; e 
   \rangle\\ \stepsone \sigma \rt \langle \iota; K; c[\lambda z.\ecomp(\cfix(f(x).c)\; z)/f][e/x] \rangle}
\)
\end{center}

\framebox{$C \steps{} C'$}
\vspace{-5pt}
\[
\inferrule{\sigma \rt T \stepsone \sigma' \rt
  T'}{ \sigma \rt T, T_1, \ldots, T_n \steps{}  \sigma'
  \rt T', T_1, \ldots, T_n}
\]%

\section{Well-formedness Judgments}

\paragraph{Well-formedness judgments for contexts and
types} 
~\\

\noindent\framebox{$\Theta \vdash \Sigma \jok$}
\begin{mathpar}
\inferrule{ }{\Theta \vdash \cdot \jok} \and 
\inferrule{\Theta \vdash \Sigma \jok \\ \Theta;\Sigma; \cdot \vdash \alpha \jok}{
\Theta \vdash \Sigma, A:: \alpha \jok}\and
%% %
%% \inferrule{\Theta \vdash \Sigma \jok \\ \Theta;\Sigma;\cdot \vdash \rho \jok}{
%% \Theta \vdash \Sigma, P:: \rho \jok}\
\end{mathpar}

\noindent\framebox{$\Theta;\Sigma \vdash \Gamma \jok$}
\begin{mathpar}
\inferrule{ \Theta \vdash\Sigma \jok}{\Theta;\Sigma\vdash \cdot \jok} \and 
\inferrule{\Theta;\Sigma \vdash \Gamma \jok \\ 
  \Theta;\Sigma;\Gamma \vdash \tau \jok}
 {\Theta;\Sigma \vdash\Gamma, x: \tau \jok}
\end{mathpar}

\noindent\framebox{$\Theta;\Sigma; \Gamma\vdash \Delta \jok$}
\begin{mathpar}
\inferrule{\Theta;\Sigma\vdash\Gamma\jok}{\Theta;\Sigma;\Gamma \vdash \cdot \jok} \and 
\inferrule{\Theta;\Sigma;\Gamma \vdash \Delta \jok \\ 
\Gamma \vdash \varphi   \jok
}
 {\Theta;\Sigma,\Gamma\vdash \Delta,\varphi \jok}
\end{mathpar}

\noindent\framebox{$ \Gamma \vdash \varphi\jok$}
\begin{mathpar}
\inferrule{ }{
   \Gamma \vdash P \jok}\and

\inferrule{ \Gamma \vdash \varphi \jok \\ 
\fv(e)\in\dom(\Gamma)
}{
 \Gamma \vdash \varphi \; e \jok}\and
\inferrule{ }{ \Gamma \vdash \top \jok}\and
\inferrule{ }{ \Gamma \vdash \bot \jok}\and
\inferrule{ \Gamma \vdash \varphi_1 \jok \\  \Gamma \vdash
  \varphi_2 \jok}{ \Gamma \vdash \varphi_1 \conj \varphi_2 
  \jok}\and
\inferrule{ \Gamma \vdash \varphi_1 \jok \\  \Gamma \vdash
  \varphi_2 \jok}{ \Gamma \vdash \varphi_1 \disj \varphi_2 
  \jok}\and
\end{mathpar}
\begin{mathpar}
\inferrule{ \Gamma \vdash \varphi \jok}{ \Gamma \vdash \neg \varphi
  \jok}\and
\inferrule{
\tau=b~\mbox{or}~\texp\\
 \Gamma, x: \tau \vdash
  \varphi \jok}{ \Gamma \vdash \forall x{:}\tau. \varphi 
  \jok}\and
\inferrule{
\tau=b~\mbox{or}~\texp\\
 \Gamma, x: \tau \vdash
  \varphi \jok}{ \Gamma \vdash \exists x{:}\tau. \varphi 
  \jok}\and
\inferrule{
\fv(e_1)\cup\fv(e_2) \subseteq\dom(\Gamma)}{
 \Gamma  \vdash e_1 \equiv e_2 \jok}
\end{mathpar}

\noindent\framebox{$\Theta;\Sigma;\Gamma \vdash \tau \jok$}
\begin{mathpar}
\inferrule{ X  \in \Theta 
\\ \Theta;\Sigma\vdash \Gamma\jok
}{\Theta;\Sigma;\Gamma \vdash  X \jok}\and 
\inferrule{\Theta;\Sigma;\Gamma \vdash \tau_1 \jok \\ \Theta;\Sigma;\Gamma, x:\tau_1  \vdash
  \tau_2 \jok}{\Theta;\Sigma;\Gamma \vdash \Pi x{:}\tau_1. \tau_2  \jok}\and
\inferrule{\Theta;\Sigma;\Gamma \vdash \eta_c  \jok}{\Theta;\Sigma;\Gamma \vdash \tcomp(\eta_c)
  \jok}\and
\inferrule{  \Theta;\Sigma \vdash \Gamma \jok
}{\Theta;\Sigma;\Gamma \vdash b \jok}\and
\inferrule{\Theta, X;\Sigma;\Gamma \vdash
  \tau \jok}{\Theta;\Sigma;\Gamma \vdash \forall X. \tau \jok}\and
\inferrule{ \Theta;\Sigma\vdash \Gamma\jok}
{\Theta;\Sigma;\Gamma \vdash \texp \jok}
\end{mathpar}

\noindent\framebox{$\Theta;\Sigma;\Gamma \vdash \alpha \jok$}
\begin{mathpar}%

\inferrule{\Theta;\Sigma;\Gamma \vdash \eta_c \jok}{\Theta;\Sigma;\Gamma \vdash \kact(\eta_c)  \jok}\and
\inferrule{\Theta;\Sigma;\Gamma \vdash \tau \jok \\ \Theta;\Sigma;\Gamma, x:\tau  \vdash
  \alpha \jok}{\Theta;\Sigma;\Gamma \vdash \Pi x{:}\tau. \alpha  \jok}\and
\inferrule{\Theta, X;\Sigma; \Gamma \vdash
  \alpha \jok}{\Theta;\Sigma;\Gamma \vdash \forall X. \alpha \jok}
\end{mathpar}

%% \noindent\framebox{$\Theta;\Sigma;\Gamma \vdash \eta \jok$}
%% \begin{mathpar}
%% %
%% \inferrule{\Theta;\Sigma;\Gamma \vdash \tau \jok \\ 
%% \Theta;\Sigma;\Gamma,x:\tau \vdash \varphi ::\kprop}{\Theta;\Sigma;\Gamma \vdash x{:}\tau.\varphi
%%   \jok} \and
%% %
%%  \inferrule{\Theta;\Sigma;\Gamma \vdash \varphi ::\kprop }{\Theta;\Sigma;\Gamma \vdash \varphi \jok}
%% \end{mathpar}

\noindent\framebox{$\Theta;\Sigma;\Gamma \vdash \eta_c \jok$}
\begin{mathpar}
\inferrule{
\Theta;\Sigma;\Gamma,u_1{:}\bt, u_2{:}\bt,i{:}\bt \vdash
\tau\jok
\\ \Gamma,u_1{:}\bt, u_2{:}\bt,i{:}\bt,x:\tau \vdash \varphi_1\jok \\
 \Gamma,u_1{:}\bt, u_2{:}\bt,i{:}\bt  \vdash \varphi_2 \jok}{\Theta;\Sigma;\Gamma \vdash
  u_1.u_2.i.(x{:}\tau.\varphi_1, \varphi_2)   \jok}\and
\inferrule{\Theta;\Sigma;\Gamma \vdash \tau \jok \\ 
\Theta;\Sigma;\Gamma, y:\tau
\vdash u_1.u_2.i.(x{:}\tau_1.\varphi_1, \varphi_2) \jok
}{\Theta;\Sigma;\Gamma \vdash
\Pi y{:}\tau.u_1.u_2.i.(x{:}\tau_1.\varphi_1, \varphi_2)   \jok}\and
\end{mathpar}

\section{Typing Rules}
\paragraph{Typing for simple terms}

\noindent \framebox{$ \Gamma\vdash_e e: \tau$}
\begin{mathpar}
\inferrule{
 x : \tau \in \Gamma}{ \Gamma\vdash x:
  \tau}\and
\inferrule{ \Gamma, x: \tau_1 \vdash e
  : \tau_2}{ \Theta;\Upsilon \vdash \lambda x.e : \Pi x{:}\tau_1. \tau_2}\and
\inferrule{ \Gamma \vdash e_1 : \Pi x{:}\tau_1.\tau_2 \\  \Gamma \vdash
  e_2: \tau_1}{ \Gamma \vdash e_1\;e_2 : \tau_2}\and
\inferrule{ }{ \Gamma\vdash e:\texp}\and
\end{mathpar}

\paragraph{Confine relation}
\begin{mathpar}
\inferrule{ }{ 
  \cnfn{b}{u_b.u_e.i.\varphi}}
\and
\inferrule{ \cnfn{\tau_1}{u_b.u_e.i.\varphi}~~~~
\cnfn{\tau_2}{u_b.u_e.i.\varphi} }{ 
  \cnfn{\Pi
    \_{:}\tau_1. \tau_2}{u_b.u_e.i.\varphi}}
\and
\inferrule{\cnfn{\tau}{u_b.u_e.i.\varphi} }{
 \cnfn{\tcomp(\tpair{u_b.u_e.i}{x{:}\tau.\varphi}{\varphi})}{u_b.u_e.i.\varphi}}
\end{mathpar}

\paragraph{Typing rules for expressions}
~\\
%\begin{figure*}[t!]

\noindent\framebox{$u;\Theta; \Sigma;\Gamma^L;\Gamma;\Delta \vdash e: \tau$}
\begin{mathpar}
\inferrule*[right=E-Var]{
\Theta;\Sigma;u,\Gamma^L;\Gamma \vdash\Delta \jok 
    \\ x : \tau \in \Gamma}{u;\Theta; \Sigma;\Gamma^L;\Gamma;\Delta  \vdash x:
  \tau}\and
\inferrule*[right=Un]{
\Theta;\Sigma;\Gamma^L, u,\Gamma\vdash\Delta \jok\\
\fv(e)\subseteq\dom(\Gamma)}
{u;\Theta;\Sigma;\Gamma^L;\Gamma;\Delta\vdash e: \texp}\and
\inferrule*[right=E-BaseVal]{
  \Theta;\Sigma;\Gamma^L,u,\Gamma \vdash\Delta \jok 
   }{u;\Theta; \Sigma;\Gamma^L;\Gamma;\Delta  \vdash \const:
  \bt}\and
\inferrule*[right=E-Fun]{\Theta;\Sigma;\Gamma^L,u,\Gamma\vdash \tau_1 \jok 
\\u;\Theta; \Sigma;\Gamma^L;\Gamma, x: \tau_1 ;\Delta \vdash_Q e
  : \tau_2}{u;\Theta; \Sigma;\Gamma^L;\Gamma;\Delta  \vdash_Q \lambda x.e : \Pi x{:}\tau_1. \tau_2}\and
\inferrule*[right=E-App]{u;\Theta; \Sigma;\Gamma^L;\Gamma;\Delta  \vdash_Q e_1 :
  \Pi x{:}\tau_1.\tau_2 
  \\ u;\Theta; \Sigma;\Gamma^L;\Gamma;\Delta  \vdash_Q e_2: \tau_1}{
  u;\Theta; \Sigma;\Gamma^L;\Gamma;\Delta  \vdash_Q e_1\;e_2 : \tau_2[e_2/x]}\and
\inferrule*[right=E-TFun]{u;\Theta, X ; \Sigma;\Gamma^L;\Gamma;\Delta  \vdash_Q e
  : \tau}{u;\Theta; \Sigma;\Gamma^L;\Gamma;\Delta  \vdash_Q \Lambda X.e : \forall X. \tau}\and
\inferrule*[right=E-TApp]{u;\Theta; \Sigma;\Gamma^L;\Gamma;\Delta 
  \vdash_Q  e: \forall X.\tau_1 
\\ \Theta;\Sigma;\Gamma^L,u,\Gamma \vdash  \tau \jok}{
  u;\Theta; \Sigma;\Gamma^L;\Gamma;\Delta  \vdash_Q e\;\cdot :
  \tau_1[\tau/X]}\and
\end{mathpar}

\begin{mathpar}
\inferrule*[right=Eq]{u;\Theta;\Sigma;\Gamma^L;\Gamma;\Delta \vdash_Q e : \tau\\
\Theta;\Sigma;\Gamma^L,u;\Gamma;\Delta \vdash  e \equiv e' \jtrue 
\\ \fv(e')\subseteq\dom(\Gamma)}{
u;\Theta;\Sigma;\Gamma^L;\Gamma;\Delta \vdash_Q e' : \tau }
\and
\inferrule*[right=Confine]
{  
 \mbox{$\varphi$ is trace composable} 
\\ u_b, u_e, i; \Theta;\Sigma;\Gamma^L,u;\Gamma;\Delta \vdash \varphi\jonempty 
\\ u_b{:}\bt, u_e{:}\bt, i{:}\bt \vdash
\varphi\jok
 \\\fa(e)=\emptyset 
\\\fv(e)\subseteq\Gamma
 \\\cnfn{\tau}{u_b.u_e.i.\varphi}
 \\\cnfn{\Gamma}{u_b.u_e.i.\varphi}
}
{u;\Theta;\Sigma;\Gamma^L;\Gamma;\Delta  \vdash_{u_b.u_e.i.\varphi} e : \tau}
\and
\inferrule*[right=Conf-sub]
{u;\Theta;\Sigma;\Gamma^L;\Gamma;\Delta  \vdash e : \tau
\\ u_b{:}\bt, u_e{:}\bt, i{:}\bt \vdash \varphi\jok
}
{u;\Theta;\Sigma;\Gamma^L;\Gamma;\Delta  \vdash_{u_b.u_e.i.\varphi} e : \tau}
\end{mathpar}

\begin{mathpar}
\mprset{flushleft}
\inferrule*[right=Comp]{
 u_1, u_2, i;  \Theta;\Sigma;\Gamma^L;u_e,\Gamma;\Delta, u_1\geq u_e 
\vdash_Q c : (x{:}\tau.\varphi_1, \varphi_2) 
%\\ \fv(\Delta)\subseteq\dom(\Gamma)
\\
   \Theta;\Sigma;\Gamma^L, u_e{:}\bt, u_1{:}\bt, u_2{:}\bt, i{:}\bt;\Gamma, x:\tau ;\Delta
\vdash \varphi_1 \Rightarrow \varphi'_1  \jtrue
\\\\
   \Theta;\Sigma;\Gamma^L,u_e{:}\bt, u_1{:}\bt, u_2{:}\bt, i{:}\bt ;\Gamma;\Delta
 \vdash \varphi_2 \imp \varphi'_2  \jtrue
 \\ \Theta;\Sigma;\Gamma^L,u_e{:}\bt;\Gamma\vdash
 \tpair{u_1.u_2.i}{x{:}\tau.\varphi'_1}{\varphi'_2}\jok
\\\fv(c)\subseteq\dom(\Gamma)
}{
u_e; \Theta;\Sigma;\Gamma^L;\Gamma;\Delta  \vdash_Q \ecomp(c) :
  \tcomp(\tpair{u_1.u_2.i}{x{:}\tau.\varphi'_1}{\varphi'_2})}\and
\end{mathpar}

\paragraph{Typing rules for silent threads}
~\\

\noindent\framebox{$\Xi;\Theta;\Sigma;\Gamma^L;\Gamma;\Delta \vdash \varphi\jonempty$}
\begin{mathpar}
\inferrule*[right=Silent]{
  \Theta;\Sigma;\Gamma^L; \Xi, \Gamma;\Delta \vdash \varphi\jtrue
\\ \Gamma^L, \Xi, \Gamma \vdash \varphi \jok
}{
\Xi;\Theta;\Sigma;\Gamma^L;\Gamma;\Delta  \vdash \varphi\jonempty
}
\end{mathpar}

\paragraph{Typing rules for actions}
~\\

\noindent\framebox{$u:\bt;\Theta;\Sigma;\Gamma^L;\Gamma;\Delta \vdash_Q a:: \alpha$}
\begin{mathpar}
\inferrule{\Theta ; \Sigma;\Gamma^L, u:\bt; \Gamma\vdash \Delta \jok 
\\ A :: \alpha \in \Sigma}{u:\bt; \Theta;\Sigma;\Gamma^L; \Gamma;\Delta \vdash
  A :: \alpha}\and
\inferrule{u:\bt; \Theta;\Sigma;\Gamma^L; \Gamma;\Delta \vdash_Q a :: \Pi
  x{:}\tau.\alpha 
\\ u:\bt; \Theta;\Sigma;\Gamma^L; \Gamma;\Delta\vdash_Q  e:
  \tau
}{u:\bt; \Theta;\Sigma;\Gamma^L; \Gamma;\Delta \vdash_Q a \; e :: \alpha[e/x]}\and
\inferrule{u:\bt; \Theta;\Sigma;\Gamma^L; \Gamma;\Delta \vdash_Q a :: \forall X.\alpha
  \\ 
\Theta;\Sigma;\Gamma^L,u:\bt,\Gamma\vdash \tau \jok}{u:\bt; \Theta;\Sigma;\Gamma^L; \Gamma;\Delta \vdash_Q a \; \cdot :: \alpha[\tau/X]}
\end{mathpar}

\paragraph{Logical reasoning rules}
~\\

\noindent\framebox{$\Theta; \Sigma;\Gamma^L; \Gamma ;\Delta
  \vdash \varphi  \jtrue$} 
\begin{mathpar}
\inferrule*[right=Honest]{
   u_1, u_2, i; \Theta;\Sigma;\Gamma^L;\cdot; \Delta \vdash c :  \varphi
\\  \Theta;\Sigma;\Gamma^L; \cdot; \Delta \vdash \pred{start}(I, c, u)
\jtrue 
\\ \Theta;\Sigma\vdash\Gamma^L, \Gamma \jok
 }{
 \Theta;\Sigma;\Gamma^L; \Gamma; \Delta  \vdash \forall u'{:}\bt.(u'{>}u)  \imp
 \varphi[u,u',I/u_1, u_2,i] \jtrue}
\and
\inferrule*[right=Cut]{
\Theta;\Sigma;\Gamma^L;\Gamma;\Delta_1  \vdash \varphi \jtrue\\
\Theta;\Sigma;\Gamma^L;\Gamma;\Delta_1 , \varphi, \Delta_2 \vdash \varphi' \jtrue\\
}{
\Theta;\Sigma;\Gamma^L;\Gamma;\Delta_1, \Delta_2  \vdash \varphi' \jtrue\\
}\and
\inferrule*[right=Init]{
\Theta;\Sigma ; \Gamma^L,\Gamma \vdash\Delta\jok \\
\varphi \in \Delta
}{
\Theta; \Sigma;\Gamma^L;\Gamma;\Delta  \vdash \varphi \jtrue
}\and
\inferrule*[right=$\neg$I]{ 
 \Theta; \Sigma;\Gamma^L;\Gamma;\Delta_1 ,\varphi,\Delta_2 \vdash \cdot
}{ \Theta; \Sigma;\Gamma^L;\Gamma;\Delta_1, \Delta_2  \vdash
  \neg\varphi\jtrue}\and
\inferrule*[right=$\neg$E]{ 
 \Theta; \Sigma;\Gamma^L;\Gamma;\Delta  \vdash \neg\varphi\jtrue
}{ \Theta; \Sigma;\Gamma^L;\Gamma;\Delta , \varphi \vdash \cdot
  }\and

 \inferrule*[right=$\conj$I]{
 \Theta; \Sigma;\Gamma^L;\Gamma;\Delta  \vdash \varphi_1 \jtrue
\\  \Theta; \Sigma;\Gamma^L;\Gamma;\Delta  \vdash \varphi_2 \jtrue
 }{\Theta; \Sigma;\Gamma^L;\Gamma;\Delta  \vdash \varphi_1\conj\varphi_2 \jtrue}\and
\end{mathpar}

\begin{mathpar}
 \inferrule*[right=$\conj$E]{
 i\in[1,2], 
 \Theta; \Sigma;\Gamma^L;\Gamma;\Delta  \vdash \varphi_1\conj\varphi_2  \jtrue
 }{\Theta; \Sigma;\Gamma^L;\Gamma;\Delta  \vdash \varphi_i \jtrue}\and
 \inferrule*[right=$\disj$I]{
 i\in[1,2], 
 \Theta; \Sigma;\Gamma^L;\Gamma;\Delta  \vdash \varphi_i \jtrue
 }{\Theta; \Sigma;\Gamma^L;\Gamma;\Delta  \vdash \varphi_1\disj\varphi_2 \jtrue}\and
 \inferrule*[right=$\disj$E]{
 \Theta; \Sigma;\Gamma^L;\Gamma;\Delta  \vdash \varphi_1\disj\varphi_2 \jtrue
 \\
 \Theta; \Sigma;\Gamma^L;\Gamma;\Delta , \varphi_1, \Gamma' \vdash \varphi\jtrue
 \\ \Theta; \Sigma;\Gamma^L;\Gamma;\Delta , \varphi_2, \Gamma'  \vdash \varphi\jtrue
 }{\Theta; \Sigma;\Gamma^L;\Gamma;\Delta, \Gamma'  \vdash \varphi\jtrue}\and
\inferrule*[right=$\forall$I]{
\Theta; \Sigma;\Gamma^L, x:\tau;\Gamma;\Delta  \vdash \varphi \jtrue
}{\Theta; \Sigma;\Gamma^L;\Gamma;\Delta  \vdash \forall
  x{:}\tau. \varphi \jtrue}\and
\inferrule*[right=$\forall$E]{
\Theta; \Sigma;\Gamma^L;\Gamma;\Delta \vdash \forall
  x{:}\tau. \varphi \jtrue \\
 \Gamma^L \vdash t: \tau
}{\Theta; \Sigma;\Gamma^L;\Gamma;\Delta  \vdash \varphi[t/x]\jtrue}\and
\inferrule*[right=$\exists$I]{
\Theta; \Sigma;\Gamma^L;\Gamma;\Delta \vdash \varphi[t/x]\jtrue
\\  \Gamma^L \vdash t: \tau
}{\Theta; \Sigma;\Gamma^L;\Gamma;\Delta  \vdash \exists
  x{:}\tau. \varphi \jtrue }\and
\inferrule*[right=$\exists$E]{
\Theta; \Sigma;\Gamma^L;\Gamma;\Delta  \vdash \exists x{:}\tau.\varphi \jtrue
\\
\Theta; \Sigma;\Gamma^L, a{:}\tau;\Gamma;\Delta, \varphi[a/x] 
 \vdash \varphi'\jtrue
\\ a\notin\fv(\varphi')
}{\Theta; \Sigma;\Gamma^L;\Gamma;\Delta  \vdash \varphi'\jtrue}\and
\end{mathpar}

\paragraph{Typing rules for computations} 
We summarize the typing rules for computations in
Figures~\ref{fig:typing-comp-1} and ~\ref{fig:typing-comp-2}.

\begin{figure*}
\paragraph{Fixpoint}\framebox{$u:\bt;\Theta; \Sigma;\Gamma^L; \Gamma;\Delta \vdash_Q c : \eta$}
\begin{mathpar}
\mprset{flushleft}
 \inferrule*[right=Fix]{
\Gamma_1 =  y:\tau, f: \Pi
 y{:}\tau. \tcomp(u_1.u_3.i.(x{:}\tau_1.\varphi, \varphi'))
\\\\
 u_1: \bt, u_2:\bt, i: \bt; \Theta;\Sigma;\Gamma^L; \Gamma;\Delta,
  u\leq  u_1\leq u_2 \vdash\varphi_0 \jonempty
\\\\ u_2, u_3, i; \Theta;
   \Sigma;\Gamma^L,u_1: \bt, u:\bt;
   \Gamma, \Gamma_1;\Delta,  u_2 <  u_3, \varphi_0  \vdash_Q c : x{:}\tau_1.\varphi_1
\\\\ u_2, u_3, i; \Theta;
   \Sigma;\Gamma^L; u_1: \bt, u:\bt;
   \Gamma, \Gamma_1;\Delta, u_2 \leq   u_3, \varphi_0  \vdash_Q c :\varphi_2
   \\\\ \Theta;\Sigma;\Gamma^L, u_1: \bt, u:\bt,  u_2: \bt, u_3:\bt,  i:\bt; \Gamma, \Gamma_1,
  x:\tau_1;\Delta   \vdash 
  (\varphi_0 \conj \varphi_1) \Rightarrow \varphi \jtrue
   \\\\ \Theta;\Sigma;\Gamma^L , u_1: \bt , u_2: \bt, u_3:\bt,
   i:\bt, u:\bt; \Gamma,\Gamma_1;\Delta
             \vdash (\varphi_0 \conj \varphi_2 \imp \varphi') \jtrue
   \\\\ \Theta;\Sigma;\Gamma^L, u_1:\bt,u_3:\bt ,
   i:\bt , u:\bt; \Gamma, y:\tau;\Delta\vdash 
   \varphi_0[u_3/u_2] \imp \varphi' \jtrue
\\\\ \Theta;\Sigma;\Gamma^L,u:\bt;\Gamma \vdash \Pi   y{:}\tau.u_1.u_3.i.(x{:}\tau_1.\varphi,
   \varphi') \jok
\\ \fv(\cfix(f(y).c))\in\dom(\Gamma)
 }{ u; \Theta;\Sigma;\Gamma^L; \Gamma;\Delta  \vdash_Q \cfix(f(y).c) : \Pi
   y{:}\tau.u_1.u_3.i.(x{:}\tau_1.\varphi, \varphi') }
\end{mathpar}

\paragraph{Partial correctness typing}\framebox{$\Xi; \Theta; \Sigma;\Gamma^L; \Gamma;\Delta  \vdash_Q c : \eta$}
 \begin{mathpar}
\mprset{flushleft}
\inferrule*[right=App]{
 u_1: \bt; \Theta; \Sigma;\Gamma^L,  u_2:\bt, i: \bt; \Gamma;\Delta  \vdash_Q c: \Pi
  y{:}\tau.u_b.u_e.j.(x{:}\tau'.\varphi, \varphi')\\
u_1: \bt;\Theta; \Sigma;\Gamma^L, u_2:\bt,
  i: \bt; \Gamma;\Delta  \vdash_Q e: \tau \\
\fv(c\;e)\subseteq\dom(\Gamma)\\
\\\mbox{let}~\gamma=[u_1,u_2,i/u_b,u_e,j]
\\ \Theta;\Sigma;\Gamma^L,\Gamma \vdash u_1.u_2.i.((x{:}\tau'.\varphi)\gamma[e/y], \varphi'\gamma[e/y])\jok
}{
u_1: \bt, u_2:\bt, i: \bt; \Theta; \Sigma;\Gamma^L; \Gamma ;\Delta
\vdash_Q c\; e: ((x{:}\tau'\gamma.\varphi\gamma)[e/y], \varphi'\gamma[e/y])
}\and
 \inferrule*[right=Act]{u_1:\bt;\Theta; \Sigma;\Gamma^L,u_2:\bt,i:\bt; \Gamma;\Delta \vdash_Q a ::
   \kact(\tpair{u_b.u_e.j}{x{:}\tau.\varphi_1}{\varphi_2} )
\\
   u_1: \bt, u_2:\bt, i: \bt; \Theta; \Sigma;\Gamma^L; \Gamma;\Delta  \vdash
   \varphi \jonempty 
\\ \fv(a)\in\dom(\Gamma)
\\ \Theta; \Sigma;\Gamma^L; \Gamma\vdash
    u_1.u_2.i.(x{:}\tau.\varphi_1[u_1,u_2,i/u_b,u_e,j],
 \varphi_2[u_1,u_2,i/u_b,u_e,j]\conj \varphi)\jok
}{
 u_1: \bt, u_2:\bt, i: \bt; \Theta; \Sigma;\Gamma^L; \Gamma;\Delta
    \vdash_Q \cact(a) : (x{:}\tau.\varphi_1[u_1,u_2,i/u_b,u_e,j],
 \varphi_2[u_1,u_2,i/u_b,u_e,j]\conj \varphi)}\and
 \inferrule*[right=SeqE]{
 u_0: \bt, u_1:\bt, i: \bt; 
 \Theta; \Sigma;\Gamma^L,  u_3:\bt; \Gamma;\Delta , u_0\leq u_1 \vdash \varphi_0 \jonempty
 \\\\ u_1:\bt;\Theta; \Sigma;\Gamma^L, u_0: \bt, u_3:\bt, i: \bt;\Gamma;\Delta, \varphi_0 
 \vdash_Q e_1 : \tcomp(\tpair{u_b, u_e, j}{x{:}\tau.\varphi_1}{\varphi'_1})
\\\\\mbox{let}~\gamma=[u_1, u_2, i/u_b, u_e, j]
   \\\\ u_2:\bt, u_3:\bt, i:\bt; \Theta; \Sigma;\Gamma^L, u_0: \bt, u_1:\bt;\Gamma,
   x: \tau\gamma;\Delta,   u_2 < u_3,\varphi_0 , \varphi_1\gamma \vdash_Q 
   c_2: y{:}\tau'.\varphi_2
\\\\
  \Theta; \Sigma;\Gamma^L, u_0: \bt, u_3:\bt, i: \bt; \Gamma,
  u_1{:}\bt,u_2{:}\bt, x{:}\tau\gamma, y:\tau' ;\Delta  \vdash 
 (\varphi_0 \conj \varphi_1\gamma \conj
  \varphi_2) \Rightarrow \varphi \jtrue 
\\\\\fv(\clete(e_1, x.c_2))\subseteq\dom(\Gamma)
%\\ \Theta;\Sigma;\Gamma^L;u_3:\bt, i: \bt, \Gamma\vdash \tau'\jok 
\\  \Theta;\Sigma;\Gamma^L, u_0: \bt, u_3:\bt, i: \bt;\Gamma, y:\tau'
 \vdash  \varphi\jok
}{
  u_0: \bt, u_3:\bt, i: \bt; \Theta; \Sigma;\Gamma^L; \Gamma;\Delta  \vdash_Q
   \clete(e_1, x.c_2): y{:}\tau'.\varphi }\and
 \inferrule*[right=SeqC]{
    u_0: \bt, u_1:\bt, i: \bt; 
 \Theta; \Sigma;\Gamma^L;  u_3:\bt,\Gamma;\Delta , u_0\leq u_1 \vdash \varphi_0 \jonempty
   \\ u_1: \bt, u_2:\bt, i: \bt; \Theta; \Sigma;\Gamma^L, u_0:
   \bt,  u_3:\bt;\Gamma;\Delta,
    u_1 < u_2, \varphi_0 \vdash_Q c_1 : x{:}\tau.\varphi_1
  \\\\ u_2: \bt, u_3:\bt, i: \bt; \Theta; \Sigma;\Gamma^L,
   u_0: \bt, u_1:\bt;\Gamma,x : \tau;\Delta,u_2 < u_3,\varphi_0, \varphi_1 
    \vdash_Q 
    c_2: y{:}\tau'.\varphi_2
\\\\
  \Theta; \Sigma;\Gamma^L,u_1{:}\bt, u_2{:}\bt,  u_0: \bt, u_3:\bt, i: \bt;\Gamma,
   x{:}\tau,y:\tau';\Delta  \vdash 
  (\varphi_0 \conj \varphi_1 \conj
  \varphi_2) \Rightarrow \varphi \jtrue 
%\\\\ \Theta;\Sigma;\Gamma^L;u_0: \bt, u_3:\bt, i: \bt,\Gamma\vdash  \tau'\jok 
\\  \Theta;\Sigma;\Gamma^L, u_0: \bt, u_3:\bt, i: \bt;\Gamma,
y:\tau' \vdash  \varphi \jok
\\\fv(\cletc(c_1, x.c_2))\subseteq\dom(\Gamma)
}{
  u_0: \bt, u_3:\bt, i: \bt; \Theta; \Sigma;\Gamma^L; \Gamma;\Delta  \vdash_Q
   \cletc(c_1, x.c_2): y{:}\tau'.\varphi
}\and
 \inferrule*[right=Ret]{ 
u_2:\bt; \Theta; \Sigma;\Gamma^L,  u_1: \bt, i: \bt;
\Gamma;\Delta %,u_0\geq u_2 
\vdash_Q e: \tau \\
    u_1: \bt, u_2:\bt, i: \bt; \Theta; \Sigma;\Gamma^L; \Gamma;\Delta  \vdash
   \varphi \jonempty
\\ \fv(e)\subseteq\dom(\Gamma) }{ 
   u_1: \bt, u_2:\bt, i: \bt;  \Theta; \Sigma;\Gamma^L; \Gamma;\Delta 
     \vdash_Q \cret(e): x{:}\tau.((x \equiv e) \conj
   \varphi)}\and
\inferrule*[right=If]{
    u_0: \bt, u_1:\bt, i: \bt; 
 \Theta; \Sigma;\Gamma^L, u_2: \bt; \Gamma;\Delta , u_0\leq u_1 \vdash \varphi_0 \jonempty
\\\\ u_0: \bt;\Theta; \Sigma;\Gamma^L, u_2:\bt, i: \bt; \Gamma;\Delta  \vdash_Q e : \bt 
\\ u_1: \bt, u_2:\bt, i: \bt;  \Theta; \Sigma;\Gamma^L, u_0:\bt;
\Gamma;\Delta, u_1 < u_2, \varphi_0, 
 (\pred{eval}\, e\, \btrue) \vdash_Q  c_1:
x{:}\tau.\varphi_1
\\ u_1: \bt, u_2:\bt, i: \bt;  \Theta; \Sigma;\Gamma^L, u_0:\bt
; \Gamma;\Delta 
    , u_1 < u_2, \varphi_0, (\pred{eval}\, e\, \bfalse) \vdash_Q  c_2: x{:}\tau.\varphi_2
\\  \Theta; \Sigma;\Gamma^L,u_0: \bt, u_2:\bt, i: \bt;\Gamma, u_1:\bt, x:\tau ;\Delta
  \vdash (\varphi_0\conj \varphi_i)\Rightarrow\varphi
  ~\mbox{where}~i\in[1,2]
\\ \Theta;\Sigma;\Gamma^L, u_0: \bt, u_2:\bt, i: \bt;\Gamma,x:\tau
\vdash  \varphi\jok
%\\ \Theta;\Sigma;\Gamma^L;u_2:\bt, i: \bt, \Gamma\vdash \tau \jok
\\\fv(e)\cup \fv(c_1)\cup\fv(c_2)\subseteq\dom(\Gamma)
}{ u_0: \bt, u_2:\bt, i: \bt;  \Theta; \Sigma;\Gamma^L; \Gamma;\Delta 
     \vdash_Q \cif\ e\ \cthen\ c_1\ \celse\ c_2:
  x{:}\tau.\varphi}\and
%% %
%% \inferrule*[right=Imp]{ u_1: \bt, u_2:\bt, i: \bt;  \Theta; \Sigma;\Gamma^L; \Gamma;\Delta 
%%   , \varphi_1 \vdash_Q c: x{:}\tau.\varphi_2}
%% { u_1: \bt, u_2:\bt, i: \bt;  \Theta; \Sigma;\Gamma^L; \Gamma;\Delta 
%%      \vdash_Q c: x{:}\tau. \varphi_1\imp\varphi_2}
\end{mathpar} 
\caption{Computation typing rules (1)}
\label{fig:typing-comp-1}
\end{figure*}
\begin{figure*}
\paragraph{Invariant typing}\framebox{$\Xi; \Theta; \Sigma;\Gamma^L; \Gamma;\Delta  \vdash_Q c : \eta$}
\begin{mathpar}
\mprset{flushleft}
 \inferrule*[right=SeqEI]{
 \Theta ;\Sigma;\Gamma^L,  u_0: \bt, u_3:\bt, i: \bt;\Gamma;\Delta\vdash  \varphi\jok
\\  u_0: \bt, u_1:\bt, i: \bt; 
 \Theta; \Sigma;\Gamma^L, u_3:\bt;\Gamma;\Delta , u_0\leq u_1 \vdash \varphi_0 \jonempty
\\  u_0: \bt, u_3:\bt, i: \bt; 
 \Theta; \Sigma;\Gamma^L; \Gamma;\Delta , u_0\leq u_3 \vdash \varphi'_0 \jonempty
\\\\  u_1: \bt; \Theta; \Sigma;\Gamma^L, u_0: \bt, u_3:\bt,
i: \bt; \Gamma;\Delta,\varphi_0
 \vdash_Q e_1 : \tcomp(\tpair{u_b, u_e, j}{x{:}\tau.\varphi_1}{\varphi'_1})
  \\ u_2: \bt, u_3:\bt, i: \bt; \Theta; \Sigma;\Gamma^L,  u_0: \bt;\Gamma;\Delta,
   u_1:\bt, x: \tau;
  u_1 < u_2\leq u_3,\varphi_0, \varphi_1[u_1, u_2, i/u_b, u_e, j] \vdash_Q 
   c_2: \varphi_2
\\\\
 \Theta; \Sigma;\Gamma^L,  u_0: \bt, u_3:\bt, i: \bt;\Gamma;\Delta  \vdash 
   \varphi'_0 \imp \varphi \jtrue 
\\\\
 \Theta; \Sigma;\Gamma^L,  u_0: \bt, u_3:\bt, i: \bt;\Gamma, u_1{:}\bt;\Delta  \vdash 
\varphi_0 \conj \varphi'_1[u_1, u_3, i/u_b, u_e, j] \imp \varphi \jtrue 
\\\\
  \Theta; \Sigma;\Gamma^L,  u_0: \bt, u_3:\bt, i: \bt;\Gamma, 
 u_1{:}\bt, u_2{:}\bt, x{:}\tau;\Delta  \vdash 
  (\varphi_0
  \conj \varphi_1[u_1, u_2, i/u_b, u_e, j] \conj
  \varphi_2) \imp \varphi \jtrue 
\\\\\fv(\clete(e_1, x.c_2))\subseteq\dom(\Gamma)
}{
  u_0: \bt, u_3:\bt, i: \bt; \Theta; \Sigma;\Gamma^L; \Gamma;\Delta  \vdash_Q
   \clete(e_1, x.c_2): \varphi }\and
 \inferrule*[right=SeqCI]{
 \Theta ;\Sigma;\Gamma^L,  u_0: \bt, u_3:\bt, i: \bt;\Gamma;\Delta\vdash  \varphi \jok
\\  u_0: \bt, u_1:\bt, i: \bt; 
 \Theta; \Sigma;\Gamma^L, u_3:\bt;\Gamma;\Delta , u_0\leq u_1 \vdash \varphi_0 \jonempty
\\  u_0: \bt, u_3:\bt, i: \bt; 
 \Theta; \Sigma;\Gamma^L; \Gamma;\Delta , u_0\leq u_3 \vdash \varphi'_0 \jonempty
 \\u_1: \bt, u_2:\bt, i: \bt;  \Theta; \Sigma;\Gamma^L,u_0:
 \bt, u_3:\bt;\Gamma;\Delta,   u_1 < u_2, \varphi_0 \vdash_Q c_1 : x{:}\tau.\varphi_1
\\u_1: \bt, u_3:\bt, i: \bt;  \Theta; \Sigma;\Gamma^L;\Gamma;\Delta, u_0:
 \bt, u_1 \leq u_3, \varphi_0 \vdash_Q c_1 : \varphi'_1
  \\ u_2: \bt, u_3:\bt, i: \bt; \Theta; \Sigma;\Gamma^L; \Gamma;\Delta,
   u_0: \bt, u_1:\bt, x: \tau, u_2\leq u_3,\varphi_0, \varphi_1 \vdash_Q 
   c_2: \varphi_2
\\\\
 \Theta; \Sigma;\Gamma^L,  u_0: \bt, u_3:\bt, i: \bt;\Gamma;\Delta  \vdash 
   \varphi'_0 \imp \varphi \jtrue 
\\\\
 \Theta; \Sigma;\Gamma^L,  u_0: \bt, u_3:\bt, i: \bt;\Gamma,u_1{:}\bt;\Delta  \vdash 
  (\varphi_0 \conj \varphi'_1) \imp \varphi \jtrue 
\\\\
  \Theta; \Sigma;\Gamma^L,  u_0: \bt, u_3:\bt, i: \bt;\Gamma, u_1{:}\bt, u_2{:}\bt,
  x{:}\tau;\Delta  \vdash (\varphi_0\conj \varphi_1 \conj
  \varphi_2) \imp \varphi \jtrue 
\\\fv(\cletc(c_1, x.c_2))\subseteq\dom(\Gamma)
}{
  u_0: \bt, u_3:\bt, i: \bt; \Theta; \Sigma;\Gamma^L; \Gamma;\Delta  \vdash_Q
   \cletc(c_1, x.c_2): \varphi }\and
 \inferrule*[right=RetI]{ 
% \Theta ;\Sigma;\Gamma^L; \Gamma;\Delta  \vdash  \varphi \jok \\
\fv(e)\subseteq\dom(\Gamma)
\\
    u_1: \bt, u_2:\bt, i: \bt; \Theta; \Sigma;\Gamma^L; \Gamma;\Delta  \vdash
   \varphi \jonempty }{ 
   u_1: \bt, u_2:\bt, i: \bt;  \Theta; \Sigma;\Gamma^L; \Gamma;\Delta 
     \vdash_Q \cret(e):  \varphi}  \and
\inferrule*[right=IfI]{
 \Theta; \Sigma;\Gamma^L, u_0: \bt, u_2:\bt, i: \bt;\Gamma;\Delta  \vdash_Q e : \bt 
\\    u_0: \bt, u_2:\bt, i: \bt; \Theta; \Sigma;\Gamma^L;
\Gamma;\Delta, u_1\leq u_2 \vdash
   \varphi_0 \jonempty
\\    u_0: \bt, u_1:\bt, i: \bt; \Theta; \Sigma;\Gamma^L, u_2:\bt;\Gamma;\Delta,
u_0 \leq u_1\vdash
   \varphi'_0 \jonempty
\\ u_1: \bt, u_2:\bt, i: \bt;  \Theta; \Sigma;\Gamma^L,  u_0: \bt;\Gamma;\Delta, 
    u_1 \leq u_2,\varphi'_0, (\pred{eval}\, e\, \btrue) \vdash_Q  c_1:
\varphi_1
\\ u_1: \bt, u_2:\bt, i: \bt;  \Theta; \Sigma;\Gamma^L,  u_0: \bt;\Gamma;\Delta,
 u_1 \leq u_2,\varphi'_0,
(\pred{eval}\, e\, \bfalse)\vdash_Q  c_2: \varphi_2
\\ \Theta; \Sigma;\Gamma^L,  u_0: \bt, u_2:\bt, i: \bt;\Gamma;\Delta
\vdash \varphi_0 \imp \varphi
 \\ \Theta; \Sigma;\Gamma^L,  u_0: \bt, u_2:\bt, i: \bt;\Gamma, u_1:\bt;\Delta 
 \vdash (\varphi'_0 \conj \varphi_1) \imp \varphi
\\ \Theta; \Sigma;\Gamma^L,  u_0: \bt, u_2:\bt, i: \bt;\Gamma, u_1:\bt;\Delta 
 \vdash (\varphi'_0 \conj \varphi_2) \imp \varphi
\\ \Theta;\Sigma;\Gamma^L,  u_0: \bt, u_2:\bt, i: \bt;\Gamma;\Delta
\vdash  \varphi \jok 
\\\fv(e)\cup\fv(c_1)\cup\fv(c_2)\subseteq\dom(\Gamma)
}{ u_0: \bt, u_2:\bt, i: \bt;  \Theta; \Sigma;\Gamma^L; \Gamma;\Delta 
     \vdash_Q \cif\ e\ \cthen\ c_1\ \celse\ c_2:\varphi}\and
\inferrule*[right=ImpI]{ u_1: \bt, u_2:\bt, i: \bt;  \Theta; \Sigma;\Gamma^L; \Gamma;\Delta 
    , \varphi_1 \vdash_Q c: \varphi_2}
{ u_1: \bt, u_2:\bt, i: \bt;  \Theta; \Sigma;\Gamma^L; \Gamma;\Delta 
     \vdash_Q c: \varphi_1\imp\varphi_2}
\end{mathpar}

\paragraph{Misc}

\begin{mathpar}
\inferrule*[right=Pair]{ k\in[1,2]
\\ u_1, u_2, i;  \Theta; \Sigma;\Gamma^L; \Gamma;\Delta \vdash_Q c: (\eta_1, \eta_2)
}{ u_1, u_2, i;  \Theta; \Sigma;\Gamma^L; \Gamma;\Delta \vdash_Q c: \eta_k}
\and
\inferrule*[right=Proj]{ u_1, u_2, i;  \Theta; \Sigma;\Gamma^L; \Gamma;\Delta 
   \vdash_Q c: x{:}\tau.\varphi_1
\\ u_1, u_2, i;  \Theta; \Sigma;\Gamma^L; \Gamma;\Delta 
   \vdash_Q c: \varphi_2 }
{ u_1, u_2, i;  \Theta; \Sigma;\Gamma^L; \Gamma;\Delta 
    \vdash_Q c: (x{:}\tau.\varphi_1, \varphi_2)}\and
 \inferrule*[right=CutC]{
 \Theta;\Sigma;\Gamma^L,\Xi;\Gamma;\Delta_1  \vdash \varphi \jtrue\\
 \Xi; \Theta;\Sigma;\Gamma^L;\Gamma;\Delta_1, \varphi,\Delta_2 \vdash_Q c:\eta
 }{
 \Xi; \Theta;\Sigma;\Gamma^L;\Gamma;\Delta_1,\Delta_2  \vdash_Q c:\eta
 }
\end{mathpar}
\caption{Computation typing (2)}
\label{fig:typing-comp-2}
\end{figure*}

\begin{figure*}[t!]
\begin{mathpar}
\mprset{flushleft}
\inferrule*[right=SeqEComp]{
 u_0: \bt, u_2:\bt, i: \bt; 
 \Theta; \Sigma;\Gamma^L, u_3:\bt;\cdot;\Delta , u_0\leq u_1\leq u_2 \vdash \varphi_0 \jonempty
 \\\\ u_1:\bt;\Theta; \Sigma;\Gamma^L,u_0: \bt, u_2:\bt,u_3:\bt, i: \bt;\cdot;\varphi_0 
 \vdash_{Q1} e_1 : \tcomp(\tpair{u_b, u_e, j}{x{:}\tau.\varphi_1}{\varphi'_1})
\\\\\mbox{let}~\gamma=[u_1, u_2, i/u_b, u_e, j]
   \\\\ u_2, u_3, i; \Theta; \Sigma;\Gamma^L, u_0: \bt, u_1:\bt;\cdot;\Delta,  
 u_2 < u_3,\varphi_0 , \varphi_1\gamma \vdash_{Q2} 
   c_2: y{:}\tau'.\varphi_2
\\\\
  \Theta; \Sigma;\Gamma^L, u_0: \bt, u_3:\bt, i: \bt;\Gamma,
  u_1{:}\bt,u_2{:}\bt, y:\tau' ;\Delta  \vdash 
 (\varphi_0 \conj \varphi_1\gamma \conj
  \varphi_2) \Rightarrow \varphi \jtrue 
%\\ \Theta;\Sigma;\Gamma^L;u_3:\bt, i: \bt;\Gamma\vdash \tau'\jok 
\\  \Theta;\Sigma;\Gamma^L, u_0: \bt, u_3:\bt, i: \bt,\Gamma, y:\tau'
 \vdash  \varphi\jok
}{
  u_0: \bt, u_3:\bt, i: \bt; \Theta; \Sigma;\Gamma^L; \Gamma;\Delta  \vdash_{Q2}
   (e_1;c_2): y{:}\tau'.\varphi }
\and
 \inferrule*[right=SeqCComp]{
    u_0: \bt, u_1:\bt, i: \bt; 
 \Theta; \Sigma;\Gamma^L, u_3:\bt;\cdot;\Delta , u_0\leq u_1 \vdash
 \varphi_0 \jonempty
   \\ u_1: \bt, u_2:\bt, i: \bt; \Theta; \Sigma;\Gamma^L, u_0:
   \bt,  u_3:\bt;\cdot;\varphi_0 \vdash_Q c_1 : x{:}\tau.\varphi_1
  \\\\ u_2: \bt, u_3:\bt, i: \bt; \Theta; \Sigma;\Gamma^L;
   u_0: \bt, u_1:\bt,\cdot;\Delta,u_2 < u_3,\varphi_0, \varphi_1 
    \vdash_{Q2} 
    c_2: y{:}\tau'.\varphi_2
\\\\
  \Theta; \Sigma;\Gamma^L;  u_0: \bt, u_3:\bt, i: \bt;\cdot,
   u_1{:}\bt, u_2{:}\bt, y:\tau';\Delta  \vdash 
  (\varphi_0 \conj \varphi_1 \conj
  \varphi_2) \Rightarrow \varphi \jtrue 
%\\\\ \Theta;\Sigma;\Gamma^L;u_0: \bt, u_3:\bt, i: \bt,\Gamma\vdash  \tau'\jok 
\\  \Theta;\Sigma;\Gamma^L;  u_0: \bt, u_3:\bt, i: \bt,\Gamma,
y:\tau' \vdash  \varphi \jok
}{
  u_0: \bt, u_3:\bt, i: \bt; \Theta; \Sigma;\Gamma^L; \Gamma;\Delta  \vdash_{Q2}
   (c_1; c_2): y{:}\tau'.\varphi
}
\and
\inferrule*[right=SeqEIComp]{
 \Theta ;\Sigma;\Gamma^L,  u_0: \bt, u_3:\bt, i: \bt;\Gamma;\Delta\vdash  \varphi\jok
\\  u_0: \bt, u_2:\bt, i: \bt; 
 \Theta; \Sigma;\Gamma^L, u_3:\bt;\cdot;\Delta , u_0\leq u_1\leq u_2 \vdash \varphi_0 \jonempty
\\\\  u_0: \bt, u_3:\bt, i: \bt; 
 \Theta; \Sigma;\Gamma^L; \cdot;\Delta , u_0\leq u_3 \vdash
 \varphi'_0 \jonempty
\\\\  u_1: \bt; \Theta; \Sigma;\Gamma^L, u_0: \bt, u_2:\bt,u_3:\bt,
i: \bt; \cdot;\varphi_0
 \vdash_Q e_1 : \tcomp(\tpair{u_b, u_e, j}{x{:}\tau.\varphi_1}{\varphi'_1})
  \\ u_2: \bt, u_3:\bt, i: \bt; \Theta; \Sigma;\Gamma^L,  u_0: \bt;\cdot;\Delta,
   u_1:\bt;
  u_1 < u_2\leq u_3,\varphi_0, \varphi_1[u_1, u_2, i/u_b, u_e, j] \vdash_Q 
   c_2: \varphi_2
\\\\
 \Theta; \Sigma;\Gamma^L,  u_0: \bt, u_3:\bt, i: \bt;\Gamma;\Delta  \vdash 
   \varphi'_0 \imp \varphi \jtrue 
\\\\
 \Theta; \Sigma;\Gamma^L,  u_0: \bt, u_3:\bt, i: \bt;\Gamma, u_1{:}\bt;\Delta  \vdash 
\varphi_0 \conj \varphi'_1[u_1, u_3, i/u_b, u_e, j] \imp \varphi \jtrue 
\\\\
  \Theta; \Sigma;\Gamma^L,  u_0: \bt, u_3:\bt, i: \bt;\Gamma, 
 u_1{:}\bt, u_2{:}\bt;\Delta  \vdash 
  (\varphi_0
  \conj \varphi_1[u_1, u_2, i/u_b, u_e, j] \conj
  \varphi_2) \imp \varphi \jtrue 
}{
  u_0: \bt, u_3:\bt, i: \bt; \Theta; \Sigma;\Gamma^L; \Gamma;\Delta  \vdash_Q
   (e_1;c_2): \varphi }\and
 \inferrule*[right=SeqCIComp]{
 \Theta ;\Sigma;\Gamma^L,  u_0: \bt, u_3:\bt, i: \bt;\Gamma;\Delta\vdash  \varphi \jok
\\  u_0: \bt, u_1:\bt, i: \bt; 
 \Theta; \Sigma;\Gamma^L, u_3:\bt;\cdot;\Delta , u_0\leq u_1 \vdash \varphi_0 \jonempty
\\\\  u_0: \bt, u_3:\bt, i: \bt; 
 \Theta; \Sigma;\Gamma^L; \cdot;\Delta , u_0\leq u_3 \vdash
 \varphi'_0 \jonempty
 \\u_1: \bt, u_2:\bt, i: \bt;  \Theta; \Sigma;\Gamma^L,u_0:
 \bt, u_3:\bt;\cdot;\varphi_0 \vdash_Q c_1 : x{:}\tau.\varphi_1
\\u_1: \bt, u_3:\bt, i: \bt;  \Theta; \Sigma;\Gamma^L;\cdot;\Delta, u_0:
 \bt, u_1 \leq u_3, \varphi_0 \vdash_Q c_1 : \varphi'_1
  \\ u_2: \bt, u_3:\bt, i: \bt; \Theta; \Sigma;\Gamma^L; \cdot;\Delta,
   u_0: \bt, u_1:\bt, x: \tau, u_2\leq u_3,\varphi_0, \varphi_1 \vdash_Q 
   c_2: \varphi_2
\\\\
 \Theta; \Sigma;\Gamma^L,  u_0: \bt, u_3:\bt, i: \bt;\Gamma;\Delta  \vdash 
   \varphi'_0 \imp \varphi \jtrue 
\\\\
 \Theta; \Sigma;\Gamma^L,  u_0: \bt, u_3:\bt, i: \bt;\Gamma,u_1{:}\bt;\Delta  \vdash 
  (\varphi_0 \conj \varphi'_1) \imp \varphi \jtrue 
\\\\
  \Theta; \Sigma;\Gamma^L,  u_0: \bt, u_3:\bt, i: \bt;\Gamma,
  u_1{:}\bt, u_2{:}\bt;
\Delta  \vdash (\varphi_0\conj \varphi_1 \conj
  \varphi_2) \imp \varphi \jtrue 
}{
  u_0: \bt, u_3:\bt, i: \bt; \Theta; \Sigma;\Gamma^L; \Gamma;\Delta  \vdash_Q
   (c_1;c_2): \varphi }\and
\end{mathpar}
\label{fig:comp-typing-sequencing}
\caption{Sequential composition}
\end{figure*}

\section{Semantics}
% \begin{tabbing}
% % Values~~~$v$~$::=$~$ \btrue ~|~ \bfalse ~|~ \iota ~|~ \ell ~|~ n ~|~ \lambda x.e ~|~\Lambda X.e
% % ~|~ \ecomp(c)$
% % \\
%  $\ubertype$ $\defeq$\= $\{\cset ~|~$\= $\cset\in \pset{\{(j, \nf) ~|~
%     j\in\nat\}}\conj$
% \\\>\> $(\forall k, \nf, (k, \nf)\in \cset \conj
% j < k \conj \nf \Longrightarrow (j, \nf)\in \cset)\}$
% %\\\> $\mathrel{\cup} \{\{(j, e) ~|~ j \in \mathbb{N}\}\}$
% \end{tabbing}

\paragraph{Semantics for invariant properties}
Next we define a logical relation indexed only by an invariant
property $u_b.u_e.i.\varphi$. 
\begin{tabbing}
$\rvinv{u_b.u_e.i.\varphi}{\trace;u}~=$ $\{(k, \nf)~|~$
$\nf \neq \lambda x.e, \Lambda X.e, \ecomp(c)\} $
\\$\cup \{(k,\ecomp(c))~|~$
$(k, c)\in \rcinv{u_b.u_e.i.\varphi}{\trace;u}
\}$
\\ $\cup \{(k,\lambda x.e')~|~$\=$
 \forall j, u', j<k, u'\geq u$
\\\>
$ (j,e')\in\reinv{u_b.u_e.i.\varphi}{\trace;u'}$ 
\\\>
$\Longrightarrow(j,
e[e'/x])\in\reinv{u_b.u_e.i.\varphi}{\trace;u'}\}$
\\ $\cup$
$\{(k,\Lambda x.e)~|~\forall j, j<k$
$\Longrightarrow(j, e)\in\reinv{u_b.u_e.i.\varphi}{\trace;u}
\}$
\end{tabbing}

\begin{tabbing}
$\reinv{u_b.u_e.i.\varphi}{\trace;u}$=
\\$\{(k, e)~|$\=
$\forall 0\leq m \leq k, e\rightarrow^{m} e'\nrightarrow$
\\\>$\Longrightarrow (n-m, e')\in \rvinv{u_b.u_e.i.\varphi}{\trace;u}\}$
\end{tabbing}

\begin{tabbing}
  $\rcinv{u_b.u_e.i.\varphi }{\trace; u} =$  
 $\{ (k, c)~|~ $\\
~~ \=  
 $\forall u_B, u_E, \iota, u\leq u_B\leq u_E,$
 let $\gamma=[u_B,u_E,\iota/u_1,u_2,i]$,
\\\>
$j_b$ is the length of the trace from time $u_B$ to the end of
$\trace$,
\\~\> 
$j_e$ is the length of the trace from time $u_E$ to the end of $\trace$
\\~\> $k \geq j_b \geq j_e$, 
\\\> the configuration at time $u_B$ is
%\\\> 
 ~~$\steps{u_B}\sigma_{b}\rt \cdots, \langle\iota; x.c'::K; c\rangle\cdots$
\\\> between $u_B$ and $u_E$ (inclusive), the stack of 
%\\\> 
thread $i$ always 
\\\>contains prefix $x.c'{::}K$
\\\> $\Longrightarrow$
 $\trace\vDash_\theta \varphi\}\cap$ 
\\$\{ (k,  c)~|~ $
 $\forall u_B, u_E, \iota, u\leq u_B\leq u_E,$
\\\quad\=
 let $\gamma=[u_B,u_E,\iota/u_1,u_2,i]$,
\\\>
 $j_b$ is the length of the trace from time $u_B$ to the end of
$\trace$
\\~\> $j_e$ is the length of the trace from time $u_E$ to the end of $\trace$
\\~\>$k \geq j_b > j_e$, 
\\\> the configuration at time $u_1$ is
%\\\>  ~~
$\steps{u_B}\sigma_{b}\rt \cdots, \langle\iota; x.c'::K;
c\rangle\cdots$
\\\> the configuration at time $u_E$ is
%\\\>  ~~
$\steps{u_E}\sigma_{e}\rt \cdots, \langle\iota; K;
c'[e'/x]\rangle\cdots$ 
\\\>between $u_B$ and $u_E$, the stack of thread $i$
%\\\> 
always contains $x.c'{::}K$
%\\\>and  
\\\> $\Longrightarrow$ \=
%$\forall u', u'\geq u_2$, 
 $(j_e, e') \in  \reinv{u_b.u_e.i.\varphi}{\trace; u_E} $
 and $\trace\vDash_\theta \varphi[e'/x]\}$ 
 \end{tabbing}

\begin{tabbing}
  $\rfinv{u_b.u_e.i.\varphi}{\trace;u} =$ \\
  $\{ (k, c)~|~ $\= 
  $ \forall e$, $(k, e)\in \reinv{u_b.u_e.i.\varphi}{\trace;u}
 \Longrightarrow $ 
\\\> $(k, c\;
  e)\in$ \=$\rcinv{u_b.u_e.i.\varphi}{\trace;u}\}$ 
 \end{tabbing}

\paragraph{Semantics for invariant indexed types}
Figure~\ref{fig:re-inv-rel} summaries the interpretation of types
indexed by the invariant property $u_b.u_e.i.\varphi$. The invariant
property is used to constrain the behavior of expressions that
evaluate to normal forms that do not agree with their types. 

\begin{figure*}[t!]
\[
\begin{array}{l@{~}c@{~}l}
 \rvi{\texp}{\theta;\trace;u}{u_b.u_e.i.\varphi} &=&  \{(k, \nf) ~|~ k \in \mathbb{N} \}
\\
 \rvi{X}{\theta;\trace;u}{u_b.u_e.i.\varphi} &= &  \theta(X) 
\\
\rvi{\bt}{\theta;\trace;u}{u_b.u_e.i.\varphi} &=&
\{(k, e)~|~ (k,e)\in\rvinv{u_b.u_e.i.\varphi}{\theta;\trace;u}\}
 \\
\rvi{\Pi x{:}\tau_1. \tau_2}{\theta;\trace;u}{u_b.u_e.i.\varphi} &=&
 \{(k,\lambda x.e) ~|~ 
\forall j< k,  \forall u', u'\geq u,
\forall e', (j,  e')  \in \rei{\tau_1}{\theta;\trace;u'}{u_b.u_e.i.\varphi}
\\& &\qquad\qquad
 \Longrightarrow (j, e_1[e'/x]) 
 \in\rei{\tau_2[e'/x]}{\theta;\trace;u'}{u_b.u_e.i.\varphi} \}\cup
\\ & &   
 \{(k,\nf) ~|~  \nf \neq \lambda x.e
 \Longrightarrow
 (k, \nf)\in\reinv{ u_b.u_e.i.\varphi}{\trace;u}\}
\\
 \rvi{\forall X.\tau}{\theta;\trace;u}{u_b.u_e.i.\varphi} &=& 
 \{(k,\Lambda X) ~|~ 
\forall j < k, \forall
 \cset\in \ubertype 
\Longrightarrow  (j, e')\in
 \rei{\tau}{\theta[X\mapsto\cset];\trace;u}{u_b.u_e.i.\varphi}\}\cup
\\ & &   
 \{(k,\nf) ~|~ \nf\neq\Lambda X.e
\Longrightarrow
 (k, \nf)\in\reinv{u_b.u_e.i.\varphi}{\trace;u}\}
\\
\multicolumn{3}{l}{\rvi{\tcomp(\tpair{u_1.u_2.i}{x{:}\tau.\varphi_1}{\varphi_2})}{\theta;\trace;u}{u_b.u_e.i.\varphi}~
=} 
\\\multicolumn{3}{l}{\quad\{(k,\ecomp(c))~|~
 \forall u_B, u_E, \iota, u\leq u_B\leq u_E,
\mbox{let}~\gamma=[u_B,u_E,\iota/u_1,u_2,i]}
\\\multicolumn{3}{l}{\qquad\qquad\qquad\quad
(k, c)\in 
\rci{x{:}\tau\gamma.\varphi_1\gamma}
{\theta;\trace;u_B, u_E, \iota}{ u_b.u_e.i.\varphi}
\cap \rci{\varphi_2\gamma}{\theta;\trace;u_B, u_E, \iota}{\_}
\}\cup}
\\\multicolumn{3}{l}{\quad
 \{(k,\nf) ~|~  \nf \neq\ecomp(c)
 \Longrightarrow
 (k, \nf)\in\reinv{u_1.u_2.i.\varphi}{\trace;u}\}
}
\\\\
\multicolumn{3}{l}{\rei{\tau}{\theta;\trace;u}{u_b.u_e.i.\varphi}~
=\{(k, e)~|~ 
\forall j < m, 
e\rightarrow^m_\beta e'\nrightarrow \Longrightarrow 
(k-m,e')\in\rvi{\tau}{\theta;\trace;u}{u_b.u_e.i.\varphi}
\}}
\end{array}
\]
\caption{Semantics for inv-indexed types}
\label{fig:re-inv-rel}
\end{figure*}

\begin{tabbing}
  $\rci{ x{:}\tau.\varphi }
  {\theta;\trace; u_1, u_2, i}{u_b.u_e.i.\varphi_1} =$  $\{ (k,
  c)~|~ $\\
~ \=
 $j_b$ is the length of the trace from time $u_1$ to the end of
$\trace$
\\~\> $j_e$ is the length of the trace from time $u_2$ to the end of $\trace$
\\~\>$k \geq j_b > j_e$, 
\\\> the configuration at time $u_1$ is
%\\\>  ~~
$\steps{u_1}\sigma_{b}\rt \cdots, \langle\iota; x.c'::K;
c\rangle\cdots$
\\\> the configuration at time $u_2$ is
%\\\>  ~~
$\steps{u_2}\sigma_{e}\rt \cdots, \langle\iota; K;
c'[e'/x]\rangle\cdots$ 
\\\>between $u_1$ and $u_2$, the stack of thread $i$
%\\\> 
always contains $x.c'{::}K$
%\\\>and  
\\\> $\Longrightarrow$ \=
%$\forall u', u'\geq u_2$, 
 $(j_e, e') \in  \rei{\tau}{\theta;\trace; u_2}{u_b.u_e.i.\varphi_1} $
\\\>\>  and $\trace\vDash \varphi[e'/x]\}$ 
 \end{tabbing}

\begin{tabbing}
  $\rci{\varphi }{\theta;\trace; u_1, u_2, i}{\_} =$  $\{ (k, c)~|~ $\\
~~ \=  
$j_b$ is the length of the trace from time $u_1$ to the end of
$\trace$,
\\~\> 
$j_e$ is the length of the trace from time $u_2$ to the end of $\trace$
\\~\> $k \geq j_b \geq j_e$, 
\\\> the configuration at time $u_1$ is
%\\\> 
 ~~$\steps{u_1}\sigma_{b}\rt \cdots, \langle\iota; x.c'::K; c\rangle\cdots$
\\\> between $u_1$ and $u_2$ (inclusive), the stack of 
%\\\> 
thread $i$ always 
\\\>contains prefix $x.c'{::}K$
\\\> $\Longrightarrow$
 $\trace\vDash \varphi\}$ 
 \end{tabbing}

\begin{tabbing}
  $\rfi{ \Pi x{:}\tau.u_1.u_2.i.(y{:}\tau'.\varphi, \varphi')}{\theta;\trace;u}{u_b.u_e.i.\varphi_1} =$ \\
  $\{ (k, c)~|~ $\= 
  $ \forall e,$ \=$\forall u', u_B, u_E, \iota, u\leq u'\leq
  u_B\leq u_E$, 
\\\>let $\gamma = [u_B, u_E, \iota/u_1, u_2,i]$
\\\>$(k, e)\in \rei{\tau\gamma}{\theta;\trace;u'}{u_b.u_e.i.\varphi_1}
 \Longrightarrow $ 
\\\> $(k, c\;
  e)\in$ \=$\rci{(y{:}\tau'\gamma.\varphi\gamma)[e/x]}{\theta;\trace;u_B, u_E, \iota}{u_b.u_e.i.\varphi_1}$ 
\\\>\> $\cap \rci{\varphi'\gamma[e/x]}{\theta;\trace;u_B, u_E, \iota}{}\}$
 \end{tabbing}

%\noindent\framebox{$\rai{\alpha}{\theta}$}
\begin{tabbing}
  $\rai{ \kact(u_1.u_2.i.(x{:}\tau.\varphi_1, \varphi_2))}
{\theta;\trace;u}{u_b.u_e.i.\varphi} =$ \\
$\{ (k, a)~|~ $\= $
\forall u_B, u_E, \iota,  u\leq u_B\leq u_E$,
\\\>let $\gamma = [u_B, u_E, \iota/u_1, u_2,i]$
\\\>$(k, \cact(a)) \in 
  ($\=$\rci{x{:}\tau\gamma.\varphi_1\gamma}
{\theta;\trace; u; u_B, u_E,\iota}{u_b.u_e.i.\varphi}$
\\\>$\cap   \rci{\varphi_2\gamma}
 {\theta;\trace; u; u_B, u_E,\iota}{u_b.u_e.i.\varphi} )\}$\\
\\
  $\rai{ \Pi x{:}\tau.\alpha}{\theta;\trace;u}{u_b.u_e.i.\varphi} =$ \\
  $\{ (k, a)~|~ $
 $\forall e,$\=$\forall u', , u'\geq u$, 
$(k,e)\in\rei{\tau}{\theta;\trace; u'}{u_b.u_e.i.\varphi}$
\\\>$ \Longrightarrow   (k, a\; e) \in 
\rai{\alpha[e/x]}{\theta;\trace;u'}{u_b.u_e.i.\varphi}\}$\\
\\
$\rai{ \forall X.\alpha}{\theta;\trace;u}{u_b.u_e.i.\varphi} =$
\\$\{ (k, a)~|~$\=$ 
\forall j\leq k,\forall  \cset\in \ubertype$
\\\>$\Longrightarrow (j, a\; \cdot)\in
\rai{\alpha}{\theta[X\mapsto\cset];\trace;u}{u_b.u_e.i.\varphi}\}$
\end{tabbing}

\paragraph{Formula semantics}

\[
\begin{array}{l@{~}c@{~}l}
 \interp{\texp} &=&  \{e ~|~ e~\mbox{is an expression}\}
\\
\interp{\bt} &=&
\{e~|~ e\rightarrow^* \const\}
 \\
\interp{\Pi x{:}\tau_1. \tau_2} &=&
 \{\lambda x.e ~|~ 
\forall e',e'  \in \interp{\tau_1}
 \Longrightarrow e_1[e'/x]
 \in\interp{\tau_2}\}
\end{array}
\]

\[
\begin{array}{lcl}
\trace \vDash P\;\vec{e} & ~\mbox{iff}~ & P\;\vec{e}  \in\varepsilon(\trace)
\\ 
\trace \vDash \pred{start}(I, c, U)
 & ~\mbox{iff}~ & 
 \mbox{thread $I$ has $c$ as the active} 
\\ & &\mbox{computation with an
   empty stack} 
\\ & & \mbox{at time $U$ on $\trace$}
\\
\trace \vDash \forall x{:}\tau.\varphi & ~\mbox{iff}~ &
\forall e, e \in\interp{\tau}~\mbox{implies}~
 \trace\vDash \varphi[e/x]
\\ 
 \trace \vDash \exists x{:}\tau.\varphi & ~\mbox{iff}~ &
\exists e, e\in\interp{\tau}~\mbox{and}~
  \trace\vDash \varphi[e/x]
\end{array}
\]

\section{Lemmas}
% \begin{lem}~\label{lem:stuck-tm-in-inv}
% If $\nf\neq\lambda x.e$ or $\Lambda X.e$ or $\ecomp(c)$, then 
% $(n, nf)\in\rvinv{\Phi}{\trace;u}$
% \end{lem}

\begin{lem}[$\mathcal{R}_\mathit{INV}$ is downward-closure]
~\label{lem:rinv-downward-closed}
%~\label{lem:validity-type}
\begin{enumerate}
\item If $(k, c)\in\rvinv{\Phi}{\trace;u}$ 
then  $\forall j{<}k$, $(j, c)\in\rvinv{\Phi}{ \trace;u}$
\item If $(k, c)\in\reinv{\Phi}{\theta, \trace; u}$ 
then  $\forall j{<}k$, $(j, c)\in\reinv{\Phi}{\trace;u}$
\item If $(k, c)\in\rcinv{\Phi}{\trace; u}$ 
then  $\forall j{<}k$, $(j, c)\in\rcinv{\Phi}{\trace;u}$

\end{enumerate}
\end{lem}
\begin{proofsketch}
By examining the definition of the relations. 
\end{proofsketch}

\begin{lem}[$\mathcal{R}_\mathit{INV}$ is closed under delay]
~\label{lem:rinv-closed-delay}
\begin{enumerate}
\item If
$(k, e)\in\rvinv {\Phi}{\trace;u}$ then $\forall u'{>}u$,
$(k, e)\in\rvinv{\Phi}{\trace;u'}$
\item If
$(k, e)\in\reinv {\Phi}{\trace;u}$ then $\forall u'{>}u$,
$(k, e)\in\reinv{\Phi}{\trace;u'}$
\item If
$(k, e)\in\rcinv {\Phi}{\trace;u}$ then $\forall u'{>}u$,
$(k, e)\in\rcinv{\Phi}{\trace;u'}$\end{enumerate}
\end{lem}
\begin{proofsketch}
By examining the definitions.
\end{proofsketch}

\begin{lem}[Indexed types are confined]
~\label{lem:cnfn-confine}
$\cnfn{\tau}{u_b.u_e.i.\varphi}$ implies
\begin{enumerate}
\item $\rvi{\tau}{\theta;\trace;u}{u_b.u_e.i.\varphi} =
\rvinv{u_b.u_e.i.\varphi}{\theta;\trace;u}$. 
\item 
$\rei{\tau}{\theta;\trace;u}{u_b.u_e.i.\varphi} =
\reinv{u_b.u_e.i.\varphi}{\trace;u}$. 
\item for all $n$, $c$,
($\forall u_B,u_E,I$ s.t. $u\leq u_B\leq u_E$, 
$(n,c)\in\rci{\tau.\varphi[u_B,u_E,I/u_b,u_e,i]}
 {\theta;\trace;u_B,u_E,I}{u_b.u_e.i.\varphi}
\\\cap
\rci{\varphi[u_B,u_E,I/u_b,u_e,i]}
 {\theta;\trace;u_B,u_E,I}{u_b.u_e.i.\varphi}$) 
\\iff
$(n,c)\in\rcinv{u_b.u_e.i.\varphi}{\trace;u}$
\end{enumerate}
\end{lem}
\begin{proof}
By induction on $\tau$. 2 uses 1 directly, 1 uses 2 when $\tau$ is
smaller, 3 uses 2 directly, and 1 uses 3 when $\tau$ is smaller.
\\Proof of 1.
\begin{description}
\item[case:] $\tau=b$. Follows directly from the definitions
\item[case:] $\tau=\Pi x:\tau_1.\tau_2$
\begin{tabbing}
By assumptions
\\\quad\=
$\cnfn{\tau_1}{u_b.u_e.i.\varphi}$  and
$\cnfn{\tau_1}{u_b.u_e.i.\varphi}$
\`(1)
\\Assume 
\\\>$(n,\nf)\in 
\rvi{\Pi x:\tau_1.\tau_2}
{\theta;\trace;u}{u_b.u_e.i.\varphi}$
\`(2)
\\To show: $(n,\nf)\in\rvinv{u_b.u_e.i.\varphi}{\trace;u}$
\\We first consider the case when $\nf=\lambda x.e_1$
\\ Given $0 \leq j< n$, $u'\geq u$ $(j, e')
  \in\reinv{u_b.u_e.i.\varphi}{\trace;u'}$
\\ By I.H. on $\tau_1$ 
\\\> $(j, e')\in\rei{\tau_1}{\theta;\trace;u'}{u_b.u_e.i.\varphi}$
\\By (2)
\\\> $(j, e_1[e'/x])\in\rei{\tau_2[e'/x]}
{\theta;\trace;u'}{u_b.u_e.i.\varphi}$  
\`(3)
\\ By I.H. on $\tau_2$ and (3)
\\\> $(j, e_1[e'/x])\in\reinv{u_b.u_e.i.\varphi}{\trace;u'}$
\`(4)
\\By (4)
\\\> $(n, \lambda x.e_1)\in\rvinv{u_b.u_e.i.\varphi}{\trace;u}$
\\Next we consider the case where $\nf = \Lambda X.e_1$ or $\ecomp(c)$
\\\> this follows from the definition directly
\\Proofs for the other direction is similar
% \\Assume 
% \\\>$(n,e)\in\reinv{u_b.u_e.i.\varphi}{\trace;u}$
% \`(6) 
% \\To show $(n,e)\in  \rei{\Pi x:\tau_1.\tau_2}
% {\theta;\trace;u}{u_b.u_e.i.\varphi}$
% \\The prove is similar to the above case
\end{tabbing}

\item[case:]
  $\tau=\tcomp(\tpair{u_b.u_e.i}{x{:}\tau.\varphi}{\varphi})$
\begin{tabbing}
By assumption
\\\quad\=
 $\cnfn{\tau}{u_b.u_e.i.\varphi}$
\`(1)
\\
Assume
\\\>
$(n,\nf)\in\rvi{\tcomp(\tpair{u_b.u_e.i}{x{:}\tau.\varphi}{\varphi})}{\theta;\trace;u}{u_b.u_e.i.\varphi}$
\`(2)
\\
To show $(n,\nf)\in\rvinv{u_b.u_e.i.\varphi}{\trace;u}$
\\ We show the case when $\nf = \ecomp(c)$, the other cases are
trivial
%\\To show $(n-m, c)\in\rcinv{u_b.u_e.i.\varphi}{\trace;u}$
\\By definitions, $\forall u_B, u_E, \iota, u\leq u_B\leq u_E$,
\\ let $\gamma=[u_B,u_E,\iota/u_b,u_e,i]$
\\\> $(n, c)\in$\=
$\rci{x{:}\tau\gamma.\varphi_1\gamma}{\theta;\trace;u_B, u_E, \iota}{
  u_b.u_e.i.\varphi}$
\\\>\>
 $\cap \rci{\varphi_2\gamma}{\theta;\trace;u_B, u_E, \iota}{\_}$
\`(3)
\\ By I.H.  and (3)
\\\>
 $(n, c)\in\rcinv{u_b.u_e.i.\varphi}{\trace;u}$
\`(4)
\\ By (4)
\\\> $(n,\nf)\in\rvinv{u_b.u_e.i.\varphi}{\trace;u}$
\`(5)
\\The proof of the other direction is similar
\end{tabbing}
\end{description}
3 is proven straightforwardly by expanding the definitions of the two
relations.
% \begin{description}
% \item[case:] $\Longrightarrow$ direction
% \begin{tabbing}
% Assume $\forall u_B,u_E,I$ s.t. $u\leq u_B\leq u_E$, 
% \\
% $(n,c)\in$\=$\rci{\tau.\varphi[u_B,u_E,I/u_b,u_e,i]}
%  {\theta;\trace;u_B,u_E,I}{u_b.u_e.i.\varphi}$
% \\\>$\cap
% \rci{\varphi[u_B,u_E,I/u_b,u_e,i]}
%  {\theta;\trace;u_B,u_E,I}{\_}$
% \\
% To show: $(n,c)\in\rcinv{u_b.u_e.i.\varphi}{\trace;u}$
% \\ 
% \end{tabbing}
% \item[case:] $\Longleftarrow$ direction

% \end{description}
\end{proof}

\begin{lem}[Invariant confinement]~\label{lem:inv-confine}

$\varphi$ is composable, 
and thread $\iota$ is silent between time $u_B$
and $u_E$ implies $\trace\vDash\varphi[u_B,u_E,\iota/u_b,u_e,i]$
\begin{enumerate}
% %%% e, rv
% \item If $\fa(e)=\emptyset$, $\fv(e)\in\dom(\gamma)$,
% $(n,\gamma)\in\reinv{u_b.u_e.i.\varphi}{\trace;u}$
% \\then $(n, e\gamma)\in\reinv{u_b.u_e.i.\varphi}{\trace;u}$
%%% e
\item If $\fa(e)=\emptyset$, $\fv(e)\in\dom(\gamma)$,
$(n,\gamma)\in\reinv{u_b.u_e.i.\varphi}{\trace;u}$
\\then $(n, e\gamma)\in\reinv{u_b.u_e.i.\varphi}{\trace;u}$
%%% c
\item If $\fa(c)=\emptyset$, $\fv(c)\in\dom(\gamma)$,
$(n,\gamma)\in\reinv{u_b.u_e.i.\varphi}{\trace;u}$
\\then $(n, c\gamma)\in\rcinv{u_b.u_e.i.\varphi}{\trace;u}$
%%% fix
\item If $\fa(c)=\emptyset$, $\fv(\cfix f(x).c)\in\dom(\gamma)$,
\\$(n,\gamma)\in\reinv{u_b.u_e.i.\varphi}{\trace;u}$
\\then $(n,\cfix f(x).c\gamma)\in\rfinv{u_b.u_e.i.\varphi}{\trace;u}$
\end{enumerate}
\end{lem}
\begin{proof}
By induction on the structure of the terms. 3 needs a sub-induction on
$n$. We show a few key cases.
\\Proof of 1.
\begin{description}
\item[case:] $e=e_1\;e_2$
\begin{tabbing}
\\By I.H.
\\\quad\=
 $(n,e_1\gamma)\in \reinv{u_b.u_e.i.\varphi}{\trace;u}$
\`(1)
\\\>
 $(n,e_2\gamma)\in \reinv{u_b.u_e.i.\varphi}{\trace;u}$
\`(2)
\\Assume $(e_1e_2)\gamma\rightarrow^{m} \nf\nrightarrow$
\\\> $e_1\gamma\rightarrow^j\nf_1\nrightarrow$
\`(3)
\\We consider two cases: $\nf_1=\lambda x.e$ and $\nf_1\neq\lambda
x.e$
\\Subcase $\nf_1=\lambda x.e$:
\\By (1)
\\\> $(n-j, \lambda x.e)\in
\rvinv{u_b.u_e.i.\varphi}{\trace;u}$
\`(4)
\\By (2) and Lemma~\ref{lem:rinv-downward-closed}
\\\> $(n-j-1, e_2\gamma)\in
\reinv{u_b.u_e.i.\varphi}{\trace;u}$
\`(5)
\\By (4) and (5)
\\\> $(n-j-1, e[e_2\gamma/x])\in
\reinv{u_b.u_e.i.\varphi}{\trace;u}$
\`(6)
\\By (6)
\\\> $(n,(e_1e_2)\gamma)\in
\reinv{u_b.u_e.i.\varphi}{\trace;u}$
\`(7)
\\Subcase $\nf_1\neq\lambda x.e$:
\\\> $(e_1e_2)\gamma\rightarrow^{m} \nf_1(e_2\gamma)\nrightarrow$
\`(8)
\\By definitions 
\\\> $(n,(e_1e_2)\gamma)\in
\reinv{u_b.u_e.i.\varphi}{\trace;u}$
\`(9)
\end{tabbing}
\end{description}

Proof of 3 is by sub-induction on $n$
\begin{description}
\item[case:] $n=0$
\\The fixpoint couldn't have returned. We only need to show
that the trace satisfies $\varphi$. This is true because the thread
executing the fixpoint is silent. 
\item[case:] $n=k+1$
\begin{tabbing}
\\Assume that $(k, \cfix f(x).c\gamma)
\in\rfinv{u_b.u_e.i.\varphi}{\trace;u}$
\`(1)
\\To show $(k+1, \cfix f(x).c\gamma)
\in\rfinv{u_b.u_e.i.\varphi}{\trace;u}$
\\
 $ \forall e$,
$(k+1, e)\in \reinv{u_b.u_e.i.\varphi}{\trace;u}$
\\To show $(k+1, c\;
  e)\in\rcinv{u_b.u_e.i.\varphi}{\trace;u}$ 
\\By (1), 
\\\quad\=
 $(k,\lambda z.\ecomp((\cfix f(x).c\gamma)\, z))
 \in\reinv{u_b.u_e.i.\varphi}{\trace;u}$
\`(2)
\\By I.H. on $c$ and Lemma~\ref{lem:rinv-downward-closed}
 and~\ref{lem:rinv-closed-delay}
\\\> $(k, c[\lambda z.\ecomp((\cfix f(x).c\gamma)\, z)/f][e/x])$ 
\\\>$\in\rcinv{u_b.u_e.i.\varphi}{\trace;u}$
\`(3)
\\Assume thread $\iota$ executes the fixpoint, 
\\we consider the following time intervals:
\\(i)Before the fixpoint is unrolled, 
\\(ii) the body of the fixpoint is
evaluated, 
\\(iii) the fixpoint returns $e_1$
\\By $\iota$ is silent in (i)
\\\>$\varphi$ holds in (i)
\`(4)
\\By (3) and $\varphi$ is composable,
\\\>$\varphi$ holds in (ii) and (iii) 
\\\> and $(j_e, e_1)\in\reinv{u_b.u_e.i.\varphi}{\trace;u_E}$
\\\> where $u_e$ is the time when $e_1$ is returned 
\\\>and $j_e$ is the
length of $\trace$ from $u_e$ till the end of $\trace$
\`(5)
\\By (4) and (5)
\\\>$(k+1, \cfix f(x).c\gamma)
\in\rfinv{u_b.u_e.i.\varphi}{\trace;u}$
\end{tabbing}
\end{description}
\end{proof}

\section{Properties of Interpretation of Types}

% \begin{lem}[strenghtening]
% ~\label{lem:theta-weakening}% todo: change label
% If $\ftv(\tau) \subseteq\dom(\theta)$ then
% \begin{enumerate}
% \item $\rvi{\tau}{\theta[X \mapsto \cset];\trace;u}{\Phi} =
% \rvi{\tau}{\theta;\trace;u}{\Phi}$
% \item $\rei{\tau}{\theta[X \mapsto \cset];\trace;u}{\Phi} =
% \rei{\tau}{\theta;\trace;u}{\Phi}$
% \item $\rci{\eta}{\theta[X \mapsto \cset], \trace, \Xi} =
% \rci{\eta}{\theta, \trace, \Xi}$ 
% \item $\trace\vDash_{\theta[X \mapsto \cset]}\varphi$ iff 
%  $\trace\vDash_\theta\varphi$
% %% \item $\rfi{\eta_c}{\theta[X \mapsto \cset]} =
% %% \rfi{\eta_c}{\theta}$ where $\cset = \rvi{\tau_1}{\theta}$
% \item $\rai{\alpha}{\theta[X \mapsto \cset];\trace;u}{\Phi} =
%  \rai{\alpha}{\theta;\trace;u}{\Phi}$ 
% \end{enumerate}
% \end{lem}
% \begin{proofsketch}
% By induction on the structure of $\tau$, $\eta$, $\varphi$ and $\alpha$.
% \end{proofsketch}

\begin{lem}~\label{lem:stuck-tm-in-inv}
If $\nf\neq\lambda x.e$ or $\Lambda X.e$ or $\ecomp(c)$, 
then 
$(n, nf)\in\rvi{\tau}{\theta;\trace;u}{\Phi}$
\end{lem}
\begin{proofsketch}
Case on $\tau$. For all cases except when $\tau=X$, the conclusion
follows from the definition of $\rvinv{\Phi}{\trace;u}$.

When $\tau=X$, $\theta(X)\in\ubertype$. By the definition of
$\ubertype$, every $\cset\in\ubertype$ contains all stuck terms that
are not functions or suspended computations.
\end{proofsketch}

\begin{lem}[Substitution]
~\label{lem:subst}
If $\cset = \rvi{\tau_1}{\theta;\trace;u}{\Phi}$ then
\begin{enumerate}
\item $\rvi{\tau}{\theta[X \mapsto \cset];\trace;u}{\Phi} =
\rvi{\tau[\tau_1/X]}{\theta;\trace;u}{\Phi}$
\item $\rei{\tau}{\theta[X \mapsto \cset];\trace;u}{\Phi} =
\rei{\tau[\tau_1/X]}{\theta;\trace;u}{\Phi}$
\item $\rci{\eta}{\theta[X \mapsto \cset], \trace, \Xi}{\Phi} =
\rci{\eta[\tau_1/X]}{\theta, \trace, \Xi}{\Phi}$ 
\item $\rai{\alpha}{\theta[X \mapsto \cset];\trace;u}{\Phi} =
 \rai{\alpha[\tau_1/X]}{\theta;\trace;u}{\Phi}$ 
\end{enumerate}
\end{lem}
\begin{proofsketch}
By induction on the structure of $\tau$, $\eta$, $\varphi$ and $\alpha$.
\end{proofsketch}

\begin{lem}[Downward-closure]
~\label{lem:downward-closed}
%~\label{lem:validity-type}
\begin{enumerate}
\item If $(k, c)\in\rci{\eta}{\theta, \trace, \Xi}{\Phi}$ 
then  $\forall j{<}k$, $(j, c)\in\rci{\eta}{\theta, \trace, \Xi}{\Phi}$

\item If $\ftv(\tau)\subseteq\dom(\theta)$,
$\forall X \in\dom(\theta)$, $\theta(X)\in\ubertype$,
and $(k, e)\in\rvi{\tau}{\theta;\trace;u}{\Phi}$, 
 then $\forall j{<}k$, $(j, e)\in\rvi{\tau}{\theta;\trace;u}{\Phi}$.

\item If   $\ftv(\tau)\subseteq\dom(\theta)$, $\forall X
  \in\dom(\theta)$, $\theta(X)\in\ubertype$, and
$(k, e)\in\rei{\tau}{\theta;\trace;u}{\Phi}$, 
 then $\forall j{<}k$, $(j, e)\in\rei{\tau}{\theta;\trace;u}{\Phi}$.
\end{enumerate}
\end{lem}
\begin{proofsketch}
By examining the definitions. Proofs of 3 uses proofs of 2 and 2 uses 1.
\end{proofsketch}

\begin{lem}[Substitutions are closed under index reduction]
~\label{lem:ctx-smaller}

\noindent If 
$\ftv(\Gamma)\subseteq\dom(\theta)$, 
$\forall X \in\dom(\theta)$, $\theta(X)\in\ubertype$, 
$(n,\gamma)\in\rgi{\Gamma}{\theta;\trace;u}{\Phi}$, and $j<n$ 
then
$(j,\gamma)\in\rgi{\Gamma}{\theta;\trace;u}{\Phi}$.
\end{lem}
\begin{proofsketch}
By induction on the structure of $\Gamma$, 
using Lemma~\ref{lem:downward-closed}.
\end{proofsketch}

% \begin{lem}
% $e_1\rightarrow^*_\beta e_2\nrightarrow$ implies 
% \begin{enumerate}
% \item $\trace\vDash_\theta \pred{start}(\iota, c, u)[e_1/x]$ iff
%   $\trace\vDash_\theta \pred{start}(\iota, c, u)[e_2/x]$
% \item $\trace\vDash_\theta (e\equiv e')[e_1/x]$ iff
%   $\trace\vDash_\theta (e\equiv e')[e_2/x]$
% \end{enumerate}
% \end{lem}

\begin{lem}[Validity of types]
~\label{lem:validity-type}
 If 
 $\ftv(\tau)\subseteq\dom(\theta)$
and $\forall X \in\dom(\theta)$, $\theta(X)\in\ubertype$, 
  then ${\rvi{\tau}{\theta;\trace;u}{\Phi}} \in\ubertype$
\end{lem}
\begin{proofsketch}
By Lemmas~\ref{lem:downward-closed}.
\end{proofsketch}

\begin{lem}[Closed under delay]~\label{lem:closed-delay}
\begin{enumerate}
\item If
$(k, e)\in\rvi{\tau}{\theta;\trace;u}{\Phi}$ and $u'>u$ 
 then  $(k, e)\in\rvi{\tau}{\theta;\trace;u'}{\Phi}$.
\item If
$(k, e)\in\rei{\tau}{\theta;\trace;u}{\Phi}$ and $u'>u$ 
then $(k, e)\in\rei{\tau}{\theta;\trace;u'}{\Phi}$.
\end{enumerate}
\end{lem}
\begin{proofsketch}
By examining the definitions and use Lemma~\ref{lem:rinv-closed-delay}
\end{proofsketch}

\begin{lem}[Substitutions are closed under delay]
~\label{lem:ctx-closed-under-delay}
\noindent If 
$(n,\gamma)\in\rg{\Gamma}{\theta;\trace;u}{\Phi}$ and
$u'>u$
then
$(n,\gamma)\in\rg{\Gamma}{\theta;\trace;u'}{\Phi}$.
\end{lem}
 \begin{proofsketch}
 By induction on the structure of $\Gamma$, 
 using Lemma~\ref{lem:closed-delay}.
 \end{proofsketch}

\section{Soundness}

\begin{thm}[Soundness]
Assume that $\forall A::\alpha \in\Sigma$, $\forall \Phi, \trace, n, u, (n,
A)\in\rai{\alpha}{\cdot;\trace;u}{\Phi}$, 
then
\begin{enumerate}
%%%%%%%%%%%%%%%%%%%%%%%%%%%%%%%%%%%%%%%%%%%%%%
%% |-varphi e:t  (hat(re)[[]])
%%%%%%%%%%%%%%%%%%%%%%%%%%%%%%%%%%%%%%%%%%%%%%
\item %$\varphi$ is trace composable, and holds on silent threads,
% let $\Phi = u_b.u_e.i.\varphi$,
\begin{packedenumerate}
%%% e
\item  
\begin{itemize}
\item $ \ee::u:\bt; \Theta;\Sigma;\Gamma^L;\Gamma;\Delta \vdash_\Phi e : \tau$, 
\item $\forall \theta\in\rd{\Theta}{}$, 
\item $\forall \gamma^L\in\interp{\Gamma^L}$,
\item $\forall U, U', U'\geq U$, let $\gamma_u=[U/u]$, 
\item  $\forall \trace$,
 $\forall n, \gamma$,   $(n; \gamma)\in \rgi{\Gamma\gamma_u\gamma^L}{\theta;\trace;U'}{\Phi}$,
\item $\trace\vDash\Delta\gamma\gamma_u\gamma^L$
\end{itemize}
 implies $(n; e\gamma) \in \rei{\tau\gamma\gamma_u\gamma^L}
 {\theta;\trace;U'}{\Phi}$
%%% c
\item 
\begin{itemize}  
\item $ \ee:: u_1, u_2, i; \Theta;\Sigma;\Gamma^L;\Gamma;\Delta \vdash_\Phi c : \eta$, 
\item $\forall$ $u$, $u_B$, $u_E$, $\iota$ s.t. $u\leq u_B\leq u_E$,  
let $\gamma_1 = [u_B, u_E,\iota/u_1,u_2,i]$
\item $\forall \theta\in\rd{\Theta}{}$, 
\item $\forall \gamma^L\in\interp{\Gamma^L}$,
\item $\forall \trace$,
   $\forall n, \gamma, (n; \gamma)\in
 \rgi{\Gamma\gamma_1\gamma^L}{\theta;\trace;u}{\Phi}$,
\item $\trace\vDash\Delta\gamma\gamma_1\gamma^L$ 
\end{itemize}
implies $(n; c\gamma) \in \rci{{\eta}\gamma\gamma_1\gamma^L}{\theta;
   \trace; u_B,u_E,\iota}{\Phi}$
%%%% fixpoint
\item 
\begin{itemize} 
\item $ \ee::u:\bt; \Theta;\Sigma;\Gamma^L;\Gamma;\Delta  \vdash_\Phi c :
  \eta_c$,  
\item $\forall \theta\in\rd{\Theta}{}$, 
\item $\forall \gamma^L\in\interp{\Gamma^L}$,
\item $\forall U, U', U'\geq U$, let $\gamma_u=[U/u]$, 
\item  $\forall \trace$,
 $\forall n, \gamma$,   $(n; \gamma)\in \rgi{\Gamma\gamma_u\gamma^L}{\theta;\trace;U'}{\Phi}$,
\item  $\trace\vDash\Delta\gamma\gamma_u\gamma^L$
\end{itemize}
 implies
  $(n; c\gamma) \in \rfi{\eta_c\gamma\gamma_u\gamma^L}
{\theta;\trace;U'}{\Phi}$
%%%% a
\item 
\begin{itemize}
\item $ \ee::u:\bt; \Theta;\Sigma;\Gamma^L;\Gamma;\Delta  \vdash_\Phi
  a : \alpha$, 
\item $\forall \theta\in\rd{\Theta}{}$, 
\item $\forall \gamma^L\in\interp{\Gamma^L}$,
\item $\forall U, U', U'\geq U$, let $\gamma_u=[U/u]$, 
\item  $\forall \trace$,
 $\forall n, \gamma$,   $(n; \gamma)\in \rgi{\Gamma\gamma_u\gamma^L}{\theta;\trace;U'}{\Phi}$,
\item $\trace\vDash\Delta\gamma\gamma_u\gamma^L$
\end{itemize}
implies
  $(n; a\gamma) \in \rai{\alpha\gamma\gamma_u\gamma^L}{\theta;\trace;U'}{\Phi}$
\end{packedenumerate}

%%%%%%%%%%%%%%%%%%%%%%%%%%%%%%%%%%%%%%%%%%%%%%
%% |-e:t  hat(re[[]])
%%%%%%%%%%%%%%%%%%%%%%%%%%%%%%%%%%%%%%%%%%%%%%
\item 
\begin{packedenumerate}
\item  
\begin{itemize}
\item $ \ee::u:\bt; \Theta;\Sigma;\Gamma^L;\Gamma;\Delta \vdash e : \tau$, 
\item $\forall \theta\in\rd{\Theta}{}$, 
\item $\forall \gamma^L\in\interp{\Gamma^L}$,
\item $\forall U, U', U'\geq U$, let $\gamma_u=[U/u]$, 
\item  $\forall \trace$, $\forall\Phi$,
 $\forall n, \gamma$,   $(n; \gamma)\in \rgi{\Gamma\gamma_u\gamma^L}{\theta;\trace;U'}{\Phi}$,
\item $\trace\vDash\Delta\gamma\gamma_u\gamma^L$
\end{itemize}
 implies $(n; e\gamma) \in \rei{\tau\gamma\gamma_u\gamma^L}
 {\theta;\trace;U'}{\Phi}$
%%% c
\item 
\begin{itemize}  
\item $ \ee:: u_1, u_2, i; \Theta;\Sigma;\Gamma^L;\Gamma;\Delta \vdash c : \eta$, 
\item $\forall$ $u$, $u_B$, $u_E$, $\iota$ s.t. $u\leq u_B\leq u_E$,  
let $\gamma_1 = [u_B, u_E,\iota/u_1,u_2,i]$
\item $\forall \theta\in\rd{\Theta}{}$, 
\item $\forall \gamma^L\in\interp{\Gamma^L}$,
\item $\forall \trace$, $\forall\Phi$,
   $\forall n, \gamma, (n; \gamma)\in
 \rgi{\Gamma\gamma_1\gamma^L}{\theta;\trace;u}{\Phi}$,
\item $\trace\vDash\Delta\gamma\gamma_1\gamma^L$ 
\end{itemize}
implies $(n; c\gamma) \in \rci{{\eta}\gamma\gamma_1\gamma^L}{\theta;
   \trace; u_B,u_E,\iota}{\Phi}$
%%%% fixpoint
\item 
\begin{itemize} 
\item $ \ee::u:\bt; \Theta;\Sigma;\Gamma^L;\Gamma;\Delta  \vdash c :
  \eta_c$,  
\item $\forall \theta\in\rd{\Theta}{}$, 
\item $\forall \gamma^L\in\interp{\Gamma^L}$,
\item $\forall U, U', U'\geq U$, let $\gamma_u=[U/u]$, 
\item  $\forall \trace$, $\forall\Phi$,
 $\forall n, \gamma$,   $(n; \gamma)\in \rgi{\Gamma\gamma_u\gamma^L}{\theta;\trace;U'}{\Phi}$,
\item $\trace\vDash\Delta\gamma\gamma_u\gamma^L$
\end{itemize}
 implies
  $(n; c\gamma) \in \rfi{\eta_c\gamma\gamma_u\gamma^L}
{\theta;\trace;U'}{\Phi}$
%%%% a
\item 
\begin{itemize}
\item $ \ee::u:\bt; \Theta;\Sigma;\Gamma^L;\Gamma;\Delta  \vdash
  a : \alpha$, 
\item $\forall \theta\in\rd{\Theta}{}$, 
\item $\forall \gamma^L\in\interp{\Gamma^L}$,
\item $\forall U, U', U'\geq U$, let $\gamma_u=[U/u]$, 
\item  $\forall \trace$, $\forall\Phi$,
 $\forall n, \gamma$,   $(n; \gamma)\in \rgi{\Gamma\gamma_u\gamma^L}{\theta;\trace;U'}{\Phi}$,
\item $\trace\vDash\Delta\gamma\gamma_u\gamma^L$
\end{itemize}
implies
  $(n; a\gamma) \in
  \rai{\alpha\gamma\gamma_u\gamma^L}{\theta;\trace;U'}{\Phi}$
%%% silent
\item 
\begin{itemize}
\item $ \ee:: u_1, u_2, i; \Theta;\Sigma;\Gamma^L;\Gamma;\Delta \vdash \varphi
   \jonempty$, 
\item $\forall$ $u$, $u_B$, $u_E$, $\iota$ s.t. $u\leq u_B\leq u_E$, 
\item let $\gamma_1 = [u_B, u_E,\iota/u_1,u_2,i]$
\item $\forall \theta\in\rd{\Theta}{}$, 
\item $\forall \gamma^L\in\interp{\Gamma^L}$,
\item $\forall \Phi$, 
$\forall \trace$,
$\forall n, \gamma, (n; \gamma)\in 
   \rgi{\Gamma\gamma_1\gamma^L}{\theta;\trace;u}{\Phi}$, 
\item  $j_b$ is the length of $\trace$ from time $u_B$ to the end of
$\trace$, 
\item $j_e$ is the length of  $\trace$ from time $u_E$ to the end
 of $\trace$, 
\item $n \geq j_b \geq j_e$ 
\item  between time $u_B$ and time
   $u_E$,  thread $\iota$ is
   silent 
\item $\trace \vDash \Delta\gamma\gamma_1$ 
\end{itemize}
implies $\trace\vDash (\varphi \gamma\gamma_1)$
%%% phi true
\item 
\begin{itemize} 
\item $\ee::\Theta;\Sigma;\Gamma^L;\Gamma;\Delta\vdash \varphi
  \jtrue$,
\item $\forall \theta\in\rd{\Theta}{}$, 
\item $\forall \gamma^L\in\interp{\Gamma^L}$,
\item $\forall \trace$, $\forall \Phi$, 
   $\forall n, \gamma,u, (n; \gamma)\in
  \rgi{\Gamma\gamma^L}{\theta;\trace;u}{\Phi}$,
\item $\trace \vDash \Delta\gamma^L\gamma$
\end{itemize}
 implies
 $\trace \vDash \varphi\gamma^L\gamma$
\end{packedenumerate}
\end{enumerate}
\end{thm}
\begin{proof}
By induction on the structure of $\ee$.
\\Proof of 1.(a).
\begin{description}
\item[case:] \rulename{Confine}
\\\\
\(\mprset{flushleft}
\inferrule*[right=Confine]
{  
 \mbox{$\varphi$ is trace composable} 
\\ \ee'::u_b, u_e, i; \Theta;\Sigma;\Gamma^L,u;\Gamma;\Delta \vdash \varphi\jonempty 
\\ u_b{:}\bt, u_e{:}\bt, i{:}\bt \vdash
\varphi\jok
 \\\fa(e)=\emptyset 
\\\fv(e)\subseteq\Gamma
 \\\cnfn{\tau}{u_b.u_e.i.\varphi}
 \\\cnfn{\Gamma}{u_b.u_e.i.\varphi}
}
{u;\Theta;\Sigma;\Gamma^L;\Gamma;\Delta  \vdash_{u_b.u_e.i.\varphi} e : \tau}
\)
\begin{tabbing}
By assumptions 
\\
\quad \=
$\theta\in\rd{\Theta}{}$, 
$\forall \gamma^L\in\interp{\Gamma^L}$,
\\\>$\gamma_u= U/u$, $U'\geq U$, 
$\trace\vDash\Delta\gamma\gamma_u\gamma^L$
\\\>and  $(n; \gamma)\in \rgi{\Gamma\gamma_u\gamma^L}{\theta;\trace;U'}{u_b.u_e.i.\varphi}$,
\`(1)
\\By Lemma~\ref{lem:cnfn-confine} and (1)
\\\>and  $(n; \gamma)\in
\reinv{\Phi}{\trace;U'}$
\`(2)
\\By I.H. on $\ee'$, given any $\iota$, $u_B$, and $u_E$, 
\\$\iota$ is silent between $u_B$ and $u_E$ implies
\\\> $\trace\vDash \varphi[u_B,u_E,\iota/u_b,u_e,i]$
\`(3)
\\By (1) and (3) and Lemma~\ref{lem:inv-confine}
\\\> $(n,e\gamma)\in\reinv{u_b.u_e.i.\varphi}{\trace;U'}$
\`(4)
\\By (4) and Lemma~\ref{lem:cnfn-confine}
\\\> $(n; e\gamma)\in
\rei{\tau\gamma_u\gamma^L}{\theta;\trace;U'}{u_b.u_e.i.\varphi}$
\`(5)
\end{tabbing}
\end{description}

\noindent Proof of 2.(a).
\begin{description}
%%% \x.e: tau1 -> tau2
\item[case:] 
\(
\inferrule*[right=E-Fun]{\ee'::\Theta;\Sigma;\Gamma^L,u,\Gamma\vdash \tau_1 \jok 
\\u;\Theta; \Sigma;\Gamma^L;\Gamma, x: \tau_1 ;\Delta \vdash e
  : \tau_2}{u;\Theta; \Sigma;\Gamma^L;\Gamma;\Delta  \vdash \lambda x.e : \Pi x{:}\tau_1. \tau_2}
\)
\begin{tabbing}
By assumptions
\\
\quad \=
$\theta\in\rd{\Theta}{}$, 
$\forall \gamma^L\in\interp{\Gamma^L}$,
\\\>$\gamma_u= U/u$, $U'\geq U$, 
$\trace\vDash\Delta\gamma\gamma_u\gamma^L$
\\\>and  $(n; \gamma)\in \rgi{\Gamma\gamma_u\gamma^L}{\theta;\trace;U'}{\Phi}$,
\`(1)
\\ Given
 $j < k$,  $u''\geq U'$,
\\\> and $(j,  e_0) \in \rei{\tau_1\gamma\gamma_u\gamma^L}{\theta;\trace;u''}{\Phi}$
\`(2)
\\By Lemma~\ref{lem:ctx-smaller} and~\ref{lem:ctx-closed-under-delay}
\\ \>
  $(j; \gamma)\in \rgi{\Gamma\gamma_u\gamma^L}{\theta;\trace;u''}{\Phi}$ 
\`(3)
\\ By (2) and (3)
\\\>
  $(j; \gamma[x\mapsto e_0])\in 
 \rgi{(\Gamma, x:\tau_1)\gamma_u\gamma^L}{\theta;\trace;u''}{\Phi}$, 
\`(4)
\\ By I.H. on $\ee'$
\\\>
 $(j, e\gamma[x\mapsto e_0]) \in \rei{\tau_2\gamma_u\gamma^L\gamma[x\mapsto
   e_0]}{\theta;\trace;u''}{\Phi}$
\`(5)
\\ By (5) is derived based on assumption in (2)
\\ \>
 $(n, \lambda x. e\gamma) \in 
\rvi{ (\Pi x{:}\tau_1.\tau_2)\gamma\gamma_u\gamma^L}{\theta;\trace;U'}{\Phi}$ 
\`(6)
\\By (6)
\\ \>
$(n, \lambda x. e\gamma) \in \rei{ (\Pi
  x{:}\tau_1.\tau_2)\gamma\gamma_u\gamma^L}{\theta;\trace;U'}{\Phi}$
\end{tabbing}

\item[case:]
\(
\inferrule*[right=E-App]{\ee_1::u;\Theta; \Sigma;\Gamma^L;\Gamma;\Delta  \vdash e_1 :
  \Pi x{:}\tau_1.\tau_2 
  \\ \ee_2:: u;\Theta; \Sigma;\Gamma^L;\Gamma;\Delta  \vdash e_2: \tau_1}{
  u;\Theta; \Sigma;\Gamma^L;\Gamma;\Delta  \vdash e_1\;e_2 : \tau_2[e_2/x]}
\)
\begin{tabbing}
By assumptions 
\\
\quad \=
$\theta\in\rd{\Theta}{}$, 
$\gamma^L\in\interp{\Gamma^L}$,
\\\>$\gamma_u= U/u$, $U'\geq U$, 
$\trace\vDash\Delta\gamma\gamma_u\gamma^L$
\\\>and  $(n; \gamma)\in \rgi{\Gamma\gamma_u\gamma^L}{\theta;\trace;U'}{\Phi}$,
\`(1)
\\ By I.H. on $\ee_2$
\\ \>
$(n,
e_2\gamma)\in\rei{\tau_1\gamma\gamma_u\gamma^L}{\theta;\trace;U'}{\Phi}$
\`(2)
\\ By I.H. on $\ee_1$
\\ \>
$(n, e_1\gamma)\in\rei{(\Pi x{:}\tau_1.\tau_2)\gamma\gamma_u\gamma^L}
{\theta;\trace;U'}{\Phi}$
\`(3)
\\Assume $(e_1\;e_2)\gamma\rightarrow^{m}_\beta \nf\nrightarrow$
\\By (3),
\\\>
  $(e_1\;e_2)\gamma \rightarrow^*_\beta  \nf_1 (e_2\gamma)$,
\\\> and  $(n-m, \nf_1) \in\rvi{(\Pi
   x{:}\tau_1.\tau_2)\gamma\gamma_u\gamma^L}{\theta;\trace;U'}{\Phi}$ 
\`(4)
\\ We consider two cases:
\\ subcase 1:  $\nf_1 = \lambda x.e'_1$
\\ By (4) 
\\\> 
 $(n{-}m{-}1, e'_1[e_2\gamma/x]) \in
 \rei{\tau_2\gamma\gamma_u\gamma^L[e_2\gamma/x]}
 {\theta;\trace;U'}{\Phi}$ 
\`(5)
\\By (4) and (5)
\\\>
$(n, (e_1e_2)\gamma)\in
\rei{(\tau_2[e_2/x])\gamma\gamma_u\gamma^L}{\theta;\trace;U'}{\Phi}$
\`(6)
\\ subcase 2: $\nf_1\neq\lambda x.e'_1$
\\ By Lemma~\ref{lem:stuck-tm-in-inv}
\\\> $(n-m, \nf_1(e_2\gamma))\in\rvi{\tau_2\gamma\gamma_u\gamma^L[e_2\gamma/x]}
 {\theta;\trace;U'}{\Phi}$
\`(8)
\\By (8)
\\\>
$(n, (e_1e_2)\gamma)\in
\rei{(\tau_2[e_2/x])\gamma\gamma_u\gamma^L}{\theta;\trace;U'}{\Phi}$
\`(9)
\end{tabbing}
\end{description}

% \noindent Proof of 2.(c)
% \begin{description}
% \item[case:] \rulename{Fix}

% \end{description}

\noindent Proof of 2.(b)
\begin{description}
\item[case:] \rulename{SeqC}

\(\mprset{flushleft}
\inferrule*{
 \ee_1::   u_0, u_1, i; 
 \Theta; \Sigma;\Gamma^L;  u_3,\Gamma;\Delta , u_0\leq u_1 \vdash \varphi_0 \jonempty
 \\\\ 
\ee_2::u_1, u_2, i; \Theta; \Sigma;\Gamma^L, u_0:
   \bt,  u_3;\Gamma;\Delta,
    u_1 \leq u_2, \varphi_0 
\\\\~~\vdash c_1 : x{:}\tau.\varphi_1
    \\\\
\ee_3:: u_2, u_3, i; \Theta; \Sigma;\Gamma^L,
   u_0, u_1;\Gamma,x : \tau;\Delta,u_2 \leq u_3,\varphi_0, \varphi_1 
   \\\\~~\vdash
    c_2: y{:}\tau'.\varphi_2
\\\\
\ee_4::  \Theta; \Sigma;\Gamma^L,u_1, u_2,  u_0, u_3, i;\Gamma,
   x{:}\tau,y:\tau';\Delta  
\\\\~~\vdash 
  (\varphi_0 \conj \varphi_1 \conj
  \varphi_2) \Rightarrow \varphi \jtrue 
%\\\\ \Theta;\Sigma;\Gamma^L;u_0, u_3, i,\Gamma\vdash  \tau'\jok 
\\  \Theta;\Sigma;\Gamma^L, u_0, u_3, i;\Gamma,
y:\tau' \vdash  \varphi \jok
\\\fv(\cletc(c_1, x.c_2))\subseteq\dom(\Gamma)
}{
  u_0, u_3, i; \Theta; \Sigma;\Gamma^L; \Gamma;\Delta  \vdash
   \cletc(c_1, x.c_2): y{:}\tau'.\varphi
}
\)

\begin{tabbing}
By assumption 
\\
\quad \=
Pick time points $u$, $u_B$, $u_E$ and thread id $\iota$,
 s.t. $u\leq u_B\leq u_E$,  
\\\>let $\gamma_1 = [u_B, u_E,\iota/u_0,u_3,i]$
\\\>Pick any trace $\trace$, such that $\trace\vDash\Delta\gamma^L\gamma\gamma_1$
\\\> $\theta\in\rd{\Theta}{}$, 
$\gamma^L\in\interp{\Gamma^L}$,
\\\>
   $(n; \gamma)\in \rgi{\Gamma\gamma_u\gamma^L}{\theta;\trace;U'}{\Phi}$,
\`(1)
\\\>the length of the trace from time $u_B$ to the end of
$\trace$ is $j_b$
\\\>the length of the trace from time $u_E$ to the end of $\trace$ is $j_e$
\\\> and  $n \geq j_b \geq j_e$ 
\`(2)
\\\>the configuration at time $u_B$ is
\\\> $\steps{u_B}\sigma_{b}\rt \cdots, \langle\iota; y.c::K;
 \clete(e_1, x.c_2)\gamma\rangle\cdots$
\\\>the configuration at time $u_E$ is
\\\>$\steps{u_E}\sigma_{e}\rt \cdots, \langle\iota; K;c[e/y]\rangle\cdots$ 
\\\>and between $u_B$ and $u_E$ (inclusive), the stack of 
  thread $\iota$ 
\\\>always contains prefix $y.c::K$
\`(3)
\\By the operational semantics
\\\>exists $u_{m1}$, $u_{m2}$,  s.t. $u_B{\leq}u_{m1}
{\leq}u_{m2} {\leq}  u_e$ 
\\\> the configuration at time $u_{m1}$ is
\\\>$\steps{u_{m1}}\sigma_{m1}\rt \cdots, \langle\iota; x.c_2\gamma::y.c::K;
c_1\gamma\rangle\cdots$,
\\\>the configuration at time $u_{m2}$ is
\\\> $\steps{u_{m2}}\sigma_{m2}\rt \cdots, \langle\iota; y.c::K;
c_2\gamma[e_0/x]\rangle\cdots$,
\`(4)
\\By (4)
\\\> between time $u_B$ and $u_{m1}$, thread $\iota$ is silent
\`(5)
\\By (1),
\\\>$\trace\vDash (\Delta\gamma^L\gamma\gamma_1, (u_0\leq u_1)\gamma_1[u_{m1}/u_1])$
\`(6)
\\By (1) 
\\\> $(j_{m1}, \gamma) \in
\rgi{\Gamma\gamma^L [u_E/u_3][u_B, u_{m1},\iota/u_0,u_1,i]}{\theta;\trace;u}{\Phi}$,
\`(7)
\\By I.H. on $\ee_1$ and (5), (6) and (7)
\\\> $\trace\vDash\varphi_0 \gamma\gamma^L\gamma_1[u_{m1}/u_1]$
\`(8)
\\Let $\gamma_2=\gamma\gamma^L\gamma_1[u_{m1}/u_1]$
\\By (1) and Lemma~\ref{lem:ctx-closed-under-delay} and $u\leq u_{m1}$
\\\>$(j_{m1}; \gamma)\in \rgi{\Gamma[u_B, u_{m1},u_{m2}, u_E, \iota/u_0,u_1,u_2,u_3,i]}
{\theta;\trace;u_{m1}}{\Phi}$
\`(9)
\\Let $\gamma_3 = [u_{m1},u_{m2},\iota/u_b,u_e,j]$, 
\\By I.H. on $\ee_2$  and (6), (8), (9)
\\\> $(n, c_1\gamma)$ $\in\rci{(x{:}\tau.\varphi_1)\gamma_2\gamma_3}
  {\theta;\trace;u_{m1};u_{m2};\iota}{\Phi}$
\`(10)
\\By (10), 
\\ let $j_{m2}$ be the length of the trace from time $u_{m_2}$ to the end of
$\trace$
\\\>
$(j_{m2}, e_0)\in\rei{\tau\gamma_2\gamma_3}
 {\theta;\trace; u_{m2}}{\Phi}$ and
\\\>
 $\trace\vDash  \varphi_1\gamma_2\gamma_3[e_0/x]$ 
\`(11)
\\By Lemma~\ref{lem:ctx-smaller} and $j_{m2} < n$
\\let $\gamma_4 = \gamma\gamma^L[u_B, u_E,\iota/u_0,u_3,i][u_{m1},u_{m2}/u_1,u_2][e_0/x]$, 
\\\> $(j_{m2},gamma[e_0/x])$
\\\>$\in\rgi{(\Gamma, x{:}\tau)\gamma_4)[u_{m2},u_E,\iota/u_2,u_3,i]}
 {\theta;\trace;u_{m2}}{\Phi}$
 \`(12)
\\By I.H. on $\ee_3$, (11), (12)
\\\>
 $(j_{m2}, c_2\gamma[e_0/x])
 \in\rci{(y{:}\tau'.\varphi_2)\gamma_4}
 {\theta;\trace; u_{m2},u_E, \iota}{\Phi}$  
\`(13)
\\By (14)
\\\>
 $(j_e, e)\in\rei{\tau'\gamma_4}{\theta;\trace;u_E}{\Phi}$ and
\\\> $\trace\vDash \varphi_2\gamma_4[e/y]$ 
\`(14)
\\By I.H. on $\ee_4$ 
\\\>$\trace\vDash (\varphi_0\conj\varphi_1\gamma_e\varphi_2[e/y])\gamma_4\Rightarrow
\varphi\gamma_4[e/y]$ \`(15)
\\\>$\trace\vDash\varphi\gamma_4[e/y]$
\`(16)
\\By (14) (15)
\\\>
 $(n, \clete(c_1, x.c_2)\gamma)\in
 \rci{(y:\tau'.\varphi)\gamma^L\gamma\gamma_1}{\theta;\trace; u_B, u_E,\iota}{\Phi}$
\end{tabbing}

%\item[case:] \rulename{SeqCComp}
\end{description}

\noindent Proof of 2.(f)
\begin{description}
\item[case:] \rulename{Honest}

\(
\inferrule*{
 \ee_1::  u_1, u_2, i; \Theta;\Sigma;\Gamma^L;\cdot; \Delta \vdash c :  \varphi
\\  \ee_2:: \Theta;\Sigma;\Gamma^L; \cdot; \Delta \vdash \pred{start}(I, c, u)
\jtrue 
\\ \Theta;\Sigma\vdash\Gamma^L, \Gamma \jok
 }{
 \Theta;\Sigma;\Gamma^L; \Gamma; \Delta  \vdash \forall u'{:}\bt.(u'{>}u)  \imp
 \varphi[u,u',I/u_1, u_2,i] \jtrue}
\)
\begin{tabbing}
By assumptions 
\\
\quad \=
$\theta\in\rd{\Theta}{}$, 
$\gamma^L\in\interp{\Gamma^L}$,
\\\>
$\trace\vDash\Delta\gamma^L$
\`(1)
\\To show  $\trace\vDash_\theta  (\forall u'.(u'>u)  \imp
\varphi[u,u',I/u_1, u_2,i])\gamma\gamma^L$
\\ By I.H. on $\ee_2$
\\\>
 $\trace\vDash_\theta \pred{start}(I, c, u)\gamma^L$
 \`(2)
\\By (2)
\\\>
 at time $u\gamma^L$, thread $I\gamma^L$ starts to evaluate $c$ on an empty
 stack, 
\`(3)
\\Given any time $U'>u\gamma^L$, and $k$ such that 
the length of $\trace$
\\ after $u\gamma^L$ is no less than $k$
\\ By I.H. on $\ee_1$
\\\>
$(k, c)\in\rci{\varphi\gamma[ u\gamma, U', I\gamma/u_1,u_2,i]}
{\theta;\trace; u\gamma, U', I\gamma}{\Phi}$
\`(4)
\\ because $c$ starts from an empty stack, 
\\  $c$ couldn't have returned at time $U'$, 
\`(5)
\\By (4) (5) and (1) and the definition of $\mathcal{RC}$, 
\\\>  $\trace\vDash_\theta\varphi\gamma^L[u\gamma, U', I\gamma/u_1, u_2,
i]$
 \`(7)
\end{tabbing}

\item[case:] \rulename{$\forall$I}

\(
\inferrule*{
\ee'::\Theta; \Sigma;\Gamma^L, x:\tau;\Gamma;\Delta  \vdash \varphi \jtrue
}{\Theta; \Sigma;\Gamma^L;\Gamma;\Delta  \vdash \forall
  x{:}\tau. \varphi \jtrue}
\)
\begin{tabbing}
By assumptions 
\\
\quad \=
$\theta\in\rd{\Theta}{}$, 
$\gamma^L\in\interp{\Gamma^L}$,
$\trace\vDash\Delta\gamma\gamma^L$
\\\>and  $(n; \gamma)\in \rgi{\Gamma\gamma^L}{\theta;\trace;u}{\Phi}$,
\`(1)
\\ Given any  $e$ such that 
$ e\in\interp{\tau}$ 
\\\> $ \gamma^L[e/x]\in\interp{\Gamma^L, x:\tau}$
\`(2)
\\By I.H. on $\ee'$
\\\> $\trace\vDash\varphi \gamma^L[e/x]\gamma$
\`(3)
\\By definitions
\\\>$\trace\vDash (\forall x{:}\tau.\varphi) \gamma^L\gamma$
\end{tabbing}
\end{description}
\end{proof}

\clearpage
\section{Proof Sketch of State Integrity for Memoir}
\label{app:memoir}

We prove the correctness of a TPM based state continuity mechanism
that closely follows Memoir~\cite{memoir}.

\newcommand{\strng}[1]{\mathtt{``#1"}}

\begin{figure*}
\(
\begin{array}{@{}l}
\linestart
\runmodule=\newl
~~\stlet{snapshot}\newl
~~\lambda (state, summary, \skey).\newl
~~~~~~\cbnd{enc\_state}{\cact(\action{encrypt}(\skey,service\_state))}\newl
~~~~~~\cbnd{auth}{\cact(\action{mac}~(\skey, (enc\_state, freshness\_tag))}\newl
~~~~~~\cret(enc\_state, freshness\_tag, auth)\newl
~~\newl
~~
~~\stlet{check\_snapshot} \newl
~~\lambda ((enc\_state, freshness\_tag, auth), request, history, \skey).\newl
~~~~~~{\cact(\action{verify\_mac}~(\skey, (enc\_state, freshness\_tag), auth)};\newl
~~~~~~\cbnd{freshness\_tag'}{\cact(\action{hash}~(freshness\_tag||request))}\newl
~~~~~~\cif(freshness\_tag = history \lor freshness\_tag' = history,\cact(\action{dec}~(\skey,enc\_state)), \cact(abort()))\newl
~~\newl
~~\stlet{initialize} \newl
~~\lambda (service, \Nloc). \newl
~~~~~~\cact(\action{extend\_pcr} (\llpcr, code\_hash(service))) ;\newl
~~~~~~\cact(\action{verify\_pcr} (\llpcr, hash\_chain (-1, code\_hash (runmodule), code\_hash(service)))) ;\newl
~~~~~~\cbnd{\skey}{\cact(\action{gen\_symkey}())} \newl
~~~~~~\stlet{history\_summary}0 \newl
~~~~~~\cact(\action{setNVRAMlocPerms} (\Nloc, \llpcr)) ;\newl
~~~~~~\cact(\action{NVRAMwrite} (\Nloc, (history\_summary, \skey)) ;\newl
~~~~~~\cact((\action{extend\_pcr}(\llpcr, 0))  ;\newl
~~~~~~\cbnd{service\_state}{(\mi{service}~\mi{ExtendPCR}~\mi{ResetPCR}~\cdots)~\mathtt{INIT}} \newl
~~~~~~\cact(\cact(\action{service\_init}(skey, service, service\_state, \Nloc))) ;\newl
~~~~~~\cbnd{snap}{snapshot(service\_state, history\_summary, \skey)} \newl
~~~~~~\cret ((), snap)\newl
~~\newl
~~\stlet{execute} \newl
~~\lambda (service, \Nloc, snap, req). \newl
~~~~~~\cact(\action{extend\_pcr} (\llpcr, code\_hash(service))) ;\newl
~~~~~~\cbnd{(\skey, history\_summary)}{\cact(\action{NVRAMread}~\Nloc)} \newl
~~~~~~\cbnd{service\_state}{check\_snapshot(snap, request, history\_summary, \skey)}\newl
~~~~~~\cbnd{new\_summary}{\cact(\action{hash}~(history\_summary || req))} \newl
~~~~~~\cact(\action{NVRAMwrite} (\Nloc, (new\_summary, \skey)) ;\newl
~~~~~~\cact(\action{extend\_pcr}(\llpcr, 0))  ;\newl
~~~~~~\cact(\action{service\_try}(skey, service, service\_state, \Nloc)) ;\newl
~~~~~~\cbnd{(new\_state, resp)}{(\mi{service}~\mi{ExtendPCR}~\mi{ResetPCR}~\cdots)~(\mathtt{EXEC}(service\_state, req))} \newl
~~~~~~\cbnd{snap}{snapshot(service\_state, history\_summary, \skey)} \newl
~~~~~~\cact(\cact(\action{service\_invoke}(skey, service, service\_state, new\_state, \Nloc))) ;\newl
~~~~~~\cret (resp, snap)\newl
~~\newl
~~\lambda (service, \Nloc, call).\newl
~~~~~~~~(resp, snap) \leftarrow (\mathtt{case}~call~\mathtt{of}\newl
~~~~~~~~~~~~~~~\mathtt{INIT} \Rightarrow initialize (service, \Nloc)\newl
~~~~~~~~~~~~~|~\mathtt{EXEC}(snap, req) \Rightarrow execute (service, \Nloc, snap, req))\newl
~~~~~~~~\cact(\action{send}(response, snap));\newl
~~~~~~~~\cact(\action{ll\_exit}())
\end{array}
\)
\caption{$\mathtt{runmodule}$: A model of Memoir's state isolation mechanism}
\label{fig:memoir:runmodule}
\end{figure*}

\subsection*{Terms, Actions and Predicates}
We first describe here the terms, actions and predicates that model the TPM functionality, cryptography
and communication.

\paragraph{TPM functionality.} The TPM is modeled by the following actions.
The actions $\action{reset\_pcr}(p)$ and $\action{extend\_pcr}(p, h)$,
respectively resets the state of the PCR $p$ to some default value and extends
the value of $p$ with the value $h$.  The action $\action{verify\_pcr}(p, h)$
checks is the state of PCR $p$ is $h$, otherwise aborts.  The action
$\action{setNVRAMlocPerms}(\Nloc, p)$ ties the permissions for NVRAM location
$\Nloc$ to the current contents of the PCR $p$.  The actions
$\action{NVRAMwrite}(\Nloc, m)$ and $\action{NVRAMread}(\Nloc)$ respectively
write the message $m$ and read from the NVRAM.  The action
$\action{ll\_enter}(e)$ starts a new late launch session with computation $e$
called on some arguments. A late launch session is modeled by a new thread that
runs $e$ with no other thread running in parallel. The action
$\action{ll\_exit}()$ exits from a late launch session.

\paragraph{Cryptography.}

Symmetric encryption is modeled by the actions $\action{encrypt}(k, m)$ and
$\action{decrypt}(k, c)$. Message authentication codes are modeled by
$\action{mac}(k, m)$ and $\action{verify\_mac}(k, m, m')$.  Hash functions are
modeled by the action $\action{hash}(m)$. A message $m$ encrypted by a key $k$
is denoted by the term $ENC_k(m)$. Similarly, a MAC of a message $m$ with key
$k$ is denotedy by $\mathit{MAC}_k(m)$. A hash is represented by the term $hash(m)$. The
special term $code\_hash(c)$ refers to the textual reification of the
computation $c$. The term $hash\_chain(m_1, m_2, \cdots, m_k)$ is syntactic
sugar for the iterated hash $hash(hash(hash(m_1) || m_2 \cdots || m_k)\cdots)$.
Here, the term $m_1 || m_2$ represents the concatenation of messages.

\paragraph{Communication.}
Communication is modeled by the $\action{send}(m)$
$\action{receive}()$ action. By default, messages are not authenticated, so we drop
the send and receive respectively do not have a recipient and sender argument.

\paragraph{Flags.}
To state the overall state continuity property, we require three flags
($\action{service\_init}$, $\action{service\_try}$ and
$\action{service\_invoke}$) which simply record the value of variables at a
particular point.

Figure~\ref{fig:memoir:runmodule} contains our model for the Memoir system. The
suspended computation $\runmodule$ is expected to run in a late launch session
that models both the initialization and execution phase of Memoir. Lines 14-26
model the initialization phase and lines 28-40 model the execution phase. We
only describe the initialization phase here and the execution phase proceeds
similarly.  During $initialize$ the code for $service$ is hashed into PCR 17.
Subsequently, it is checked whether PCR 17 contains a hash chain starting with
-1 and followed by a hash of the textual reification of $runmodule$. This
ensures that a late launch session with $runmodule$ was initiated.
A symmetric key is then generated that acts as the encryption and
MAC key for subsequent sessions of Memoir. Then, the permissions on $\Nloc$,
the NVRAM location allocated for the session is tied to the current value of
PCR17.  An initial history summary and the symmetric key are then written to
the NVRAM location, and then the value of PCR 17 is extended with a dummy value
so that $\Nloc$ cannot be read unless a new $runmodule$ session is started.
The service is then initiated to generate a state of the service that is then
encrypted and MACed along with the history summary and sent to the adversary
for persistent storage.

\paragraph{Predicates.} Each action has a corresponding action predicate. All
action predicates are listed in Figure~\ref{fig:memoir1:preds}. Every action
predicate has an additional argument that corresponds to the thread that
performed that action. The one exception is the action predicate
$\pred{LLEnter}$, for which the first argument $j$ is the thread corresponding
to the late launch session.

Apart from action predicates, we have predicates which capture state. The
predicate $\valpcr(p, h)@u$ states that at time $u$, the value of the PCR $p$
is the hash $h$. The predicate $\pred{NVPerms}(\Nloc, p, h)@u$ states that the
permissions on the NVRAM location $\Nloc$ are set to the value of the PCR $p$
being the hash $h$. The predicate $\valNV(\Nloc, m)@u$ states that the NVRAM
location $\Nloc$ contains the value $m$ at time $u$.

We have some predicates about the structure of terms.  The predicate
$\pred{hash\_prefix}(h_1, h_2)$ states that the hash chain $h_2$ can be
obtained by extending $h_1$ with additional hashes.

\begin{figure}[h]
$$\small \begin{array}{@{}l|l}
      \textbf{Action} & \textbf{Predicate} \\ \hline
      \action{reset\_pcr}(p)&\pred{ResetPCR}(i, p)\\
      \action{extend\_pcr}(p, h)&\pred{ExtendPCR}(i, p, h)\\
      \action{verify\_pcr}(p, h)&\pred{VerifyPCR}(i, p, h)\\
      \action{setNVRAMlocPerms}(\Nloc, p)&\pred{SetNVPerms}(i, \Nloc, p)\\
      \action{NVRAMwrite}(\Nloc, m)&\pred{NVWrite}(i, \Nloc, m)\\
      \action{NVRAMread}(\Nloc)&\pred{NVRead}(i, \Nloc, m)\\
      \action{ll\_enter}(e)&\pred{LLEnter}(j, e)\\
      \action{ll\_exit}()&\pred{LLExit}(i)\\
      \action{encrypt}(k, m)&\pred{Encrypt}(i, k, m)\\
      \action{decrypt}(k, m)&\pred{Decrypt}(i, k, m)\\
      \action{mac}(k, m)&\pred{MAC}(i, k, m)\\
      \action{verify\_mac}(k, m, m')&\pred{verifyMAC}(i, k, m, m')\\
      \action{hash}(k, m)&\pred{Hash}(i, k, m)\\
      \action{service\_init}(skey, &\pred{service\_init}(i, skey, \\
      ~~~~~~~service, state, Nloc)&   ~~~~~~~~     service, state, Nloc) \\
      \action{service\_try}(skey, &\pred{service\_try}(i, skey, \\
      ~~~~~~~service, state, Nloc)&   ~~~~~~~~     service, state, Nloc) \\
      \action{service\_invoke}(skey, &\pred{service\_invoke}(i, skey, \\
      ~~~~~~~service, state, state', Nloc)&   ~~~~~~~~     service, state, state', Nloc) \\
      %\action{service\_try}(skey, service, state, Nloc)&\pred{service\_try}(i, skey, service, state, Nloc)\\
      %\action{service\_invoke}(skey, service, state, state', Nloc)&\pred{service\_invoke}(i, skey, service, state, state', Nloc)\\
    \end{array}
$$
\caption{Action Predicates}
\label{fig:memoir1:preds}
\end{figure}

\subsection*{Abbreviations and Definitions}
Figure~\ref{fig:memoir1:abbs} summarizes  the abbreviations we use.
\begin{figure}[h]
\framebox{Abbreviations}
\[\small \begin{array}{@{}ccc}
(\varphi \conj \psi) \at u & = & (\varphi \at u) \conj (\psi \at u) \\
(\varphi \disj \psi) \at u & = & (\varphi \at u) \disj (\psi \at u) \\
(\varphi \imp \psi) \at u & = & (\varphi \at u) \imp (\psi \at u) \\
(\neg \varphi) \at u & = & \neg (\varphi \at u) \\
\top \at u & = & \top \\
\bot \at u & = & \bot \\
(\forall x. \varphi) \at u & = & \forall x. ~(\varphi \at u) \\
(\exists x. \varphi) \at u & = & \exists x. ~(\varphi \at u) \\
(\varphi \at u') \at u & = & \varphi \at u' \\
\\
\varphi \jon (u_1, u_2) & = & \forall u. ~(u_1 < u < u_2) \imp (\varphi
\at u) \\
\varphi \jon (u_1, u_2] & = & \forall u. ~(u_1 < u \leq u_2) \imp (\varphi
\at u) \\
\varphi \jon [u_1, u_2) & = & \forall u. ~(u_1 \leq u < u_2) \imp (\varphi
\at u) \\
\varphi \jon [u_1, u_2] & = & \forall u. ~(u_1 \leq u \leq u_2) \imp (\varphi
\at u) \\
\end{array}\]
\caption{Abbreviations}
\label{fig:memoir1:abbs}
\end{figure}

\begin{figure}[h]

\framebox{Definitions}
\[\small \begin{array}{@{}l}

\pred{LL}(u_1, u_2, e, j) =  \pred{LLenter}(e, j)@u_1 \land \neg \pred{LLexit}(j)\jon[u_1, u_2)\\
~~~~~~~~~~~~~~~~~~~~~~~~~~~~~~~~~~~~~~~~~~~~~~~~~~~~~~~~~~~~~ \land \pred{LLexit}(j)@u_2\\
\pred{InLLSess}(u, e, j) = \exists u_1. (u_1 \leq u) \land \pred{LLenter}(e, j)@u_1\\
~~~~~~~~~~~~~~~~~~~~~~~~~~~~~~~~~~~~~~~~~~~~~~~~~~~~~~~~~~~~~ \land \neg \pred{LLexit}(j)\jon[u_1, u)\\
\pred{InSomeLLSess}(u, e) = \exists j. \pred{InLLSess}(u, e, j)\\
\pred{LLThread}(j, e) = \exists u. \pred{LLenter}(e, j)@u\\
\pred{PCRPrefix}(p, \shash) = \exists h.~\valpcr(\llpcr, h) \land \hprefix(h, \shash) \\
\pred{ExitsPCRProtected}(i, u, \shash) =  \pred{LLexit}(i)@u \imp \\
~~~~~~~~~~\neg\pred{PCRPrefix}(\llpcr, \shash)@u\\
\pred{LLChain}(h, e) = \pred{hash\_prefix}(\hchain(-1, \chash(e)), h)\\\
\end{array}\]

\framebox{Axioms}
\[\small \begin{array}{@{}ll}
%\mbox{(LLChain)} & \pred{LLChain}(\hchain(-1, \chash(e), ...), e)\\
\mbox{(LLExit)}&
\forall \shash, u_2, e\\
&\pred{LLChain}(\shash, e) \imp\\
&~~~~~ \valpcr(\llpcr, \shash)@u_2\\
&~~\land \neg \pred{InSomeLLSess}(u_2, e)\imp\\
&~~~~~~\exists j, u_3.\\
&~~~~~~~~~~~~~\pred{LLThread}(j, e)\\
&~~~~~~~~~~~~~\land \pred{LLexit}(j)@u_3\\
&~~~~~~~~~~~~~\land \valpcr(\llpcr, h)@u_3\\
&~~~~~~~~~~~~~\land \hprefix(h, \shash)\\
&~~~~~~~~~~~~~\land \forall u \in (u_1, u_3). \\
&~~~~~~~~~~~~~~~~~~~~\valpcr(\llpcr, \shash)@u \imp \pred{InSomeLLSess}(u, e)\\

\mbox{(PCRInit)}&
\valpcr(p, 0)@-\infty\\

\mbox{(LLHonest)}& \pred{LLEnter}(i, e)@u \imp \exists e'.~\pred{start}(-\infty, e~e', i)\\

&\mbox{Th next two axiom schemas holds for any action $a(i, t)$}\\
\mbox{(LLAct1)} &a(i, t)@u \land \pred{InSomeLLSess}(u, e) \imp \pred{InLLSess}(u, e, i)\\
\mbox{(LLAct2)} &a(i, t)@u \land \pred{LLThread}(i, e) \imp \pred{InLLSess}(u, e, i)\\

\end{array}\]
\caption{Definitions and Model-specific axioms about late launch}
\label{fig:memoir1:axioms}
\end{figure}

\subsection{Proof Overview}

The proof proceeds in four stages. Each step employs the rely-guarantee
technique in the style of~\cite{garg10:ls2} to prove a particular invariant
about executions of the system.  At a high level, the four stages of the proof
are as follows:

\begin{enumerate}

\item\textbf{PCR Protection:} We show that the value of $\llpcr$ contains a
  certain measurement $h$ only during late launch sessions running
  a session of Memoir.

\item\textbf{NVRAM Protection:} We show that after the permissions on a
  location in the NVRAM has been set to $h$, then the permissions on that
  location are never changed.

\item\textbf{Key Secrecy:} We show that if the key corresponding to the service
  is available to a thread, then it must have either generated it or read it from
  the NVRAM.

\item\textbf{History Summary-State Correspondence:} We show that if on any two
  executions of the Memoir, if the history summaries are equal then the states
  must also be equal.

\end{enumerate}

Finally, from these, we prove the overall state continuity property for Memoir.

Next, we sketch the proofs of each of the above stage.
The proofs require axioms about the above predicates, which we state
along with the stage the axioms are first required.

\subsubsection{PCR Protection.}

%\subsection*{Axioms}

In Figure \ref{fig:memoir1:axioms}, we list the definitions and model specific
axioms we need. The predicate $\pred{LL}(u_1, u_2, e, j)$ states that thread
$j$ runs a late launch session for $e$ between $u_1$ and $u_2$. The
predicate $\pred{InLLSess}(u, e, j)$ states at time $u$, thread $j$ runs
a late launch session for $e$. The predicate $\pred{InSomeLLSess}(u, e)$ states
that at time $u$, some thread is running a late launch session for $e$.
$\pred{LLThread}(j, e)$ states that $j$ is a thread that runs
a late launch session for $e$. $\pred{PCRPrefix}(p, \shash)$ states that
the value contained in $p$ is a hash prefix of $\shash$. $\pred{ExitsPCRProtected}(i, u, \shash)$
states that whenever a late launch thread exists, the state of PCR 17 is not
a prefix of $\shash$. $\pred{LLChain}(h,e)$ states that $h$ is a hash chain,
which if contained in PCR 17, is evidence of a late launch session for $e$.

Axiom (LLExit) states that whenever outside a late launch session,
the value of PCR 17 is found to be a late launch chain $\shash$,
we can conclude, that some late launch session exited with the
state of PCR 17 being a prefix of $\shash$. (PCRInit) states
that the value of any PCR begins at 0. (LLEnter) states
that late launch threads for a computation $e$ exclusively run $e$
with some arguments $e'$. (LLAct1) and (LLAct2) are axiom schemas
that essentially state that no other threads are active during
late launch sessions.

Consider an arbitrary service $s$.
Let $\shash = \hchain(-1, \\\chash(runmodule), \chash(s))$.
We show that if the value of $\llpcr$ at time $u$ is $\shash$, then it must be the case
that we are in a late launch session at time $u$. Formally, we show that,

\begin{equation}
 \forall u. \valpcr(\llpcr, \shash)@u \imp \pred{InSomeLLSess}(u, runmodule)
\label{eq:memoir:pcrinv}
\end{equation}

%To prove \ref{eq:memoir:pcrinv}, we use rely-guarantee reasoning .
To prove an invariant $\forall u\mbox{$>$}u_i.~\varphi(u)$, using rely guarantee
reasoning, it is sufficient to show for a choice of $\psi(i, u)$ and $\iota(i)$
that

\begin{enumerate}
\item[(1)]
$\varphi(u_i)$
\item[(2)]
$\begin{array}{@{}l}
\forall i,u.~(\iota(i) \conj \forall u' < u. ~\varphi(u')) \imp \psi(u,i)
\end{array}$
\item[(3)]
$\begin{array}{@{}l}
(\varphi(u_1) \conj \neg \varphi(u_2) \conj (u_1 < u_2)) \imp \\
~~~~\exists i, u_3. ~(u_1 < u_3 \leq u_2) \conj \iota(i) \conj
\neg \psi(u_3,i) \conj \\
~~~~~~~~~~~~~~~~~\forall u_4 \in (u_1, u_3).~\varphi(u_4)
\end{array}$
\end{enumerate}

We choose $\varphi$,$\psi$ and $\iota$ as below:
\[\small \begin{array}{@{}l}
\varphi(u) = \valpcr(\llpcr, \shash)@u \imp \pred{InSomeLLSess}(u, runmodule)\\
\psi(i, u) =  \pred{ExitsPCRProtected}(i, u, \shash)\\
\iota(i) = \pred{LLThread}(i, runmodule)\\
\end{array}
\]

Condition (1) follows (PCRInit) and $\neg \hprefix(0, \shash)$. Condition (3) follows
directly from axiom (LLExit). To prove conditions (2) above, expanding out the definitions
of $\varphi$, $\iota$ and $\psi$ above, we need to show that

\begin{equation}
\begin{array}{@{}l}
\forall i,u.~(\pred{LLThread}(i, runmodule) \\
~~~~~~~~~~~~\conj \forall u' < u. ~(\valpcr(\llpcr, \shash)@u\\
~~~~~~~~~~~~~~~~~~~~~~~~~~~~~\imp \pred{InSomeLLSess}(u, runmodule)(u')) \\
~~~~~~~~~~~~~~~~~~\imp \pred{ExitsPCRProtected}(u,i)\\
\end{array}
\end{equation}

This can be rewritten as
 \begin{equation}
\begin{array}{@{}l}
\forall i.~(\pred{LLThread}(i, runmodule) \\
~~~~~~~~~~~~\conj\forall u.~(\forall u' < u. ~(\valpcr(\llpcr, \shash)@u'\\
~~~~~~~~~~~~~~~~~~~~~~~~~~~~~\imp \pred{InSomeLLSess}(u', runmodule)) \\
~~~~~~~~~~~~~~~~~~\imp \pred{ExitsPCRProtected}(u,i))\\
\end{array}
\label{eq:memoir:pcrinv:rg3}
\end{equation}

Choose an arbitrary thread $i$ such that $\pred{LLThread}(i, runmodule)$.
Therefore, we have by (LLHonest) that for some $e'$, \\$\pred{start}(-\infty,
runmodule~e', i)$. To use rule \rulename{Honest} to show
(\ref{eq:memoir:pcrinv:rg3}), we need to show that $runmodule$
satisfies the following invariant.

\begin{equation}
\begin{array}{@{}l}
\vdash runmodule : \\
%~~~~~~~~\Pi (s:\texp, l:\tptr, snap:\tmsg).~\tcmp(u_b, u, i.~\\
~~~~~~~~~~~~~(\forall u_b < u' < u_e.(\valpcr(\llpcr, \shash)@u'\\
~~~~~~~~~~~~~~~~~~~~~~~~~~~~~\imp \pred{InSomeLLSess}(u', runmodule)) \\
~~~~~~~~~~~~~~~~~~\imp \pred{ExitsPCRProtected}(u,i))\\
\end{array}
\label{eq:memoir:pcrinv:rmtype}
\end{equation}

The key step in typing $runmodule$ is to type the execution of $s$ supplied by
the adversary using the \rulename{Confine} rule. Essentially, we need to show
that the service cannot exit with the $\llpcr$ containing a prefix of
$\shash$. The service is confined to the actions provided by the TPM and we can
show that each of them has the following computational type $\tcmp(u_b, u_e, i. (x.\varphi_c, \varphi_c))$, where $\varphi_c$ is:

\begin{equation}
\begin{array}{@{}l}
\varphi_c = \neg\pred{PCRPrefix}(\llpcr, \shash)@u_b \imp \\
~~~~~~~~~~~~~~~~ \forall u \in [u_b, u_e].~(\pred{InLLSess}(u, runmodule, i)  \\
~~~~~~~~~~~~~~~~~~~~~~~~~~~~~~\imp \neg\pred{PCRPrefix}(\llpcr,\shash)@u)
 \end{array}
\end{equation}

Therefore, we can give $s$ the same type. We have now shown that by the end of
service, the late launch session has either terminated or the value of $\llpcr$
is not a prefix of $\shash$.
Using (LLAct2), we can now show
(\ref{eq:memoir:pcrinv:rmtype}).

\subsubsection{NVRAM Protection.}

%\subsection*{Axioms}

\begin{figure}

%\framebox{Definitions}
%\[\small \begin{array}{@{}l}
%\end{array}\]

\framebox{Axioms}
\[\small \begin{array}{@{}ll}

(SetPerms)&
\pred{SetNVPerms}(i, Nloc, p)@u \land \valpcr(p, h)@u\\
&~~~~~~~~~~~\imp \pred{NVPerms}(Nloc, p, h)@u\\
(GetPerms)&
(\pred{SetNVPerms}(i, Nloc, p')@u \lor\\
&~\pred{NVRead}(i, Nloc, p')@u \lor \\
&~\pred{NVWrite}(i, Nloc, p')@u) \\
&~~~~~~~~~~~~~~~ \land \pred{NVPerms}(Nloc, p, h)@u \imp \valpcr(p, h)@u\\
(NVPerms)&
\pred{NVPerms}(Nloc, p, h)@u_1 \land \neg \pred{NVPerms}(Nloc, p, h)@u_2\\
&~~~~\land (u_1 < u_2) \imp\\
&~~~~~~\exists u_3,j.p',h'.~(u_1< u_3 \leq u_2) \land \valpcr(p',h')@u_3\\
&~~~~~~~~~~~~~~~~~~~~~~~~~~~~~~~~~~~\land\pred{SetNVPerms}(j, Nloc, p') @u_3 \\
&~~~~~~~~~~~~~~~~~~~~~~~~~~~~~~~~~~~\land (p\not=p' \lor h\not=h')\\
&~~~~~~~~~\land \forall u_4 \in (u_1, u_3).~ \pred{NVPerms}(Nloc, p, h)@u_4

\end{array}\]
\caption{Model-specific axioms about NVRAM}
\label{fig:memoir2:axioms}
\end{figure}

Figure~\ref{fig:memoir2:axioms} contains axioms governing the behavior of NVRAM. (SetPerms)
states that on the successful execution of setting permissions on NVRAM at time $u$, the permissions
are correct at $u$. (GetPerms) states that when the permissions on a particular NVRAM location
is tied to the PCR $p$ being $h$, then accessing that NVRAM location implies that
the value of PCR $p$ is $h$. (NVRAMPerms) states that if the permissions on a NVRAM location
changes, then it must have been changed via a $\action{setNVRAMlocPerms}$ action.

We wish to show that the permissions on the NVRAM are always tied to the value of $\llpcr$ being
$\shash$:

\begin{equation}
\begin{array}{@{}l}
(\pred{SetNVPerms}(i, Nloc, \llpcr) \land \valpcr(\llpcr, \shash))@u_i\\
~~~\imp \forall (u>u_i).~\pred{NVPerms}(Nloc, \llpcr, \shash)@u
\end{array}
\label{eq:memoir:nvraminv}
\end{equation}

Assume that for some time point $u_i$.
\begin{equation}
\pred{SetNVPerms}(i, Nloc, \llpcr) \land \valpcr(\llpcr, \shash))@u_i
\label{eq:memoir2:ass1}
\end{equation}

We now need to show that
\[
\forall (u>u_i) \imp \pred{NVPerms}(Nloc, \llpcr, \shash)@u
\]
Again, we prove this invariant by rely guarantee reasoning, where we choose $\varphi$, $\psi$ and $\iota$ to be the following.

\[\small \begin{array}{@{}l}
\varphi(u) = \pred{NVPerms}(\Nloc, \llpcr, \shash)@u\\
\psi(u, i) = (\pred{SetNVPerms}(i, Nloc, p)\\
~~~~~~~~~~~~~~~~~ \imp (p = \llpcr) \land \valpcr(\llpcr, \shash))@u\\
\iota(i) = \pred{LLThread}(i, runmodule)
\end{array}\]

Expanding condition (1), we need to show the following  $$\pred{NVPerms}(\Nloc, \llpcr, \shash)@u_i$$
This holds by Axiom (SetPerms) and \ref{eq:memoir2:ass1}.

Expanding condition (2), choose $i$ such that $\pred{LLThread}(i, \runmodule)$.
We need to show that
 $\forall u>u_i.(\forall u' \in (u_i,u).~\varphi(u')) \imp \psi(i, u)$. To
 use \rulename{Honest}, we need to show that $runmodule$ satisfies the following invariant.

\begin{equation}
\begin{array}{@{}l}
\vdash runmodule : \\
%~~~~~~~~\Pi (s:\texp, l:\tptr, snap:\tmsg).~\tcmp(u_b, u_e, i.~\\
~~~~~~~~~~~~~~~~~~~~~~\forall u\in(u_b, u_e]\forall u'\in [u_i,u).\\
~~~~~~~~~~~~~~~~~~~~~~~~~\pred{NVPerms}(\Nloc, \llpcr, \shash)@u'\imp\\
~~~~~~~~~~~~~~~~~~~~~~~~~\pred{SetNVPerms}(i, Nloc, p)@u\imp\\
~~~~~~~~~~~~~~~~~~~~~~~~~(p = \llpcr) \land \valpcr(\llpcr, \shash)@u\\
\end{array}
\label{eq:memoir2:rmtype}
\end{equation}

Again, the key step in typing $runmodule$ is to type the execution of $s$ supplied by
the adversary using the \rulename{Confine} rule. Essentially, we show that the
service is not allowed to set the permissions of $\Nloc$ at all. Each
action $f$ provided by the TPM interface can be confined to the type $\tcmp(u_b, u_e, i.(x.\varphi_c, \varphi_c))$, where $\varphi_c$ is:

\begin{equation}
\begin{array}{@{}l}
f : \tcmp(u_b,u_e,i.~\neg\pred{PCRPrefix}(\llpcr, \shash)@u_b \imp \\
~~~~~~~~~~~~~~~~~~~ \forall u \in [u_b, u_e].~(\pred{InLLSess}(u, \runmodule, i)  \\
~~~~~~~~~~~~~~~~~~~~~~~~~~~~~~~~~~~\imp\forall p.~ \neg\pred{SetNVRAMPerms}(i,\Nloc, p)@u)
 \end{array}
\end{equation}

Condition (3) follows from
(NVPerms), (GetPerms) and (\ref{eq:memoir:pcrinv}).

In particular, we can show from  \ref{eq:memoir:nvraminv}
and (GetPerms):

\begin{equation}
\begin{array}{@{}l}
(\pred{SetNVPerms}(i, Nloc, \llpcr) \land \valpcr(\llpcr, \shash))@u_i\\
~~~\imp \forall (u>u_i).~ \pred{ReadNV}(I, Nloc)@u\\
~~~~~~~~~~~~~~~~~~~~~~\imp \valpcr(\llpcr, \shash)@u\\

\end{array}
\end{equation}
And by \ref{eq:memoir:pcrinv}

\begin{equation}
\begin{array}{@{}l}
(\pred{SetNVPerms}(i, Nloc, \llpcr) \land \valpcr(\llpcr, \shash))@u_i\\
~~~ \imp \forall (u>u_i) \imp \pred{ReadNV}(I, Nloc)@u\\
~~~~~~~~~~~~~~~~~~~~~~\imp \pred{InSomeLLSess}(u, runmodule)
\end{array}
\end{equation}

Therefore, by (LLAct), we have that

\begin{equation}
\begin{array}{@{}l}
(\pred{SetNVPerms}(i, Nloc, \llpcr) \land \valpcr(\llpcr, \shash))@u_i\\
~~~~\imp \forall I,(u>u_i) \imp \pred{ReadNV}(I, Nloc)@u\\
~~~~~~~~~~~~~~~~~~~~~~\imp \pred{InLLSess}(u, runmodule, I)\\
\end{array}
\label{eq:memoir:nvramread}
\end{equation}

This means that whenever, a thread $i$ reads from the $\Nloc$ at time $u$,
it must be the case that $i$ is in a late launch session running $\runmodule$ at time $u$.

\subsubsection{Key Secrecy.}

Figure~\ref{fig:memoir3:axioms} lists the definitions and axioms pertaining to key secrecy.
The definition $\pred{NVContains}(\Nloc, s)$ states that the NVRAM location $\Nloc$ contains
the secret $s$. $\pred{Private}(s, \Nloc, u)$ states that the secret $s$ hash not been sent
out on the network and the only NVRAM location it has been stored in is $\Nloc$. $\pred{KeepsPrivate}(i, s, \Nloc)$ states that whenever thread $i$ sends a message,
it does not contain the secret $s$. Additionally, it only stores $s$ in $\Nloc$.
$\pred{NewInLL}(s, e)$ states that $s$ was generated in a late launch session of $e$.

The axiom (Shared) states that if a secret is private at time $u_1$ and not private
at $u_2$, then it must be the case, that at some point in the middle some thread
violated $\pred{KeepsPrivate}(i, s, \Nloc)$. (POS) states that if some thread
posses a secret $s$ that is private to $\Nloc$, then it must have been either
generated in that thread or read from $\Nloc$. (PrivateInit) states that
a secret is private as soon as it is generated. (New3) is an axiom about
non-collision of nonce values. (Init) is a logical assumption we make that
states that $\action{service\_init}$ can only be called by honest threads
running $runmodule$.

\begin{figure}

\framebox{Definitions}
\[\small \begin{array}{@{}l}
\pred{NVContains}(Nloc, s) = \exists m. \pred{Contains}(m, s) \land \pred{\valNV}(m, s)\\
\pred{Private}(s, Nloc, u) = \forall u' < u. (Send(i, m)@u \imp \neg \pred{Contains}(m, s) \\
~~~~~~~~~~~~~~~~~~~~~~~~~~~\land \forall Nloc'. (\pred{NVContains}(Nloc, s)@u' \imp (Nloc' = Nloc))\\
\pred{KeepsPrivate}(i, s, Nloc) = Send(i, m) \imp \neg \pred{Contains}(m, s) \\
~~~~~~~~~~~~~~~~~~~~~~~~~~~~~~~~~~~~~~~\land \forall Nloc'. (\pred{WriteNV}(Nloc', m) \land \pred{Contains}(m,s)\\
~~~~~~~~~~~~~~~~~~~~~~~~~~~~~~~~~~~~~~~~~~~~~~~~~~~~~~~~~~~~~~~~ \imp Nloc = Nloc')\\

\pred{NewInLL}(s, e) = \pred{New}(i, s)@u \imp \pred{InLLSess}(u, e, i)\\
\end{array}\]

\framebox{Axioms}
\[\small \begin{array}{@{}ll}

\mbox{(Shared)}&
\pred{LLChain}(h, e) \land\\
&\pred{NewInLL}(s, e) \land\\
&\forall u > u_i. \pred{NVPerms}(Nloc, \llpcr, h) \imp\\
&\forall u_1, u_2 \in (u_i, \infty]\\
&\pred{Private}(s, Nloc, u_1) \land \neg \pred{Private}(s, Nloc, u_2) \imp\\
&~~~~~\exists i,u_3. (u_1 < u_3 <=u_2)\\
&~~~~~~~~~~(\pred{LLThread}(i, e)\\
&~~~~~~~~~~~\neg \pred{KeepsPrivate}(i, s, Nloc)@u_3) \land\\
&~~~~~~~~~~~\forall u \in(u_1, u_3) \pred{Private}(s, K, u))\\

\mbox{(POS)}&
(\pred{Private}(s, Nloc, u) \land \pred{Has}(i, s)@u \imp\\
&~~(\exists u'. (u' < u) \land \pred{New}(i, s) @u') \lor\\
&~~(\exists u'. (u' < u) \land \pred{ReadNV}(i, Nloc, m) @u' \land \pred{Contains}(m, s))\\

\mbox{(PrivateInit)}&
\pred{New}(s)@u \imp \pred{Private}(s, Nloc, u)\\

\mbox{(New3)}&
\pred{New}(i, n)@u \land \pred{New}(i', n)@u' \imp (i = i') \land (u = u')\\

&\text{Assumption about about $\sinit$}\\
\mbox{(Init)}&
\sinit(i, \skey, service, state, Nloc)@u_i \imp\\
&~~~~\exists u. (u<u_i) \land \pred{Start}(i, \runmodule~\service~\Nloc~\mathtt{INIT})@u\\

\end{array}\]
\caption{Definitions and Model-specific axioms about Secrecy}
\label{fig:memoir3:axioms}
\end{figure}

We now show that after initialization, if any thread $j$ has the key corresponding to
the service, then that thread must have read it from $\Nloc$ or that the thread $j$
is the initialization thread itself.

\begin{equation}
\begin{array}{@{}l}
\forall i, u_i, state, \skey, Nloc\\
~~~~\sinit (i, skey, service, state, Nloc)@u_i \imp\\
~~~~~~~\forall j, u>u_i . \pred{Has}(j, skey)@u \imp (j = i) \lor\\
~~~~~~~~~~\exists u',m.( u_i< u' < u) \land \pred{ReadNV}(j, Nloc, m)@u'\\
~~~~~~~~~~\land \pred{Contains} (m, skey)\\
\end{array}
\label{eq:memoir:keysecrecy}
\end{equation}

\noindent Fix $I_i$, $u_i$, $skey$, $service$, $Nloc$.
\\Assume $\sinit(I_i, skey, service, state, Nloc)@u_i$

We prove \ref{eq:memoir:keysecrecy} by another rely-guarantee proof, very similar to the proof of Kerberos in ~\cite{garg10:ls2}. We choose the following $\varphi$, $\psi$ and $\iota$.

\[\small
\begin{array}{@{}l}
\varphi(u) = \pred{Private}(skey, Nloc, u)\\
\psi(i, u) = \pred{KeepsPrivate}(i, skey, Nloc)@u\\
\iota(i) = \pred{LLThread}(i, runmodule)\\
\end{array}\]

To show condition (1): $\varphi(u_i)$ we can first show using (Init), (HON) and reasoning about ordering and atomicity of events that:

\begin{equation}
\begin{array}{@{}l}
\exists u_1, u_2, u_3, u_4. (u_1<u_2<u_3<u_4<u_i)\\
~~~~~~~~\pred{VerifyPCR}(\llpcr, \shash)@u_1\\
~~~~~~~~\pred{New}(skey)@u_2 \land\\
~~~~~~~~\pred{SetNVPerms}(I_i, Nloc, \llpcr)@u_3 \land\\
~~~~~~~~\pred{NVWrite}(I_i, Nloc, (\skey, h))@u_4 \land\\
~~~~~~~~\neg \pred{SetNVPerms}(I_i, Nloc, p) \circ (u_3, u_i] \land\\
~~~~~~~~\neg (\pred{Extend}(I_i, \llpcr, t) \lor \pred{Reset}(I_i, \llpcr)) \circ (u_1, u_3]\\
~~~~~~~~\neg \pred{Send}(I_i, m) \circ (u_1, u_i]
\end{array}
\label{eq:memoir3:d1}
\end{equation}

Now we can show using (Shared), (PrivateInit), (LLAct) and (\ref{eq:memoir3:d1}) that $\pred{Private}(skey, Nloc, u_i)$ holds. Essentially, at $u_i$, $s$ is still private
because, the thread $I_i$ did not leak the key, and no other thread was
running in parallel.

To prove condition (2) we again use the \rulename{Honest} rule. However, the property required
is not derived using \rulename{Confine}. The key step is to show that if the service, which
is untrusted code, is not given the key as an input, then it cannot leak the key during execution.
We do this by assuming that the original service that Memoir was initialized with had this property
and then prove that the service passed into any session of Memoir has to be equal to the
service was initialized with. This is where we require the \rulename{Eq} rule to be used.
\cut{
\begin{equation}
\begin{array}{@{}l}
\vdash
~~~~~~~~\Pi (s:\texp, l:\tptr, snap:\tmsg).~\tcmp(u_b, u_e, i.~\\
~~~~~~~~~~~~\forall u \in (u_1,u_2] \forall u' < u.\pred{Private}(skey, Nloc,u') \imp
~~~~~~~~~~~~\pred{KeepsPrivate}(i, skey, Nloc)@u)
\end{array}
\end{equation}}

The key step here is the typing of the execution of $s$
\[(\mi{s}~\mi{ExtendPCR}~\mi{ResetPCR}~\cdots)~(\mathtt{EXEC}(service\_state, req))\]
Here, we use the \rulename{Eq} rule
As we can show that $s = service$, by comparing the hash chains in PCR 17, we assign $s$ the following type:

\begin{equation}
\begin{array}{@{}l}
  (s~\mi{ExtendPCR}~\mi{ResetPCR}\cdots) : \Pi i:\tmsg.~\tcmp(
  u_b, u_e, i.~\\
  ~~~~~~~(x:\tmsg. \neg \pred{Contains}(i, s) \imp \neg \pred{Contains}(x,s),\\
  ~~~~~~~~\pred{KeepsSecret}(i, skey, Nloc)\jon[u_b, u_e])) \\
\end{array}
\end{equation}

Condition (3) Follows from (Shared), and (\ref{eq:memoir:nvramread}).

\subsubsection{State to History Summary Correspondence.}

We state without proof an invariant that the history summary has a one-to-one correspondence with the state. This is proved through an induction on the history summary.

\begin{equation}
\begin{array}{@{}l}
~~\forall i, u_i, state, \skey, \Nloc\\
~~~~~~\sinit (i, \skey, state, \Nloc, ....)@u_i \land\imp\\
~~~~~~~~~\forall h,state, state', j, j' u, u'. u>u_i \land u'>u_i \imp\\
~~~~~~~~~~~~\pred{mac}(j, \skey, (state, h))@u \land \pred{mac}(j', \skey, (state', h))@u' \imp\\
~~~~~~~~~~~~~~(state = state')\\
\end{array}
\label{eq:lemma4}
\end{equation}

\subsubsection{State Continuity}

The property we prove about Memoir is as follows:

\begin{equation}
\begin{array}{@{}l}
~~\forall u_i, \state, \newstate, \skey, i_{init}, s_{init}\\
~~~~~~\sinit (i_{init}, \skey, \service, s_{init})@u_i \imp\\
~~~~~~~~~~\forall u>u_i.~\stry(i, \skey, \state)@u \imp\\
~~~~~~~~~~~  \exists j, u' < u.~((\exists s. \sexec(j, \skey, s, \state)@u'\\
~~~~~~~~~~~~~~~~~~~~~~~~~~~~~~~~~~~~\lor \stry(j, \skey, \state)@u'  \\
~~~~~~~~~~~~~~~~~~~~~~~~~~~~~~~~~~~~\lor \sinit(j, \skey, \state)@u')\\
~~~~~~~~~~~~~~~~~~~~~~~~~~~\land (\forall j'.~ \neg \sexec(j', skey, \cdots)\jon(u',u]))
\end{array}
\end{equation}

In the above statement, we elide unnecessary arguments in the flag predicates.
This property states that for every execution attempt of the service with state
$state$ at time $u$, there exists a prior time point $u'$ such that at $u'$
either (1) service was invoked resulting in state $state$, or (2) there was an
execution attempt of the service with state $state'$ or (3) the service was
initialized with state $state$. Additionally, since  $u'$, the service has not
been invoked, which would have advanced the state of the service. This last
clause rules out any rollback attacks.  Each flag is indexed with the same
secret key $skey$ that the service was initialized with.  This key ties all the
flags in the property to to the same instance of Memoir.

\noindent
   Fix an $i$, $u_i$, $state$, $\skey$,

\noindent
   Assume $\sinit (i_{init}, \skey, \service, s_{init})@u_i$

\noindent
   For some $u > u_i$ assume that

\begin{equation}
\begin{array}{@{}l}
\stry(i, \skey, \state, \newstate)@u.
\end{array}
\label{eq:fin0}
\end{equation}

   Therefore we have $\pred{Has}(i, \skey)@u$.
   By (\ref{eq:memoir:keysecrecy}) we have that one of the two hold

\begin{equation}
\begin{array}{@{}l}
  i = i_{init} \lor\\
   \exists u'. u_i< u' < u. \pred{ReadNV}(i, \Nloc, m)@u' \land \pred{Contains} (m, \skey)
\end{array}
\end{equation}

We analyze each case:
\begin{itemize}
  \item   \textbf{Case} $i = i_{init}$:

    We have from (Init) and $\sinit(i_{init}, \skey, service, s_{init})$ that
   \[\exists u. (u<u_i) \land \pred{Start}(i, \runmodule~\service~\Nloc~\mathtt{INIT})@u\]

   With \rulename{Honest}, we can show that $\stry$ does not occur on $i$ and we have a contradiction.

\item   \textbf{Case} $\exists u'\in (u_i, u).\pred{ReadNV}(i, \Nloc, m)@u' \land \pred{Contains}(m, \skey)$:
   In this case, by (\ref{eq:memoir:nvramread})
   We have that $\pred{LLThread}(j, \runmodule)$
  Therefore, by \rulename{Honest}

\begin{equation}
\begin{array}{@{}l}
\pred{ReadNV}(i, \Nloc, (\skey, h))@u'
\end{array}
\label{eq:fin1}
\end{equation}

   By (NVRAMRead), we have that $\exists u'' < u$ such that

\begin{equation}
\begin{array}{@{}l}
   \pred{WriteNV}(j, \Nloc, (\skey, h))@u'' \land\\  \forall j''.~\neg \pred{WriteNV}(j''. \Nloc, m') \jon (u'', u']
\end{array}
\label{eq:fin2}
\end{equation}

Again, by (\ref{eq:memoir:nvramread}) and (\ref{eq:fin2}) we have that
\begin{equation}
\begin{array}{@{}l}
   \pred{LLThread}(j, \runmodule)
\end{array}
\end{equation}

%\begin{equation}
%\begin{array}{@{}l}
%   \pred{WriteNV}(j', \Nloc, (skey, h))@u''
%\end{array}
%\label{eq:fin3}
%\end{equation}

And by \rulename{Honest}, as we know (\ref{eq:fin2}), we can derive that
\begin{equation}
\begin{array}{@{}l}
  \pred{mac}(j, skey, (ENC_{skey}(state', h))
\end{array}
\label{eq:fin3a}
\end{equation}

Also, from \ref{eq:fin0} and \rulename{Honest}  we know that the branch at Line 12 of $\runmodule$ executed. This gives us two cases:
\begin{itemize}
\item
      \textbf{Case} 1:
\begin{equation}
\begin{array}{@{}l}
 \pred{verifyMAC}(i, \skey, (ENC_{skey}(state), h))
\end{array}
\label{eq:fin4}
\end{equation}
This is the case where the history summary $h$ matches the MACed history summary.
      From \ref{eq:fin4} and (MAC), we have for some $j'$
\begin{equation}
\begin{array}{@{}l}
 \pred{mac}(j', \skey, (ENC_{skey}(state), h))
\end{array}
\label{eq:fin5}
\end{equation}
      By (\ref{eq:lemma4}) along with (\ref{eq:fin4}) and (\ref{eq:fin5}), we have $state' = state$
      We then have from (\ref{eq:fin3a}) that there exists a $u'$ such that
\begin{equation}
\begin{array}{@{}l}
     \sexec(j, \skey, s', state)@u' \\\lor \sinit(j, \skey, service, state)@u'
\end{array}
\label{eq:fin5}
\end{equation}
Also, from (\ref{eq:fin2}), we can show that $\forall j''.~\neg\sexec(j'', \cdots)$.

\item \textbf{Case} 2:
\begin{equation}
\begin{array}{@{}l}
\pred{verifyMAC}(j, \skey, (ENC_{\skey}(state), h') \land h = H(req||h')
\end{array}
\label{eq:fin6}
\end{equation}
This is the case where at Line 12 of runmodule, the current history summary
is the hash of the current request and the history summary in the snapshot.
This means that Memoir was called with exactly the same request in the past
and no other request has completed since then.
This case proceeds similarly to Case 1.
\end{itemize}
\end{itemize}

\end{document}